\documentclass[11pt]{article}
\usepackage{microtype}
\usepackage{graphicx,wrapfig}
\usepackage{subcaption}
\usepackage{amsmath,amsbsy,amsfonts,amssymb,amsthm,bm}
\usepackage{mathtools}
\usepackage{color,cases,multirow}
\usepackage[margin=1in]{geometry} 
\usepackage[colorlinks=True, citecolor=blue]{hyperref}
\usepackage{float}
\usepackage{array}
\usepackage{booktabs}
\usepackage[ruled, linesnumbered]{algorithm2e}

\usepackage{authblk}

\newtheorem{theorem}{Theorem}
\newtheorem{remark}{Remark}

\newtheorem{lemma}{Lemma}
\newtheorem{assumption}{Assumption}
\newtheorem{corollary}{Corollary}
\theoremstyle{definition}

\DeclareMathOperator{\tr}{Tr}

\newcommand{\RR}{\mathbb{R}}

\newcommand{\EE}{\mathbb{E}}

\newcommand{\cP}{\mathcal{P}}

\newcommand{\cX}{\mathcal{X}}
\newcommand{\cS}{\mathcal{S}}

\newcommand{\hc}{\hat{c}}
\newcommand{\hs}{\hat{s}}

\newcommand{\bx}{\mathbf{x}}

\newcommand{\bz}{\bm{z}}

\newcommand{\bb}{\bm{b}}
\newcommand{\bc}{\bm{c}}

\newcommand{\by}{\bm{y}}

\newcommand{\hot}{\textrm{hot}}

\newcommand{\xobs}{\{\bx_i\}}

\newcommand{\bi}{\begin{enumerate}}
\newcommand{\ei}{\end{enumerate}}

\title{Correcting Convexity Bias in Function and Functional Estimate}

\author[1]{Chao Ma}
\author[1]{Lexing Ying}

\affil[1]{Department of Mathematics, Stanford University}

\date{\today}

\begin{document}

\maketitle 

\begin{abstract}
A general framework with a series of different methods is proposed to improve the estimate of convex function (or functional) values when only noisy observations of the true input are available. Technically, our methods catch the bias introduced by the convexity and remove this bias from a baseline estimate. Theoretical analysis are conducted to show that the proposed methods can strictly reduce the expected estimate error under mild conditions. When applied, the methods require no specific knowledge about the problem except the convexity and the evaluation of the function. Therefore, they can serve as off-the-shelf tools to obtain good estimate for a wide range of problems, including optimization problems with random objective functions or constraints, and functionals of probability distributions such as the entropy and the Wasserstein distance. Numerical experiments on a wide variety of problems show that our methods can significantly improve the quality of the estimate compared with the baseline method. 
\end{abstract}

\section{Introduction}
In this paper, we study the problem of estimating function/functional values when uncertainty exists on the input value. Specifically, let $F$ be a function or functional, and $\Omega$ be its input domain. We want to estimate $F(\bx^*)$ for some $\bx^*\in\Omega$, while only having access to a set of observations sampled from a probability distribution $\mu$ on $\Omega$ with $\EE_{\mu}\bx=\bx^*$. Let $\bx_1, \bx_2, ..., \bx_n$  be the observations. By the law of large numbers, the average of the observations, $\bar\bx:=(\bx_1+\cdots+\bx_n)/n$, is close to $\bx$ when $n$ is large. Therefore, a straightforward estimate of $F(\bx^*)$ is to use $F(\bar\bx)$. However, since $F$ is not necessarily linear, biases may exist when $F(\bar\bx)$ is used as an estimate. For example, when $F$ is convex, the Jensen's inequality 
\begin{equation*}
    F(\bx^*) = F(\EE\bar\bx) \leq \EE F(\bar\bx)
\end{equation*}
reveals positive bias when $F(\bar\bx)$ is used to estimate $F(\bx^*)$. Taking into consideration the special properties of $F$ (such as convexity/concavity), it is possible to find estimate methods that have less bias and outperform the straightforward estimate using the sample average.

The function/functional value estimate problem appears in many applications, such as the estimate of functionals of probability distributions (e.g. expectation, entropy, distance etc) or the estimate of optimal value of optimization problems with noisy constraints. Some of these problems have been under investigation for a long time, with many methods proposed. For instance, the James-Stein estimator of the expectation of Gaussian distributions~\cite{stein1956variate,james1992estimation}, the estimate of entropy~\cite{miller1955note,zahl1977jackknifing,paninski2003estimation,grassberger2003entropy,vu2007coverage,valiant2013estimating,han2020optimal}, mutual information~\cite{paninski2003estimation,kraskov2004estimating,gao2015efficient,marek2008estimation}, other functionals of probability distributions~\cite{wolpert1995estimating,jiao2015minimax,acharya2016estimating,han2020minimax}, etc. Besides, the estimate of the minimizers of quadratic functions has also been studied in recent works~\cite{etter2020operator,etter2021operator}. These works usually study a specific class of problems and propose methods that perform well on these problems. Besides methods, related theoretical analysis are also rich in the literature~\cite{antos2001convergence,jiao2015minimax,wu2016minimax,bu2018estimation}. The works mentioned here are not meant to be comprehensive. For more detailed discussion of the literature, readers can refer to the reviews in~\cite{verdu2019empirical,paninski2003estimation,jiao2015minimax}

In this work, instead, we propose a general framework that can improve the estimate for all convex or concave functions/functionals. Our methods produce improved estimate for $F(\bx^*)$ by estimating the bias introduced by convexity/concavity using noisy observations and removing that bias from the naive estimate $F(\bar\bx)$. Concretely, we either {\bf shift} $F(\bar\bx)$ by an appropriate amount $c$ and use $F(\bar\bx)+c$ as a new estimate, or {\bf scale} $F(\bar\bx)$ by a factor $s$ and use $sF(\bar\bx)$ as a new estimate. These ``debiasing quantities'' $c$ and $s$ are derived to minimize the square error of the estimate, i.e.
\begin{equation*}
\EE\big(F(\bar\bx)+c-F(\bx^*)\big)^2\ \ \textrm{or}\ \ \EE\big(sF(\bar\bx)-F(\bx^*)\big)^2.
\end{equation*}
The so derived debiasing quantities $c$ and $s$ are in population sense, i.e. they inevitably depend on unknown quantities such as $\bx^*$ and $\EE F(\bar\bx)$. With only the observations $\bx_1,\cdots\bx_n$ in hand, we use bootstrap to approximate these population debiasing quantities and obtain ``empirical'' debiasing quantities $\hat{c}$ and $\hat{s}$ (See Section~\ref{sec:general} for details). Theoretical analysis are conducted to show that the $\hat{c}$ and $\hat{s}$ obtained using bootstrap appropriately can indeed reduce the square error of the estimate. In the following is an informal statement of our main theorems~\ref{thm:main_shifting} and~\ref{thm:main_scaling}.

\begin{theorem}{(Informal main theorem)}
Let $F$ be a convex function on $\Omega$ and $\mu$ be a probability distribution on $\Omega$ with $\bx^*=\EE_\mu \bx$. Let $\hat{c}$ and $\hat{s}$ be the shifting and scaling debiasing quantities computed using samples $\bx_1,\cdots\bx_n$ drawn from $\mu$. Under some assumptions on $F$, $\mu$, and the bootstrap method, when $n$ is sufficiently large, we have 
\begin{equation*}
\EE\big(F(\bar\bx)+\hat{c}-F(\bx^*)\big)^2<\EE\big(F(\bar\bx)-F(\bx^*)\big)^2,
\end{equation*}
and
\begin{equation*}
\EE\big(\hat{s}F(\bar\bx)-F(\bx^*)\big)^2<\EE\big(F(\bar\bx)-F(\bx^*)\big)^2.
\end{equation*}
\end{theorem}
Besides the bootstrap method, other method can be employed to estimate the debiasing quantities when more information of $F$ is available. One ``covariance estimate'' method that makes use of $F$'s second-order curvature is also introduced in Section~\ref{sec:general}. % and tested in Section~\ref{sec:exp}. 

Finally, extensive numerical experiments are conducted on a series of problems, from simple convex functions, to optimization problems with stochastic objective function and constrains, and then to the estimate of entropy and Wasserstein distance of probability distributions. The same framework is applied to all the problems, with the only difference being the methods used to solve the problems themselves. In the experiments, our methods can reduce both the expected bias and the expected square error by a significant amount compared with the naive estimate. This shows that the proposed framework can serve as a handy and convenient method to improve the estimate of functions/functionals without requiring specific domain knowledge except convexity/concavity.

\section{Problem settings and methods}\label{sec:general}
\subsection{Problem settings}
Without loss of generality, we introduce our method using convex functions on Euclidean spaces. 
Let $F: \Omega\rightarrow\RR$ be a convex function with $\Omega\in\RR^d$. Let $\mu$ be a probability distribution on $\RR^d$ that satisfies $\EE_{\mu}\bx=\bx^*$. Given noisy observations $\bx_1, \bx_2, ..., \bx_n$ sampled i.i.d from $\mu$, our goal is to estimate $F(\bx^*)$ using only the observations. 
As mentioned in the introduction, a naive estimate of $F(\bx^*)$ is to use the function value at the sample average of the observations, $F(\bar{\bx})$, where $\bar{\bx}:=\frac{1}{n}\sum_{i=1}^n \bx_i$. 
When the function $F$ is continuous at $\bx^*$, this estimate is consistent. However, bias exists due to convexity. By the Jensen's inequality we have $F(\bx^*)\leq \EE F(\bar\bx)$, where the expectation is taken over the sampling of $\bx_1,\cdots,\bx_n$. This bias may be large if the curvature of $F$ is big or the number of available observations $n$ is small. 
To reduce the convexity bias, we design methods to correct the naive estimate $F(\bar\bx)$. We explore two possible ways of changing $F(\bar\bx)$: {\bf the shifting method} and {\bf the scaling method}. The shifting method takes an additive approach. We estimate a debiasing quantity $c$, add it to $F(\bar\bx)$, and use $F(\bar\bx)+c$ as an improved estimate for $F(\bx^*)$. On the other hand, the scaling method takes a multiplicative approach. We estimate another debiasing quantity $s$, multiply it to $F(\bar\bx)$, and use $sF(\bar\bx)$ as an improved estimate for $F(\bx^*)$. In the following subsections we introduce the two methods in detail.

\subsection{The shifting method}
In the shifting method, we find an additive debiasing quantity $c$ such that $F(\bar\bx)+c$ is a better estimate for $F(\bx^*)$ than $F(\bar\bx)$. We measure the quality of the estimate using the squared error
\begin{equation}\label{eqn:shifting_c_def_1}
    \EE \left( F(\bar{\bx})+c - F(\bx^*) \right)^2,
\end{equation}
and find $c$ to minimize~\eqref{eqn:shifting_c_def_1}. Treating $c$ as a constant and expanding~\eqref{eqn:shifting_c_def_1}, we obtain a quadratic function of $c$,
\begin{equation*}
c^2 + 2c\EE(F(\bar\bx)-F(\bx^*)) + \EE(F(\bar\bx)-F(\bx^*))^2,
\end{equation*}
whose minimum is achieved at 
\begin{equation}\label{eqn:shifting_c_def_2}
c = F(\bx^*) - \EE F(\bar\bx). 
\end{equation}
Hence, ideally we can use the $c$ defined in~\eqref{eqn:shifting_c_def_2} as the debiasing quantity. 

In practise, however, $\bx^*$ and $\EE F(\bar\bx)$ are unknown, and can only be approximated using noisy observation $\bx_1,\cdots,\bx_n$. For general $F$ and $\mu$, we can use bootstrap~\cite{efron1992bootstrap} to estimate the $c$ in~\eqref{eqn:shifting_c_def_2}. bootstrap uses a uniform distribution on $\bx_1,\cdots,\bx_n$ to approximate $\mu$. Denote $\cX=\{\bx_1,...,\bx_n\}$ and let $\mu_{\cX}$ be the uniform distribution on $\cX$. For $k=1,2,...,K$, define
\begin{equation}\label{eqn:bootstrap_x}
    \tilde{\bx}_k = \frac{1}{n}\sum\limits_{i=1}^n \bx_{k,i},
\end{equation}
where $\{\bx_{k,i}\}_{i=1}^n$ are i.i.d. sampled from $\mu_{\cX}$. Then, the empirical distribution of $\{\tilde{\bx}_k\}$ is an estimate of the distribution of $\bar{\bx}$. Hence, the shifting debiasing quantity $c$ can be approximated by
\begin{equation*}
    c=F(\bx^*)-\EE F(\bar{\bx}) \approx F\left(\frac{1}{K}\sum\limits_{k=1}^K \tilde{\bx}_k\right) - \frac{1}{K}\sum\limits_{k=1}^K F(\tilde{\bx}_k).
\end{equation*}
Since $\EE_{\mu_{\cX}}\tilde{\bx}_k=\bar\bx$, we have $F(\bar{\bx})\approx F\big(\frac{1}{K}\sum\limits_{k=1}^K \tilde{\bx}_k\big)$ when $K$ is large. Therefore, we can use
\begin{equation}\label{eqn:shifting_c_def_3}
    \hat{c} := F(\bar{\bx}) - \frac{1}{K}\sum\limits_{k=1}^K F(\tilde{\bx}_k),
\end{equation}
as an approximation of $c$. With this definition of $\hat{c}$, the steps of the shifting method with bootstrap is listed in Algorithm~\ref{alg:shift_bs}.

\begin{algorithm}[H]
\caption{The Shifting Debiasing Method with bootstrap}\label{alg:shift_bs}
\SetKwInOut{Input}{Input}
\SetKwInOut{Output}{Output}

\Input{A convex function $F:\RR^d\rightarrow\RR$, observations $\bx_1,\cdots,\bx_n\in\RR^d$, $K$}
\Output{Debiased estimate $F(\bar\bx)+\hat{c}$}

\Indp
$\bar\bx \gets \frac{1}{n}\sum\limits_{i=1}^n \bx_i$\;

\For{$k\gets1$ \KwTo $K$}{
    Sample $\tilde{\bx}_{k,1},\cdots,\tilde{\bx}_{k,n}$ i.i.d. from the uniform distribution on $\bx_1,\cdots,\bx_n$\;
    
    $\tilde{\bx}_k \gets \frac{1}{n}\sum\limits_{i=1}^n \tilde{\bx}_{k,i}$
}
$\hat{c} \gets F(\bar\bx) - \frac{1}{K}\sum\limits_{k=1}^K F(\tilde{\bx}_k)$\;
\end{algorithm}

\paragraph{The covariance estimate method}
bootstrap is not the only approach to estimate the debiasing quantity $c$, especially when more knowledge about the problem is available.
For example, when the distribution $\mu$ is concentrated in a region where $F$ is close to a quadratic function, we can estimate $c$ by estimating the covariance of $\bx$. To see this, assume $F$ is close to its second Taylor polynomial at $\bx^*$, then for $F(\bar\bx)$ we have
\begin{equation}\label{eqn:second_expansion}
    F(\bar\bx) \approx F(\bx^*) + \nabla F(\bx^*)^T(\bar\bx-\bx^*) + \frac{1}{2}(\bar\bx-\bx^*)^T\nabla^2F(\bx^*)(\bar\bx-\bx^*).
\end{equation}
Taking expectation, note that $\EE(\bar{\bx}-\bx^*) = 0$, we have
\begin{equation*}
    \EE\big(F(\bar{\bx}) - F(\bx^*)\big) \approx \frac{1}{2}\EE(\bar{\bx}-\bx^*)^T \nabla^2 F(\bx^*) (\bar{\bx}-\bx^*) = \frac{1}{2}\tr(\bar{C}\nabla^2F(\bx^*)),
\end{equation*}
where $\bar{C}$ is the covariance matrix of $\bar{\bx}$. Let $C$ be the covariance of $\bx$, then we have $\bar{C}=C/n$ and
\begin{equation*}
    \EE\big(F(\bar{\bx}) - F(\bx^*)\big) \approx \frac{1}{2n}\tr(C\nabla^2F(\bx^*)).
\end{equation*}
Therefore, once we have some knowledge on $\nabla^2F(\bx^*)$, e.g. having access to a matrix $H\approx \nabla^2F(\bx^*)$, we can obtain an estimate of $c$ without using bootstrap:
\begin{align}
\hat{c} &= -\frac{1}{2n}\tr\left(\frac{1}{n-1}\sum\limits_{i=1}^n(\bx_i-\bar\bx)(\bx_i-\bar\bx)^TH\right) \nonumber\\
 &= -\frac{1}{2n(n-1)}\sum\limits_{i=1}^n(\bx_i-\bar\bx)^TH(\bx_i-\bar\bx).\label{eqn:cov_est}
\end{align}
An algorithm using the covariance estimate method takes similar inputs and outputs as Algorithm~\ref{alg:shift_bs}, 
replacing the bootstrap steps in line 2-6 by computing the $\hat{c}$ in~\eqref{eqn:cov_est}. We skip the step-by-step algorithm here.

\subsection{The scaling method}
The scaling method can be applied when $F$ is always positive (or negative) for any $\bx$. In this case, we find a multiplicative debiasing quantity $s$ such that $sF(\bar\bx)$ becomes a good estimate for $F(\bx^*)$. Still consider the squared error, which becomes
\begin{equation*}
    \EE \left(sF(\bar{\bx})- F(\EE\bx)\right)^2.
\end{equation*}
Treating $s$ as a constant and minimizing the squared error as a quadratic function of $s$ gives the following ideal choice of $s$:
\begin{equation}\label{eqn:scaling_s_def_1}
    s = \frac{F(\bx^*)\EE F(\bar{\bx})}{\EE F^2(\bar\bx)}.
\end{equation}

In practice, we can still use bootstrap to obtain an estimate of the $s$ in~\eqref{eqn:scaling_s_def_1}. Recall the choice of $\{\tilde{\bx}_k\}_{k=1}^K$ in~\eqref{eqn:bootstrap_x}. Using the uniform distribution on $\{\tilde{\bx}_k\}_{k=1}^K$ as an approximation of $\mu$, and $\bar\bx$ as an approximation of $\bx^*$, we can take
\begin{equation}\label{eqn:scaling_s_def_2}
    \hat{s} := \frac{F(\bar{\bx})\frac{1}{K}\sum\limits_{k=1}^K F(\tilde{\bx}_k)}{\frac{1}{K}\sum\limits_{k=1}^K F^2(\tilde{\bx}_k)} = \frac{F(\bar{\bx})\sum_{k=1}^K F(\tilde{\bx}_k)}{\sum_{k=1}^K F^2(\tilde{\bx}_k)}.
\end{equation}
as an approximation of $s$. An algorithm using the scaling method is given in Algorithm~\ref{alg:scale_bs}. \\

\begin{algorithm}[H]
\caption{The Scaling Debiasing Method with bootstrap}\label{alg:scale_bs}
\SetKwInOut{Input}{Input}
\SetKwInOut{Output}{Output}

\Input{A convex  and positive function $F:\RR^d\rightarrow\RR$, observations $\bx_1,\cdots,\bx_n\in\RR^d$, $K$}
\Output{Debiased estimate $\hat{s}F(\bar\bx)$}
\Indp
$\bar\bx \gets \frac{1}{n}\sum\limits_{i=1}^n \bx_i$\;

\For{$k\gets1$ \KwTo $K$}{
    Sample $\tilde{\bx}_{k,1},\cdots,\tilde{\bx}_{k,n}$ i.i.d. from the uniform distribution on $\bx_1,\cdots,\bx_n$\;
    
    $\tilde{\bx}_k \gets \frac{1}{n}\sum\limits_{i=1}^n \tilde{\bx}_{k,i}$
}
$\hat{s} \gets \frac{F(\bar{\bx})\sum_{k=1}^K F(\tilde{\bx}_k)}{\sum_{k=1}^K F^2(\tilde{\bx}_k)}$\;
\end{algorithm}

\subsection{Dealing with functionals}
The methods above can be easily applied to the cases where $F$ is a convex/concave functional. 
Taking a functional of probability distribution, $F[p]$, as an example. Suppose $p\in\cP(\Omega)$ is a distribution on $\Omega\in\RR^d$ and $F$ is convex with respect to $p$. Such functional can be the entropy of $p$, or the distance from $p$ to some other probability distributions (e.g. the KL divergence, the Wasserstein distance). Let $\bx_1, ..., \bx_n$ be points in $\Omega$ i.i.d. sampled from $p$. Then, the Dirac delta distributions $\delta_{\bx_1}, ..., \delta_{\bx_n}$ can be understood as noisy observations of $p$. In this case, the naive estimate of $F[p]$ is $F[\bar{p}]$, where $\bar{p}$ is the empirical distribution
$\bar{p} = \frac{1}{n}\sum_{i=1}^n \delta_{\bx_i}$. To estimate the debiasing quantities with bootstrap, for $k=1,2,...,K$ and $i=,2,...,n$, let $\bx_{k,i}$ be i.i.d. samples taken uniformly from $\{\bx_1,...,\bx_n\}$, and let $\tilde{p}_k = \frac{1}{n}\sum_{i=1}^n \delta_{\bx_{k,i}}.$
Then, for the shifting method we can take
\begin{equation}
    \hat{c} = F[\bar{p}] - \frac{1}{K}\sum\limits_{k=1}^K F[\tilde{p}_k],
\end{equation}
and for the scaling method we can take
\begin{equation}
    \hat{s} = \frac{F[\bar{p}]\sum_{k=1}^K F[\tilde{p}_k]}{\sum_{k=1}^K F^2[\tilde{p}_k]}.
\end{equation}

\section{Theoretical results}\label{sec:thm}
In this section, we provide theoretical results which show that our debiasing methods can reduce the expected squared error of the estimate. In the statement of the theorems, without loss of generality, we consider $\bx\in\RR^n$ and $F$ as a convex function of $\bx$. In the theorems and the proofs, we use Einstein notations to represent tensor contractions involving tensors with rank $\geq3$. For some examples, if $A, B\in\RR^{d\times d}$ are matrices, $A_{ab}B^{bc}$ represents the matrix product $AB$; if $A\in\RR^{d\times d\times d\times d}$ and $B\in\RR^{d\times d}$, $A_{abcd}B^{cd}$ gives a $d\times d$ matrix by contracting the last two dimensions of $A$ with the two dimensions of $B$; if $F:\RR^d\rightarrow\RR$ is a function and $\bx=(x^1,\cdots,x^d), \by=(y^1,\cdots,y^d)\in\RR^d$ are vectors, $\partial_{abc}^3F(\bx)y^ay^by^c$ means the sum
\begin{equation*}
\sum\limits_{i,j,k=1}^d \frac{\partial^3F(\bx)}{\partial x^i\partial x^j\partial x^k}y^iy^jy^k.
\end{equation*}

\paragraph{The shifting method.}
We first make some assumptions on the input distribution $\mu$ and the function $F$. In the following, the norm $\|\cdot\|$ is by default the $\ell_2$ norm for vectors and matrices. 

\begin{assumption}\label{assump:mu}
The probability distribution $\mu$ has up to 8-th finite moments. 
\end{assumption}

\begin{assumption}\label{assump:F}
{ $F$ has finite fourth-order derivatives.} 
\end{assumption}

Under the assumptions above, the following main result shows that the shifting method can strictly reduce the expected squared error, as long as $K$ is at least at the same order of $n$. 

\begin{theorem}\label{thm:main_shifting}
Suppose Assumption~\ref{assump:mu} and~\ref{assump:F} hold. Consider the shifting debiasing method using the debiasing quantity $\hat{c}$ defined in~\eqref{eqn:shifting_c_def_3}.
Denote
$M_2=\EE_{\mu} (\bx-\bx^*)(\bx-\bx^*)^T\in\RR^{d\times d}$, { $M_3=\EE_{\mu} (\bx-\bx^*)^{\otimes3}\in\RR^{d\times d\times d}$} be the second and third centered moment tensors of $\mu$. Define
\begin{align*}
\sigma_1 &=\partial_aF(\bx^*)\partial_bF(\bx^*)(M_2)^{ab} = \nabla F(\bx^*)^T M_2\nabla F(\bx^*),\\ \sigma_2&=\partial^2_{ab}F(\bx^*)(M_2)^{ab} = \tr(M_2\nabla^2 F(\bx^*)), \\ 
\sigma_3&= \partial_a F(\bx^*)\partial^2_{bc}F(\bx^*) (M_3)^{abc},\\  
\sigma_4&= \partial_a F(\bx^*)(M_2)^{ab}\partial^3_{bcd} F(\bx^*)(M_2)^{cd}.
\end{align*}
Then, if $n$ is sufficiently large, $K\geq C_K n$ for some constant $C_K$, and
\begin{equation}\label{eqn:thm_cond}
\frac{\sigma_2^2}{4}+\sigma_3+\sigma_4 -\frac{\sigma_1}{C_K} > 0,
\end{equation}
we have
\begin{equation}
    \EE (F(\bar\bx) + \hat{c} - F(\bx^*))^2 < \EE (F(\bar{\bx})-F(\bx^*))^2,
\end{equation}
where the expectation is taken on both the sampling of $\xobs_{i=1}^n$  and the bootstrap.
\end{theorem}

{ Actually, the assumptions for the theorem above do not require $F$ to be convex or concave. The debiasing method is effective as long as the quantity $\sigma_2^2$ is larger than the other three terms in the condition~\ref{eqn:thm_cond}. Roughly speaking, this requires the second derivative of $F$ is large compared with its first and third derivatives. Though, the condition is easier to hold when $F$ is convex or concave, otherwise $\sigma_2$ might be very small.}

%{\red limiting case: $C_K$ go to infinity}
{ The condition~\eqref{eqn:thm_cond} can be simplified if we problems with certain properties. for example, if $\mu$ is symmetric with respect to $\bx^*$, then we have $M_3=0$, and hence $\sigma_3=0$. If we further assume that $F$ is quadratic, we have the following corollary derived from Theorem~\ref{thm:main_shifting}:
\begin{corollary}\label{coro:quad}
Let $F(\bx)=\frac{1}{2}\bx^TA\bx+b$ be a quadratic function with $A\in\RR^{d\times d}$ being a positive definite matrix and $b\in\RR$ being a scalar. Let $\mu$ be a probability distribution on $\RR^d$ that satisfies $\EE_{\mu}\bx=\bx^*$. Suppose $\mu$ is symmetric with respect to $\bx^*$. Then, the shifting debias method can reduce the expected squared error as long as 
\begin{equation*}
    \frac{\tr(AM_2)^2}{4} > \frac{(\bx^*)^TA^TM_2A\bx^*}{C_K}.
\end{equation*}
\end{corollary}
By this corollary, the shifting method works everywhere for quadratic functions as long as $C_K$ is large enough. This is possible because $C_K$ is a hyperparameter that we can choose freely as long as the computational resource is sufficient. The following corollary states the condition for the shifting method to work in the limit case in which $C_K$ is pushed to infinity:
\begin{corollary}\label{coro:limit}
Let $F$ satisfy the assumptions in Theorem~\ref{thm:main_shifting}, and $\mu$ be a probability distribution on $\RR^d$ that satisfies $\EE_{\mu}\bx=\bx^*$ and is symmetric with respect to $\bx^*$. Suppose $C_K$ is sufficiently large. Then, the shifting debias method can reduce the expected squared error as long as
\begin{equation}\label{eqn:cond_coro_limit}
    \sigma_2^2 > -4\sigma_4,
\end{equation}
where $\sigma_2$ and $\sigma_4$ are defined the same way as in Theorem~\ref{thm:main_shifting}.
\end{corollary}
This corollary shows more clearly that the condition on $F$ generally requires the second derivative to be large compared with the first and third derivatives. Quadratic function under the conditions in Corollary~\ref{coro:limit} always satisfy~\eqref{eqn:cond_coro_limit} because $\sigma_4$ vanishes.
}

A sketch of the proof of Theorem~\ref{thm:main_shifting} is shown in the next subsection. The full proof is provided in Section~\ref{sec:pf}.

\paragraph{The scaling method.}
Next, we study the scaling method, and show a similar results as that for the shifting method. For the ease of analysis, we make some additional assumptions on $\mu$ and $F$:
\begin{assumption}\label{assump:mu_2}
The probability distribution $\mu$ has finite moment of all orders.
\end{assumption}
\begin{assumption}\label{assump:pos}
There exists a constant $B>0$, such that $F(\bx)\geq B$ for any $\bx$.
\end{assumption}

\begin{remark}
Assumption~\ref{assump:mu_2} is made for the convenience of analysis. The ``finite moment of all orders'' can be relaxed to finite moments up to a certain finite order. The orders we need can be obtained by tracing the higher-order-terms in the proof of Theorem~\ref{thm:main_scaling} (given in Section~\ref{sec:pf2}).
\end{remark}

Under the new assumptions, we have the following theorem whose proof is given in Section~\ref{sec:pf2}.
\begin{theorem}\label{thm:main_scaling}
Suppose Assumption~\ref{assump:F}, \ref{assump:mu_2} and~\ref{assump:pos} hold. Consider the scaling debiasing method using the debiasing quantity $\hat{s}$ defined in~\eqref{eqn:scaling_s_def_2}.
Let $M_2$, $M_3$, $\sigma_1$, $\sigma_2$, $\sigma_3$, $\sigma_4$ be defined in the same as those in Theorem~\ref{thm:main_shifting}. Define
\begin{equation*}
\sigma'_3=\partial_a F(\bx^*)\partial_b F(\bx^*)\partial_b F(\bx^*)(M_3)^{abc}.
\end{equation*}
Then, if $n$ is sufficiently large, $K\geq C_K n$ for some constant $C_K$, and
\begin{equation}\label{eqn:thm_cond_scaling}
\frac{\sigma_2^2}{4}+\sigma_3+\sigma_4+\frac{2\sigma'_3}{F(\bx^*)}+\frac{4\sigma_1\sigma_2}{F(\bx^*)}-\frac{3\sigma_1^2}{F(\bx^*)^2}-\frac{\sigma_1}{C_K}>0,
\end{equation}
we have
\begin{equation}
    \EE (\hat{s}F(\bar\bx) - F(\bx^*))^2 < \EE (F(\bar{\bx})-F(\bx^*))^2,
\end{equation}
where the expectation is taken on both the sampling of $\xobs$ and the bootstrap.
\end{theorem}

\subsection{Proof sketch of Theorem~\ref{thm:main_shifting}}
In this section, we show the idea of the proof for Theorem~\ref{thm:main_shifting}. The proof for Theorem~\ref{thm:main_scaling} takes a similar approach, with more involved analysis. As an illustration, some arguments in this section might not be rigorous. Strict proofs for both Theorem~\ref{thm:main_shifting} and~\ref{thm:main_scaling} are given in later sections.

\paragraph{Notations} 
In this proof sketch, we write $h=O(g)$ if there exists a constant $C$ independent with $n$ such that $|h|< C|g|$ always holds. We use $\hot(k)$ to denote higher-order-terms containing $\bar\bx-\bx^*$ or (and) $\tilde{\bx}-\bar\bx$ with total orders $\geq k$. Given observations $\bx_1,...,\bx_n$, we use $\tilde{\bx}$ to denote the random variable given by the average of $n$ uniformly drawn samples from $\bx_1,...,\bx_n$. 
When taking expectation, we use $\EE$ to denote the expectation over both the sampling of $\xobs$ and the choice of $\tilde\bx$ during bootstrap, and use $\EE_{\tilde{\bx}}$ to denote the expectation over $\{\tilde{\bx}\}$ based on a fixed set of $\bx_1, ..., \bx_n$. Let $H(\bx):=\nabla^2 F(\bx)$. For any matrix $A\in\RR^{d\times d}$ and vector $\bx\in\RR^d$, denote $\|\bx\|^2_A = \bx^TA\bx$. {We note that this is not a norm when $A$ is not positive definite.} \\

%We use Einstein notations mostly when tensors with dimension $\geq3$ are involved. When only vectors and matrices are involved, we still use the conventional matrix product notations. For vectors, we use superscript as the index of entries, and subscript as the index of the vector in a set of vectors. Entries are denoted by lower case letters. For example, the $i$-th noisy observation vector is given by $\bx_i=(x^1_i, x^2_i, \cdots, x^d_i)\in\RR^d$. 

To prove Theorem~\ref{thm:main_shifting}, first notice that 
\begin{equation*}
E(F(\bar\bx)+\hc-F(\bx^*))^2 = \EE (F(\bar\bx) - F(\bx^*))^2 + \EE \hc^2 + 2\EE (F(\bar\bx) - F(\bx^*))\hc.
\end{equation*}
Hence, we only need to show $\EE \big(\hc^2 + 2(F(\bar\bx) - F(\bx^*))\hc\big) < 0$. Let 
\begin{equation*}
\bar{c} = \EE_{\tilde{\bx}}\hc = F(\bar\bx) - \EE_{\tilde{\bx}} F(\tilde{\bx}),\ \ \textrm{and}\ \delta=\EE_{\tilde{\bx}} F(\tilde{\bx}) - \frac{1}{K}\sum_{k=1}^K F(\tilde{\bx}_k). 
\end{equation*}
Then, we have $\hc = \bar{c} + \delta$ and $\EE_{\tilde{\bx}}\delta = 0$, and 
\begin{equation*}
\EE \big(\hc^2 + 2(F(\bar\bx) - F(\bx^*))\hc\big) = \EE\big(\bar{c}^2+2\bar{c}\delta+\delta^2 + 2(F(\bar\bx) - F(\bx^*))\bar{c}+2(F(\bar\bx) - F(\bx^*))\delta\big).
\end{equation*}
Since $\bar{c}$ and $F(\bar{\bx})-F(\bx^*)$ do not depend on the sampling of $\tilde{\bx}_k$, by the law of total expectation, we have 
\begin{equation*}
    \EE \bar{c}\delta = \EE\left(\EE_{\tilde{\bx}}\bar{c}\delta\right) = \EE\left(\bar{c}\EE_{\tilde{\bx}}\delta\right) = 0,
\end{equation*}
and similarly $\EE(F(\bar{\bx})-F(\bx^*))\delta=0$. Therefore,
\begin{equation}\label{eqn:pf_sketch_1}
\EE \big(\hc^2 + 2(F(\bar\bx) - F(\bx^*))\hc\big) = \EE\big(\bar{c}^2+\delta^2 + 2(F(\bar\bx) - F(\bx^*))\bar{c}\big).
\end{equation}
Among the three terms on the right hand side, $\EE\bar{c}^2$ and $\EE\delta^2$ are positive. As an expected debiasing quantity, $\bar{c}$ has a negative correlation with $F(\bar\bx)-F(\bx^*)$, hence $2\EE(F(\bar\bx) - F(\bx^*))\bar{c}$ is negative. 
Next, we will estimate the three terms, and show that under the condition~\ref{eqn:thm_cond} the negative term $\EE(F(\bar\bx) - F(\bx^*))\bar{c}$ has larger absolute value than the first two terms. Hence, \eqref{eqn:pf_sketch_1} is negative in total.

\subsubsection*{Estimate of $\EE\bar{c}^2$}
Taking a Taylor expansion for $F(\tilde{\bx})$ at $\bar\bx$, by Assumption~\ref{assump:F}, we have
{\small
\begin{equation*}
F(\tilde\bx)-F(\bar\bx) = \nabla F(\bar\bx)^T(\tilde\bx-\bar\bx) + \frac{1}{2}\|\tilde\bx-\bar\bx\|^2_{H(\bar\bx)} + \frac{1}{6}\partial^3_{abc} F(\bar\bx)(\tilde{x}^a-\bar{x}^a)(\tilde{x}^b-\bar{x}^b)(\tilde{x}^c-\bar{x}^c) + O(\|\tilde\bx-\bar\bx\|^4). 
\end{equation*}}
Taking expectation over $\tilde\bx$, noting that $\EE_{\tilde{\bx}}(\tilde\bx-\bar\bx)=0$, we have 
\begin{equation}\label{eqn:pf_sketch_c2_0}
-\bar{c} = \frac{1}{2}\EE_{\tilde\bx}\|\tilde\bx-\bar\bx\|^2_{H(\bar\bx)} + \hot(3) = \frac{1}{2n^2} \sum\limits_{i=1}^n \|\bx_i-\bar\bx\|^2_{H(\bar\bx)} + \hot(3).
\end{equation}
The second equality above follows the Lemma~\ref{lm:quadratic_form}. Then, a Taylor expansion of $H(\bar\bx)$ at $\bx^*$ further gives
{\small
\begin{align}
-\bar{c} &= \frac{1}{2n^2} \sum\limits_{i=1}^n \|\bx_i-\bar\bx\|^2_{H(\bx^*)} + \frac{1}{2n^2} \sum\limits_{i=1}^n \|\bx_i-\bar\bx\|^2_{H(\bar\bx)-H(\bx^*)} + \hot(3) \label{eqn:pf_skecth_c2_1} \\
&= \frac{1}{2n^2} \sum\limits_{i=1}^n \|\bx_i-\bar\bx\|^2_{H(\bx^*)} + \frac{\hot(1)}{n} + \hot(3). \label{eqn:pf_sketch_c2_2}
\end{align}}
In Lemma~\ref{lm:moments2}, we show that generally we have $\EE\hot(k)=O(n^{-k/2})$ (strictly speaking, we need to take the expectation of its square). Also, $\sum_{i=1}^n \|\bx_i-\bar\bx\|^2_{H(\bx^*)}\sim O(n)$. Therefore, from~\ref{eqn:pf_sketch_c2_2} we can obtain
\begin{equation*}
    \EE \bar{c}^2 = \frac{1}{4n^4}\EE\left(\sum\limits_{i=1}^n \|\bx_i-\bar\bx\|^2_{H(\bx^*)}\right)^2 + O(\frac{1}{n^{2.5}}).
\end{equation*}
Finally, by showing
\begin{equation*}
    \frac{1}{4n^4}\EE\left(\sum\limits_{i=1}^n \|\bx_i-\bar\bx\|^2_{H(\bx^*)}\right)^2 = \frac{\tr(M_2H(\bx^*))^2}{4n^2} + O(\frac{1}{n^{3}}) = \frac{\sigma_2^2}{4n^2} + O(\frac{1}{n^{3}}),
\end{equation*}
we have the following estimate for $\EE\bar{c}^2$:
\begin{equation}\label{eqn:pf_sketch_c2}
    \EE \bar{c}^2 = \frac{\sigma_2^2}{4n^2} + O(\frac{1}{n^{2.5}}).
\end{equation}
Note that the leading term in the estimate has order $O(\frac{1}{n^2})$. This is also the leading term in all the following estimates. Finally we compare the coefficients of the leading terms and complete the proof by showing that the total coefficient is negative. The higher-order-terms $O(\frac{1}{n^{2.5}})$ can be made sufficiently small when $n$ is sufficiently large.

\subsubsection*{Estimate of $\EE\delta^2$} 
First, we have
\begin{align*}
\EE \delta^2 & = \EE \left(\EE_{\tilde{\bx}} F(\tilde{\bx}) - \frac{1}{K}\sum\limits_{k=1}^K F(\tilde{\bx}_k)\right)^2 = \EE \textrm{Var}_{\tilde\bx}\left(\frac{1}{K}\sum\limits_{k=1}^K F(\tilde{\bx}_k)\right) = \frac{1}{K}\EE \textrm{Var}_{\tilde\bx}(F(\tilde\bx)) \\
& = \frac{1}{K}\EE\EE_{\tilde\bx} \big(F(\tilde{\bx})-\EE_{\tilde{\bx}}F(\tilde{\bx})\big)^2 \leq \frac{1}{K}\EE\EE_{\tilde\bx} \big(F(\tilde{\bx})-F(\bar\bx)\big)^2 = \frac{1}{K}\EE (F(\tilde{\bx})-F(\bar\bx))^2.
\end{align*}
Still using a Taylor expansion of $F(\tilde\bx)$ at $\bar\bx$, substituting $K=C_Kn$, we have 
\begin{equation*}
\EE \delta^2 \leq \frac{1}{C_Kn}\EE\left(\nabla F(\bar\bx)^T(\tilde{\bx}-\bar\bx) + \hot(2) \right)^2.
\end{equation*}
Then, expanding $\nabla F(\bar\bx)$ at $\bx^*$ gives
\begin{align*}
\EE \delta^2 & \leq \frac{1}{C_Kn}\EE\left(\nabla F(\bx^*)^T(\tilde{\bx}-\bar\bx) + \hot(2) \right)^2 \\
& = \frac{1}{C_Kn}\left(\EE \big(\nabla F(\bx^*)^T(\tilde\bx-\bar\bx)\big)^2 + \hot(3)\right) \\
&= \frac{1}{C_Kn}\EE \big(\nabla F(\bx^*)^T(\tilde\bx-\bar\bx)\big)^2 + O(\frac{1}{n^{2.5}}).
\end{align*}
By Lemma~\ref{lm:quadratic_form}, we have 
\begin{align*}
\EE \big(\nabla F(\bx^*)^T(\tilde\bx-\bar\bx)\big)^2 & =\frac{1}{n^2}\EE \sum\limits_{i=1}^n\|\bx_i-\bar\bx\|_{\nabla F(\bx^*)\nabla F(\bx^*)^T}^2 = \frac{1}{n} \EE \|\bx_1-\bar\bx\|_{\nabla F(\bx^*)\nabla F(\bx^*)^T}^2 \nonumber\\
&= \frac{1}{n}\EE \|\bx_1-\bx^*\|_{\nabla F(\bx^*)\nabla F(\bx^*)^T}^2 + O(\frac{1}{n^{1.5}}).
\end{align*}
Therefore, 
\begin{equation}\label{eqn:pf_sketch_delta2}
\EE\delta^2 \leq \frac{\EE \|\bx_1-\bx^*\|_{\nabla F(\bx^*)\nabla F(\bx^*)^T}^2}{C_Kn^2} + O(\frac{1}{n^{2.5}}) = \frac{\sigma_1}{C_Kn^2} + O(\frac{1}{n^{2.5}}).
\end{equation}

\subsubsection*{Estimate of $\EE (F(\bar\bx) - F(\bx^*))\bar{c}$}
A Taylor expansion of $F(\bar\bx)$ at $\bx^*$ gives 
\begin{align}
\EE(F(\bar\bx)-F(\bx^*))\bar{c} = \EE\nabla F(\bx^*)^T(\bar\bx-\bx^*)\bar{c} + \EE \frac{1}{2}\|\bar\bx-\bx^*\|^2_{H(\bx^*)}\bar{c} + \EE\hot(3)\bar{c}. \label{eqn:pf_sketch_Fc_1}
\end{align}
For the third term on the right hand side of~\ref{eqn:pf_sketch_Fc_1}, recall that the estimate for $\bar{c}$ gives $\EE\bar{c}^2=O(\frac{1}{n^2})$, which implies $\EE\ \hot(3)\bar{c}\leq O(\frac{1}{n^{2.5}})$ via the Cauchy-Schwarz inequality. 

For the second term, note that the leading term of $\bar{c}$ has a negative sign, this second term can be shown to be negative. This term represents the goodness of the debiasing quantity, and is used to offset all other positive terms, including $\EE\bar{c}^2$ and $\EE\delta^2$. Specifically, by~\eqref{eqn:pf_sketch_c2_2}, we have
\begin{align}
\EE \frac{1}{2}\|\bar\bx-\bx^*\|^2_{H(\bx^*)}\bar{c} & = -\frac{1}{4n^2}\EE \|\bar\bx-\bx^*\|^2_{H(\bx^*)}\sum_{i=1}^n \|\bx_i-\bar\bx\|^2_{H(\bx^*)} + O(\frac{1}{n^{2.5}}).\nonumber
\end{align}
Then, we show that 
{\small
\begin{equation*}
\EE \|\bar\bx-\bx^*\|^2_{H(\bx^*)}\sum_{i=1}^n \|\bx_i-\bar\bx\|^2_{H(\bx^*)} = \frac{1}{n^2}\left(\sum\limits_{i,j=1}^n \EE \|\bx_i-\bx^*\|_{H(\bx^*)}\right)^2 +  O(\frac{1}{\sqrt{n}}) = \sigma_2^2 + O(\frac{1}{\sqrt{n}}),
\end{equation*}}
(details in Section~\ref{sec:pf}). Therefore, we have 
\begin{equation}\label{eqn:pf_sketch_Fc_term2}
\EE \frac{1}{2}\|\bar\bx-\bx^*\|^2_{H(\bx^*)}\bar{c} = -\frac{\sigma_2^2}{4n^2} + O(\frac{1}{n^{2.5}}).
\end{equation}
The leading term of the estimate above has the same magnitude as the estimate for $\EE\bar{c}^2$. However, this negative part will take over $\EE\bar{c}^2$ considering the factor $2$ before $\EE(F(\bar\bx)-F(\bx^*))\bar{c}$.

Next, we consider the first term on the right hand side of~\ref{eqn:pf_sketch_Fc_1}. We need a finer representation for $\bar{c}$. By Lemma~\ref{lm:3rd_form}, we can improve~\eqref{eqn:pf_sketch_c2_0} into 
\begin{align}
-\bar{c} &= \frac{1}{2n^2} \sum\limits_{i=1}^n \|\bx_i-\bar\bx\|^2_{H(\bar\bx)} + \frac{1}{6n^3}\sum\limits_{i=1}^n \partial^3_{abc} F(\bar\bx)(x_i^a-\bar{x}^a)(x_i^b-\bar{x}^b)(x_i^c-\bar{x}^c) + \hot(3) \nonumber\\
& = \frac{1}{2n^2} \sum\limits_{i=1}^n \|\bx_i-\bar\bx\|^2_{H(\bar\bx)} + \frac{\hot(0)}{n^2} + \hot(3). \label{eqn:pf_sketch_Fc_2}
\end{align}
Taylor expansions of $H(\bar\bx)$ and $\nabla^3F(\bar\bx)$ at $\bx^*$ further give
\begin{equation*}
\bar{c} = -\frac{1}{2n^2} \sum\limits_{i=1}^n \|\bx_i-\bar\bx\|^2_{H(\bx^*)} - \frac{1}{2n^2}\sum\limits_{i=1}^n \partial^3_{abc} F(\bx^*)(x_i^a-\bar{x}^a)(x_i^b-\bar{x}^b)(\bar{x}^c-(x^*)^c) + \frac{\hot(0)}{n^2} + \hot(3)
\end{equation*}
Therefore, 
\begin{align}
\EE (F(\bar\bx) - F(\bx^*))\bar{c} & = -\EE\nabla F(\bx^*)^T(\bar\bx-\bx^*)\frac{1}{2n^2} \sum\limits_{i=1}^n \|\bx_i-\bar\bx\|^2_{H(\bx^*)} \nonumber\\
&\ \ \ - \EE\nabla F(\bx^*)^T(\bar\bx-\bx^*)\frac{1}{2n^2}\sum\limits_{i=1}^n \partial^3_{abc} F(\bx^*)(x_i^a-\bar{x}^a)(x_i^b-\bar{x}^b)(\bar{x}^c-(x^*)^c) \nonumber\\
&\ \ \ + \EE\nabla F(\bx^*)^T(\bar\bx-\bx^*)\left(\frac{\hot(0)}{n^2} + \hot(3)\right). \label{eqn:pf_sketch_Fc_3}
\end{align}
The last term in~\eqref{eqn:pf_sketch_Fc_3} is obviously $O(\frac{1}{n^{2.5}})$. For the other two terms, we have the following estimate (see the details in Section~\ref{sec:pf}):
\begin{equation*}
-\EE\nabla F(\bx^*)^T(\bar\bx-\bx^*)\frac{1}{2n^2} \sum\limits_{i=1}^n \|\bx_i-\bar\bx\|^2_{H(\bx^*)} = -\frac{\sigma_3}{2n^2} + O(\frac{1}{n^{2.5}}),
\end{equation*}
\begin{equation*}
- \EE\nabla F(\bx^*)^T(\bar\bx-\bx^*)\frac{1}{2n^2}\sum\limits_{i=1}^n \partial^3_{abc} F(\bx^*)(x_i^a-\bar{x}^a)(x_i^b-\bar{x}^b)(\bar{x}^c-(x^*)^c) =-\frac{\sigma_4}{2n^2} + O(\frac{1}{n^3}).
\end{equation*}
Altogether, we have 
\begin{equation*}
    \EE\nabla F(\bx^*)^T(\bar\bx-\bx^*)\bar{c} \leq -\frac{\sigma_3+\sigma_4}{2n^2} + O(\frac{1}{n^{2.5}}).
\end{equation*}
And for $\EE(F(\bar\bx)-F(\bx^*)\bar{c})$ we have 
\begin{equation}\label{eqn:pf_sketch_Fc}
\EE(F(\bar\bx)-F(\bx^*)\bar{c}) \leq -\frac{\sigma_2^2}{4n^2} - \frac{\sigma_3+\sigma_4}{2n^2} + O(\frac{1}{n^{2.5}}).
\end{equation}

\subsubsection{Putting together}
Finally, combining the estimates~\eqref{eqn:pf_sketch_c2}, \eqref{eqn:pf_sketch_delta2}, and~\eqref{eqn:pf_sketch_Fc}, we have 
\begin{equation}
\EE \left(\bar{c}^2+\delta^2+2(F(\bar\bx)-F(\bx^*)\bar{c})\right) \leq -\frac{\sigma_2^2}{4n^2} + \frac{\sigma_1}{C_Kn^2} - \frac{\sigma_3+\sigma_4}{n^2} + O(\frac{1}{n^{2.5}}).
\end{equation}
The coefficient of $\frac{1}{n^2}$ is negative when condition~\ref{eqn:thm_cond} is satisfied. This completes the proof.

\section{Numerical experiments}\label{sec:exp}

\newcommand{\rmser}{\textrm{RMSE}_r}
\newcommand{\biasr}{\textrm{Bias}_r}

In this section we show numerical results of the debiasing methods proposed in Section~\ref{sec:general}. As a general framework of debiasing convex/concave functions, we test our methods on a wide variety
of problems ranging from simple convex functions, optimization problems, to functionals of probability distributions. 
Specifically, we test the following seven problems:
(P1) Quadratic functions,
(P2) Fourth-order polynomials,
(P3) Rational functions,
(P4) Unconstrained optimization problems with random objective functions,
(P5) Constrained optimization problems with random constraints,
(P6) Entropy of discrete probability distribution,
(P7) Wasserstein distance between two probability distributions.
For all the problems, we apply both the shifting and scaling methods using bootstrap. The covariance estimate method is tested on a subset of problems for which it is easy to obtain an estimate of the Hessian at $\bx^*$.

\subsection{P1: Quadratic functions}
We consider simple multivariate quadratic functions 
\begin{equation*}
F(\bx) = \bx^TA\bx\tag{P1},
\end{equation*}
with $\bx\in\RR^d$ and $A\in\RR^{d\times d}$ being a positive definite matrix. Obviously, $F$ is a convex function with respect to $\bx$. As described in previous sections, given $A$ and a probability distribution $\mu$ on $\RR^d$ which satisfies $\EE_\mu \bx=\bx^*$, we estimate $F(\bx^*)$ using samples from $\mu$. The debiasing methods are tested for many different $A$'s and $\mu$'s. For each $A$ and $\mu$, the experiment is repeated by $R=1000$ times, and in each experiment  we first sample a new set of noisy observations and then run the debiasing methods. 

To measure and compare the performance of our debiasing methods, we compute the root mean squared error (RMSE) and the average bias of the estimates across $1000$ experiments, and compare the values with that of the naive estimate. Concretely, let $F_{\textrm{debias}}^i$ and $F_{\textrm{naive}}^i$ be the estimate given by the debiasing method and the naive method in the $i$-th experiment, respectively, we compute and compare the following two relative quantities:
\begin{equation}\label{eqn:error_def}
\textrm{RMSE}_r = \frac{\sqrt{\sum_{i=1}^{1000} (F_{\textrm{debias}}^i-F(\bx^*))^2}}{\sqrt{\sum_{i=1}^{1000}(F_{\textrm{naive}}^i-F(\bx^*))^2}},\qquad \textrm{Bias}_r = \frac{\sum_{i=1}^{1000}(F_{\textrm{debias}}^i-F(\bx^*))}{\sum_{i=1}^{1000}(F_{\textrm{naive}}^i-F(\bx^*))}.
\end{equation}
By the definitions, when the debiasing method is effective, we expect $\rmser\in[0,1)$ and $\biasr\in(-1,1)$. The closer these quantities are to zero, the more effective the debiasing method. These two errors will also be studied in the experiments for other problems. 

Before the debiasing methods are applied, a generic positive definite matrix $A$ is generated by its eigenvalue decomposition. Concretely, we generate a diagonal matrix $\Lambda$ with positive entries, and an orthogonal matrix $Q$, and then let $A=Q\Lambda Q^T$. We take $\mu$ as an isotropic Gaussian distribution centered at $\bx^*$, i.e. $\mu=N(\bx^*,\sigma^2I)$. 
The experiment results for (P1) are shown in Figure~\ref{fig:exp_P1_1} and~\ref{fig:exp_P1_2}. In the experiments, we study the performance of all the three debiasing methods (shifting, scaling, covariance estimate) for problems with different dimensions (Figure~\ref{fig:exp_P1_1} left), different condition numbers of $A$ (Figure~\ref{fig:exp_P1_1} middle), different noisy strength $\sigma$ (Figure~\ref{fig:exp_P1_1} right), and different norm of $\bx^*$ (Figure~\ref{fig:exp_P1_2} left). Throughout these experiments, we take $n=10$ noisy observations of $\bx^*$, and do $K=10$ rounds of resampling when bootstrap is applied. The figures show that the shifting debiasing method and the covariance estimate method can usually significantly reduce the estimate error and bias, while the scaling method does not perform as well when the dimension is large or the noise is strong. The debiasing methods become less effective when the noise level is low or $\|\bx^*\|$ is large, in which case the noise observations are (effectively) close to $\bx^*$ and the naive estimate is already good. The experiments on the condition number of $A$ show that our methods are not sensitive to the spectrum of $A$. Here we remark that the covariance estimate method performs especially well because the function $F$ is quadratic, hence the Hessian matrix has all information about the function.

\begin{figure}[ht]
    \centering
    \includegraphics[width=0.3\textwidth]{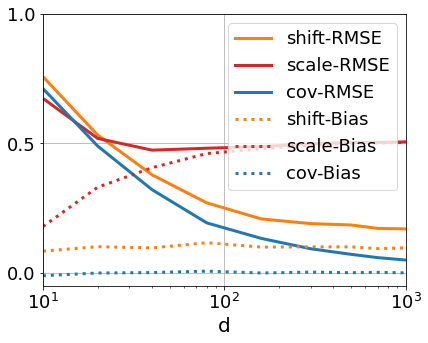}
    \includegraphics[width=0.3\textwidth]{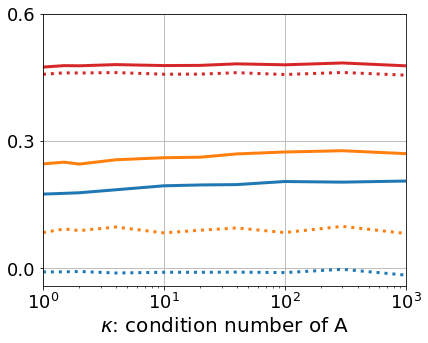}
    \includegraphics[width=0.3\textwidth]{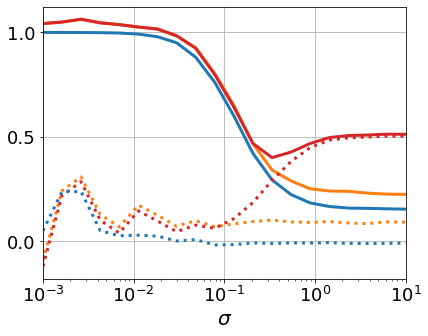}
    \caption{Experiment results for the quadratic function (P1). For all the experiments shown in this figure, we take $n=K=10$ and $\|\bx^*\|^2=2$. {\bf (left)} the relative RMSEs and biases for problems with different dimension $d$, while we fix $\sigma=1$ and $\kappa=2$ ($\kappa$ is the condition number of $A$). {\bf (middle)} results for problems with different $\kappa$, while we fix $d=100$ and $\sigma=1$. {\bf (right)} results for problems with different $\sigma$, while we fix $d=100$ and $\kappa=2$.}
    \label{fig:exp_P1_1}
\end{figure}

For the bootstrap methods, we also study the number of resampling rounds $K$ required for different $n$. Results in Figure~\ref{fig:exp_P1_2} show that a constant $K$ is usually sufficient to approach best-achievable performance for any $n$. The RMSE almost stops decreasing after $K\approx 50$ even when $n$ is much bigger. This shows that our theories in Section~\ref{sec:thm} can potentially be improved.

\begin{figure}[ht]
    \centering
    \includegraphics[width=0.3\textwidth]{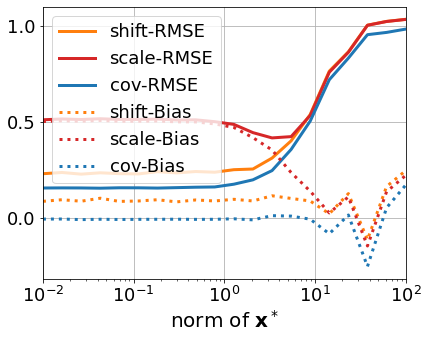}
    \includegraphics[width=0.3\textwidth]{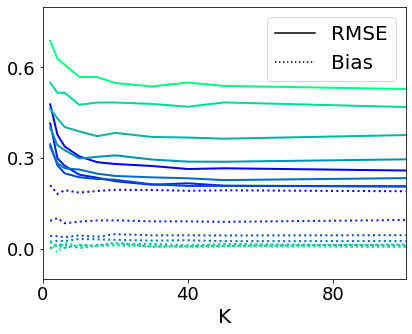}
    \includegraphics[width=0.38\textwidth]{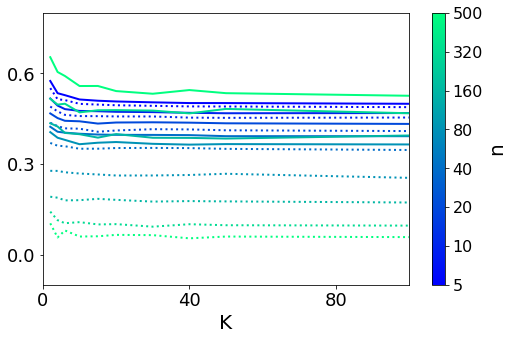}
    \caption{{\bf (left)} results for P1 with different norms of $\bx^*$, while we fix $d=100$, $\kappa=2$, and $\sigma=1$. {\bf (middle and right)} Bootstrap methods (shifting and scaling) for different $n$ and $K$. The middle panel shows the results for the shifting method, and the right panel shows the results for the scaling method. Results for different $n$ are shown in lines with different colors. During the experiments we fix $d=100$, $\kappa=2$, and $\sigma=1$.}
    \label{fig:exp_P1_2}
\end{figure}

\subsection{P2: Fourth-order polynomials}
We then consider a higher-order polynomial function given by
\begin{equation*}
F(\bx) = (\bx^TA\bx)^2\tag{P2},
\end{equation*}
where $\bx\in\RR^d$ and $A\in\RR^{d\times d}$ is a positive definite matrix. Compared with P1, this function has non-uniform curvature. It is flat when $\bx$ is small and grows fast when $\bx$ is large. In P2, we consider similar experiment settings as those for $P1$. We test the methods against problems with different dimensions, noise levels, and norms of $\bx^*$. The results are shown in Figures~\ref{fig:exp_P2_1}. We see that for this problem only the scaling method performs well in most cases. The covariance estimate method works when $\sigma$ and $\|\bx^*\|$ lie in specific ranges. The shifting method sometimes even performs worse than the naive estimate. Experiments on the condition number of $A$ still show that the methods are insensitive with $\kappa$. We ignore the figure here.

\begin{figure}[ht]
    \centering
    \includegraphics[width=0.3\textwidth]{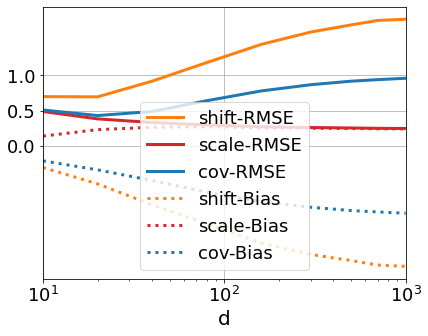}
    \includegraphics[width=0.3\textwidth]{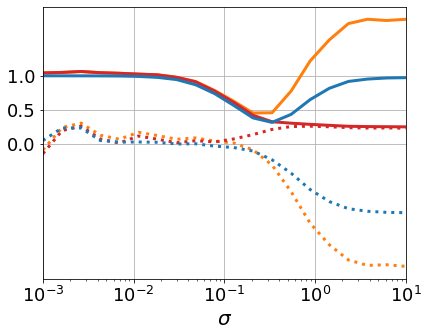}
    \includegraphics[width=0.3\textwidth]{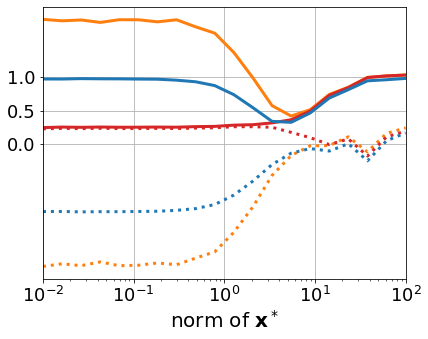}
    \caption{Experiment results for the fourth-order function (P2). For all the experiments, we take $n=K=10$ and $\kappa=2$. {\bf (left)} the relative RMSEs and biases for problems with different dimension $d$, while we fix $\sigma=1$ and $\|\bx^*\|^2=2$. {\bf (middle)} results for problems with different $\sigma$, while we fix $d=100$ and $\|\bx^*\|^2=2$. {\bf (right)} results for problems with different $\|\bx^*\|$, while we fix $d=100$ and $\sigma=1$.}
    \label{fig:exp_P2_1}
\end{figure}

\subsection{P3: Rational functions}
Next, we consider a rational function
\begin{equation*}
F(\bx) = \sum\limits_{i=1}^d \big(b_ix_i + \frac{c_i}{x_i}\big)\tag{P3},
\end{equation*}
where $\bx=(x_1,...,x_d)^T\in\RR^d$ is the input and $b_i, c_i$, $i=1,2,...,d$ are positive numbers. This function is convex for $x_i>0$, $i=1,...,d$. Hence, in the experiments, we take a coordinate-wise exponential distribution as $\mu$. Fixing $n=K=10$, we test our methods on problems with different dimensions, magnitudes of $c_i$, and $\|\bx^*\|$. Note that changing the magnitude of $b_i$ is equivalent with changing $c_i$ in the opposite direction. Results in Figure~\ref{fig:exp_P3} show that the three debiasing methods perform similarly. They become more effective when the dimension of the problem increases. When $c$ is small or $\bx^*$ is large, the debiasing methods do not provide much improvement compared with the naive estimate. This is because in these cases the function is almost linear around $\bx^*$ and the convexity bias is small.

\begin{figure}[ht]
    \centering
    \includegraphics[width=0.3\textwidth]{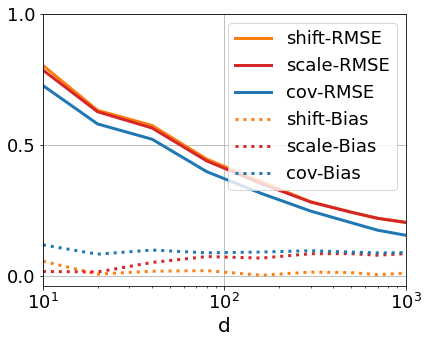}
    \includegraphics[width=0.3\textwidth]{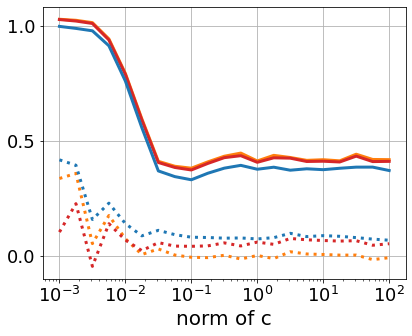}
    \includegraphics[width=0.3\textwidth]{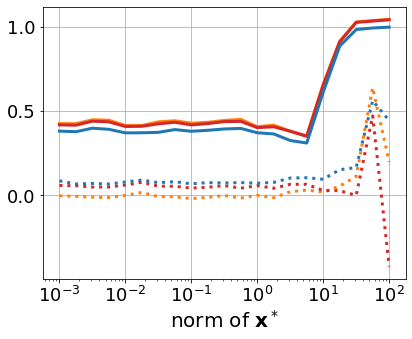}
    \caption{Experiment results for the rational function (P3). For all the experiments, we keep $n=K=10$ and $\|\bb\|=1$. {\bf (left)} the relative RMSEs and biases for problems with different dimension $d$, while we fix $\|\bc\|=1$ and $\|\bx^*\|=2$. {\bf (middle)} results for problems with different $\|\bc\|$, while we fix $d=100$ and $\|\bx^*\|^2=2$. {\bf (right)} results for problems with different $\|\bx^*\|$, while we fix $d=100$ and $\|\bc\|=1$.}
    \label{fig:exp_P3}
\end{figure}

\subsection{P4: Unconstrained optimization problems}
When the objective function is parameterized by a parameter vector, the optimal value of some optimization problems is a convex or concave function of the parameters. For instance, let $F(\alpha)$ be the optimal value of the following minimization problem given some $\alpha\in\RR^d$,
\begin{equation}\label{eqn:opt_uncons}
    \min_\bx \alpha^Tf(\bx) + g(\bx),
\end{equation}
where $f:\RR^p\rightarrow\RR^d$, $g:\RR^p\rightarrow\RR$ are functions of $\bx$. Then, $F$ is a concave function of $\alpha$ because it is the minimum of a family of affine functions~\cite{boyd2004convex}. Therefore, our debiasing methods can be applied when we are interested in the minimum value of~\eqref{eqn:opt_uncons} at $\alpha^*$ but can only get access to its noisy observations. 

In the experiments, we consider specifically the minimization of a quadratic function when the Hessian of the objective function has some randomness, i.e. the concave function we debias is 
\begin{equation*}
F(A) = \min_{\bx\in\RR^d} \frac{1}{2}\bx^TA\bx + \bb^T\bx,\tag{P4}
\end{equation*}
where $A\in\RR^{d\times d}$ is a positive definite matrix with randomness, and $\bb$ is a fixed vector. We easily have $F(A) = -\frac{1}{2} \bb^T A^{-1} \bb$. Given a generic groundtruth Hessian matrix $A^*=U\Lambda U^T$, where $U\Lambda U^T$ is the eigenvalue decomposition of $A^*$ and $\Lambda=\textrm{diag}(\lambda_1, ..., \lambda_d)$, we generate noisy observations of $A^*$ by sampling 
\begin{equation*}
    A \sim U\cdot\textrm{diag}(\xi_1\lambda_1, ..., \xi_d\lambda_d)\cdot U^T,
\end{equation*}
where $\xi_1, ..., \xi_d$ are independent random variables that follows the gamma distribution with shape $k$ and scale $1/k$, whose density function is $f(x)=\frac{k^k}{\Gamma(k)}x^{k-1}e^{-kx}$ and expectation is $1$. This makes sure that all noisy observations are positive definite. 

Fixing $n=10$ and $K=100$, Figure~\ref{fig:exp_P4} shows the RMSEs and biases of the shifting and scaling methods for different dimensions, condition numbers of $A^*$, and the shape parameter $k$. Note that the actual dimension of the problem is $d^2$. The results show that our methods are generally effective. The methods perform better when the dimension is higher, or the condition number of $A$ is not too large. When $k$ is small, i.e. the variance of the data distribution is large, the shifting method does not work well, while the scaling method keeps working well.

\begin{figure}[ht]
    \centering
    \includegraphics[width=0.3\textwidth]{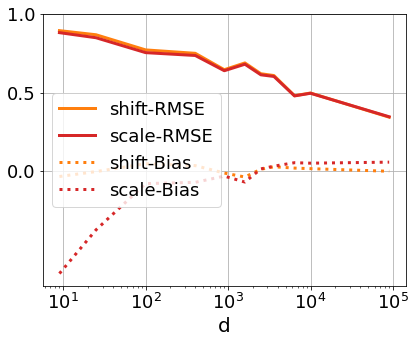}
    \includegraphics[width=0.31\textwidth]{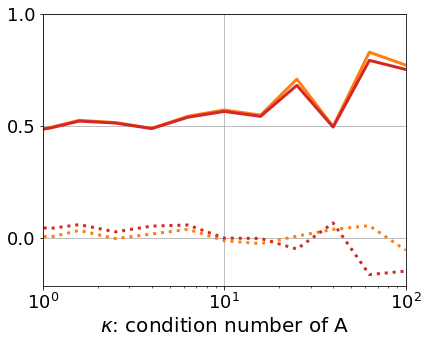}
    \includegraphics[width=0.3\textwidth]{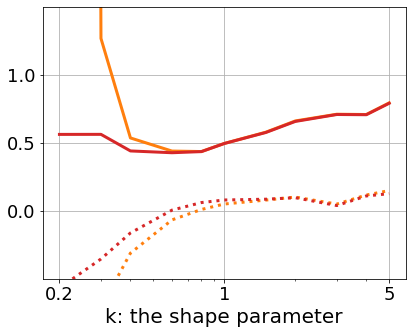}
    \caption{Experiment results for the rational function (P4). For all the experiments, we keep $n=10$ and $K=100$. {\bf (left)} the relative RMSEs and biases for problems with different dimension $d^2$, while we fix $\kappa=2$ and $k=1$. {\bf (middle)} results for problems with different $\kappa$, while we fix $d=100$ and $k=1$. {\bf (right)} results for problems with different $k$, while we fix $d=100$ and $\kappa=2$.}
    \label{fig:exp_P4}
\end{figure}

\subsection{P5: Constrained optimization problem}
Convex debiasing methods can also be applied to constrained optimization problems with randomness. The randomness can appear on the constraints. Consider
\begin{align}
    \min_\bx\ f(\bx)\ \ \ \  s.t.\ A\bx=\bb. \label{eqn:opt_cons}
\end{align}
Fixing $f$ and $A$, let $F(\bb)$ be the optimal value of~\eqref{eqn:opt_cons} given $\bb$. Suppose strong duality holds, then 
\begin{equation*}
F(\bb) = \min_\bx \max_\mu f(\bx) + \mu^T(A\bx-\bb) = \max_\mu \big(\mu^T\bb + \max_\bx (f(\bx)-\mu^TA\bx)\big),
\end{equation*}
which shows that $F(\bb)$ is the maximum of affine functions of $\bb$. Therefore, $F$ is convex. The debiasing methods applies when we are interested in the optimal value at some $\bb^*$ while only its noisy observations can be obtained. 

In the experiments, we estimate
\begin{equation*}
F(\bb^*) = \arg\min_{\bx}\ \bx^TB\bx, \quad s.t.\ A\bx=\bb\tag{P5},
\end{equation*}
where $\bx\in\RR^p$, $B\in\RR^{p\times p}$ is a positive definite matrix, $A\in\RR^{d\times p}$, and $\bb\in\RR^d$. Given $\bb^*$, we take $\mu$ as an isotropic Gaussian distribution centered at $\bb^*$, i.e. $\mu=N(\bb^*, \sigma^2I)$. Throughout the experiments, we fix $n=10$ and $K=100$, and test the methods for different dimensions, ratio of $d/p$, and noise level $\sigma$. Figure~\ref{fig:exp_P5} shows the results. We see that both methods can significantly reduce the bias and RMSE, while the shifting method usually outperforms the scaling method. Better performance is achieved when $d$ is large, $d$ is not too close to $p$, or $\sigma$ is large. 

\begin{figure}[ht]
    \centering
    \includegraphics[width=0.3\textwidth]{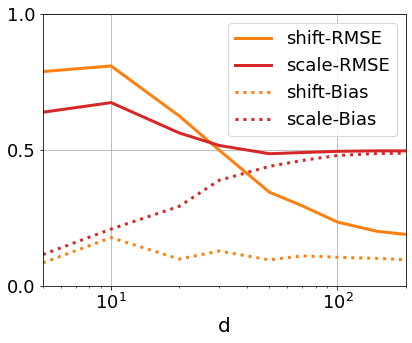}
    \includegraphics[width=0.3\textwidth]{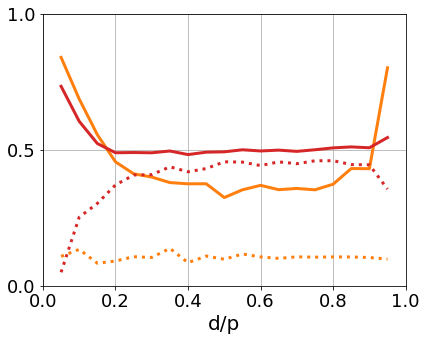}
    \includegraphics[width=0.3\textwidth]{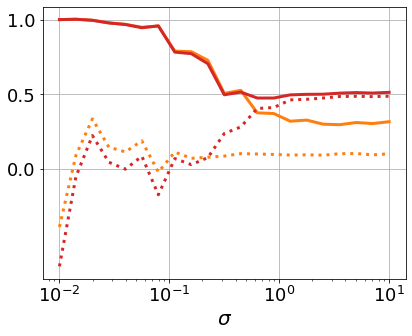}
    \caption{Experiment results for the constrained optimization problem (P5). For all the experiments, we keep $n=10$, $K=100$, and $\|\bb\|=1$. {\bf (left)} the relative RMSEs and biases for problems with different dimension $d$, while we take $p=2d$. {\bf (middle)} results for problems with different $d$ fixing $p=100$. The horizontal axis is the ratio of $d$ and $p$. {\bf (right)} results for problems with different $\sigma$, while we fix $d=100$ and $p=200$.}
    \label{fig:exp_P5}
\end{figure}

\subsection{P6: Entropy}
Next, we test the estimate of entropy for discrete probability distributions, i.e.
\begin{equation*}
    H(p) = -\sum\limits_{i=1}^n p_i \ln p_i, \tag{P6}
\end{equation*}
where $p=(p_1,...,p_d)$ satisfies $p_i\geq0$ and $\sum_{i=1}^d p_i=1$.
We generate the groundtruth distribution $p^*$ from symmetric Dirichlet distribution with parameter $\alpha$. Noisy observations of $p^*$ are the empirical distributions of single samples from $p^*$. 
In the experiments, we fix $K=100$, and study the performance of the debiasing methods for problems with different dimension $d$, $\alpha$, and the ratio of $n$ and $d$.

For this problem, the covariance estimate method takes simple form. Let $\bar{p}=(p_1,p_2,...,p_d)$ be the average of noisy observations of $p^*$. Then, the covariance matrix is $\textrm{diag}(\bar{p}) - \bar{p}\bar{p}^T$, and the Hessian of the entropy functional (viewed as a function of vector $p$) is $\textrm{diag}(-1/\bar{p})$. Hence, we have $\tr(CH)=-(d-1)$, and the debiased estimate is $H(\bar{p})+\frac{d-1}{2n}$. It is the same as the classical debiasing scheme in~\cite{miller1955note}.

The results are shown in Figure~\ref{fig:exp_P6}. For well-behaved problems, all three methods perform similarly. They can all significantly reduce the bias and error. When $\alpha$ is small, i.e. the distribution is sparse, or when $n$ is small compared with $d$, the bootstrap methods perform better than the covariance estimate method. We remark here that the experiments are part of the effort to show the wide applicability of our methods. Many delicate methods have been developed specifically for the entropy estimate problem, and we do not expect our methods to outperform all of them on this specific problem. Hence, we do not compare our methods with state-of-the-art methods here. 

\begin{figure}[ht]
    \centering
    \includegraphics[width=0.3\textwidth]{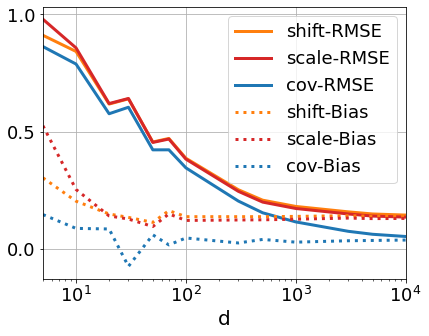}
    \includegraphics[width=0.3\textwidth]{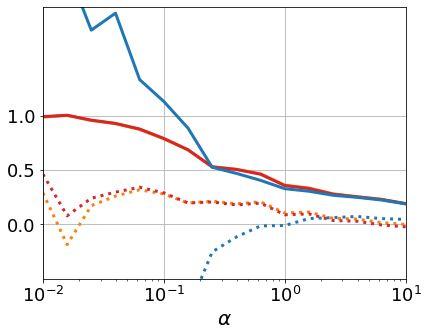}
    \includegraphics[width=0.3\textwidth]{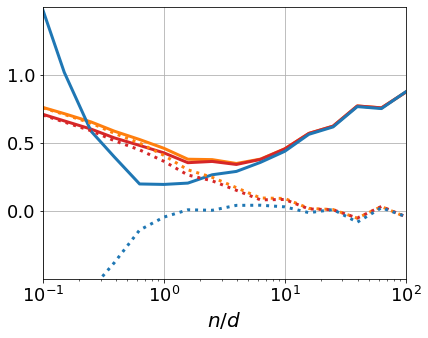}
    \caption{Experiment results for the entropy estimate problem (P6). For all the experiments, we keep $K=100$. {\bf (left)} the relative RMSEs and biases for problems with different dimension $d$, while we take $n=5d$ and $\alpha=1$. {\bf (middle)} results for problems with different $\alpha$ fixing $d=100$ and $n=500$. {\bf (right)} results for problems with different $n/d$, while we fix $d=100$ and $\alpha=1$.}
    \label{fig:exp_P6}
\end{figure}

\subsection{P7: Wasserstein distance}
Finally, we estimate the 2-Wasserstein distance between a pair of distributions ($p$, $q$) on $\RR^d$, defined by
\begin{equation*}
    W_2(p,q)^2 = \int_{\RR^d \times \RR^d} \|\bx-\by\|_2^2 d\gamma(\bx,\by),\tag{P7}
\end{equation*}
where $\gamma$ is a coupling of $p$ and $q$ whose two marginals are $p$ and $q$, respectively. In the experiments, $p^*$ and $q^*$ are taken as Gaussian distributions $N(\mu_1, \sigma_1^2I)$ and $N(\mu_2, \sigma_2^2I)$, where $\mu_1, \mu_2\in\RR^d$ are two vectors fixed for all trials. Noisy observations of $p^*$ and $q^*$ are empirical distributions of i.i.d. samples of the distributions. The Wasserstein distance between empirical distributions are computed by a linear programming problem. Specifically, let $\hat{p}$ be an empirical distribution on $\bx_1,...,\bx_n$, $\hat{q}$ be an empirical distribution on $\by_1,...,\by_m$, $W(\hat{p},\hat{q})$ is given by the following problem
\begin{equation}\label{eqn:wass}
    \min_{\gamma\in\RR^{m\times n}} \sum\limits_{i,j} \gamma_{ij}\|\bx_i-\by_j\|^2,\quad s.t. \gamma\mathbf{1}=\frac{1}{m}\mathbf{1},\ \ \gamma^T\mathbf{1}=\frac{1}{n}\mathbf{1},
\end{equation}
where $\mathbf{1}$ denotes the all-one vector with proper size. 

In the experiments, without loss of generality, we take $\mu_1=0$. We consider $\sigma_1=\sigma_2=\sigma$, and fix $\sigma=1$, $K=50$. Figure~\ref{fig:exp_P7} shows the results for problems with different dimension $d$, $\|\mu_2\|$, and $n$. Results show that our methods can improve the naive estimate. The improvement gets smaller when the dimension is higher. This might be caused by the essential difficulty of estimating the Wasserstein distance in high dimensional spaces~\cite{weed2019sharp}.

\begin{figure}[ht]
    \centering
    \includegraphics[width=0.3\textwidth]{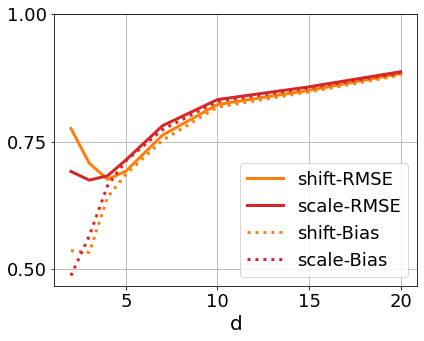}
    \includegraphics[width=0.3\textwidth]{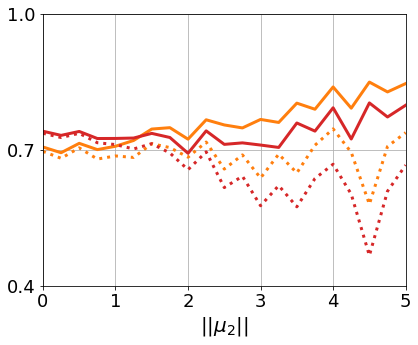}
    \includegraphics[width=0.3\textwidth]{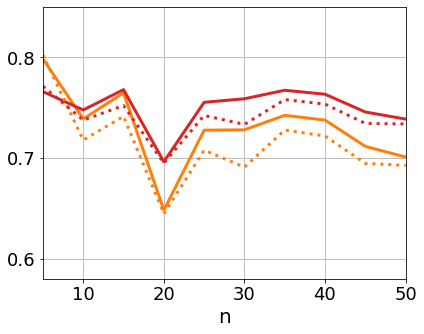}
    \caption{Experiment results for the Wasserstein distance problem (P7). For all the experiments, we keep $K=50$. {\bf (left)} the relative RMSEs and biases for problems with different dimension $d$, while we take $n=10$ and $\|\mu_2\|=1$. {\bf (middle)} results for problems with different $\|\mu_2\|$ fixing $d=5$ and $n=10$. {\bf (right)} results for problems with different $n$, while we fix $d=5$ and $\|\mu_2\|=1$.}
    \label{fig:exp_P7}
\end{figure}

\section{Summary}
In this work, we propose a general framework to correct the bias introduced by convexity/concavity when estimating the value of functions or functionals. Our methods do not require domain knowledge of the objective functions, and can be applied as long as noisy observations of the groundtruth input are available. Numerical experiments on a wide range of problems show the effectiveness of our methods. 

While our methods have general applicability, their performance may not be optimal on specific problems. For example, extensive researches are conducted on the estimate of entropy, and some proposed methods might have better empirical or theoretical properties than our methods. We emphasize that the value of our approach lies mostly on its generality---it can serve as an off-the-shelf tool to debias the estimate of convex functions and obtain improvement upon the naive estimate using the sample mean. 

Nevertheless, the gain is not free---it comes with a price of additional computational cost. Compared with the naive estimate, our methods with bootstrap requires evaluating the function for multiple times. For some applications, such as optimization problems, evaluating the function can be expensive. Though, what to blame is not the methods, but the limited observations and knowledge on the problem. This is the cost of generality. If we have more domain knowledge, the covariance estimate method can be applied without high computational cost.

\section{Proof of Theorem~\ref{thm:main_shifting}}\label{sec:pf}
\subsection{Notations.} 
In the proof, we write $h=O(g)$ if there exists a constant $C$ independent with $n$ such that $|h|< C|g|$ always holds. When using these $O$ notations, the dimension $d$ is treated as a constant. We use $\hot(\cdot;\cdot)$ to denote higher-order-terms appearing in the Taylor expansion. Specifically, for quantities $f_1,\cdots,f_k$ and integer $r$, $\hot(f_1,\cdots,f_k;r)$ denotes (sum of) terms with form $f_1^{r_1}\cdots f_k^{r_k}$ with $r_1+\cdots+r_k\geq r$.
When taking expectation, we use $\EE$ to denote the expectation over both the sampling of $\xobs$ and the sampling of $\{\tilde{\bx}\}$ during bootstrap, and use $\EE_{\tilde{\bx}}$ to denote the expectation over the sampling of $\{\tilde{\bx}\}$ based on a fixed set of $\bx_1, ..., \bx_n$. The norm $\|\cdot\|$ always means 2-norm. We use $H(\bx)$ to denote the Hessian matrix of $F$ at $\bx$, i.e. $H(\bx)=\nabla^2 F(\bx)$. For any matrix $A\in\RR^{d\times d}$ and vector $\bx\in\RR^d$, denote $\|\bx\|^2_A = \bx^TA\bx$.

We use Einstein notations mostly when tensors with dimension $\geq3$ are involved. When only vectors and matrices are involved, we still use the conventional matrix product notations. For vectors, we use superscript as the index of entries, and subscript as the index of the vector in a set of vectors. Entries are denoted by lower case letters. For example, the $i$-th noisy observation vector is given by $\bx_i=(x^1_i, x^2_i, \cdots, x^d_i)\in\RR^d$.

\subsection{Main proof}
To prove the theorem, first notice that 
\begin{equation*}
E(F(\bar\bx)+\hc-F(\bx^*))^2 = \EE (F(\bar\bx) - F(\bx^*))^2 + \EE \hc^2 + 2\EE (F(\bar\bx) - F(\bx^*))\hc.
\end{equation*}
Hence, we only need to show $\EE \big(\hc^2 + 2(F(\bar\bx) - F(\bx^*))\hc\big) < 0$. Let 
\begin{equation*}
\bar{c} = \EE_{\tilde{\bx}}\hc = F(\bar\bx) - \EE_{\tilde{\bx}} F(\tilde{\bx}),\ \ \textrm{and}\ \delta=\EE_{\tilde{\bx}} F(\tilde{\bx}) - \frac{1}{K}\sum_{k=1}^K F(\tilde{\bx}_k). 
\end{equation*}
Then, we have $\hc = \bar{c} + \delta$ and $\EE_{\tilde{\bx}}\delta = 0$, and 
\begin{equation*}
\EE \big(\hc^2 + 2(F(\bar\bx) - F(\bx^*))\hc\big) = \EE\big(\bar{c}^2+2\bar{c}\delta+\delta^2 + 2(F(\bar\bx) - F(\bx^*))\bar{c}+2(F(\bar\bx) - F(\bx^*))\delta\big).
\end{equation*}
Note that $\bar{c}$ and $F(\bar{\bx})-F(\bx^*)$ do not depend on the sampling of $\tilde{\bx}_k$, by the law of total expectation, we have 
\begin{equation*}
    \EE \bar{c}\delta = \EE\left(\EE_{\tilde{\bx}}\bar{c}\delta\right) = \EE\left(\bar{c}\EE_{\tilde{\bx}}\delta\right) = 0,
\end{equation*}
and similarly $\EE(F(\bar{\bx})-F(\bx^*))\delta=0$. Therefore,
\begin{equation}\label{eqn:thm_shifting:pf:1}
\EE \big(\hc^2 + 2(F(\bar\bx) - F(\bx^*))\hc\big) = \EE\big(\bar{c}^2+\delta^2 + 2(F(\bar\bx) - F(\bx^*))\bar{c}\big).
\end{equation}

Next, we estimate the three terms on the right hand side of~\eqref{eqn:thm_shifting:pf:1} separately. In the estimates below, we take $O(\frac{1}{n^2})$ terms as leading terms, and show that there is no term with lower order. 

\subsubsection{Estimate of $\EE\bar{c}^2$.}\label{sssec:cbar2_shift}
In the proof, we will extensively use the Taylor expansion of $F$ and its derivatives. The results are given in Lemma~\ref{lm:taylor}.
By the the Taylor expansion~\eqref{eqn:taylor} in Lemma~\ref{lm:taylor}, we have 
\begin{equation*}
\small
F(\tilde\bx)-F(\bar\bx) = \nabla F(\bar\bx)^T(\tilde\bx-\bar\bx) + \frac{1}{2}\|\tilde\bx-\bar\bx\|^2_{H(\bar\bx)} + \frac{1}{6}\partial^3_{abc} F(\bar\bx)(\tilde{x}^a-\bar{x}^a)(\tilde{x}^b-\bar{x}^b)(\tilde{x}^c-\bar{x}^c) + O(\|\tilde\bx-\bar\bx\|^4). 
\end{equation*}
Taking expectation over $\tilde\bx$, note that $\EE_{\tilde{\bx}}(\tilde\bx-\bar\bx)=0$, we have 
\begin{equation*}
-\bar{c} = \frac{1}{2}\EE_{\tilde\bx}\|\tilde\bx-\bar\bx\|^2_{H(\bar\bx)} + \frac{1}{6}\EE_{\tilde\bx}\partial^3_{abc} F(\bar\bx)(\tilde{x}^a-\bar{x}^a)(\tilde{x}^b-\bar{x}^b)(\tilde{x}^c-\bar{x}^c) + \EE_{\tilde\bx}O(\|\tilde\bx-\bar\bx\|^4).
\end{equation*}
According to Lemma~\ref{lm:quadratic_form} and~\ref{lm:3rd_form}, we have
\begin{equation*}
-\bar{c} = \frac{1}{2n^2} \sum\limits_{i=1}^n \|\bx_i-\bar\bx\|^2_{H(\bar\bx)} + \frac{1}{6n^3} \sum\limits_{i=1}^n \partial^3_{abc} F(\bar\bx)(x_i^a-\bar{x}^a)(x_i^b-\bar{x}^b)(x_i^c-\bar{x}^c) + \EE_{\tilde\bx}O(\|\tilde\bx-\bar\bx\|^4). 
\end{equation*}
By the Taylor expansion~\eqref{eqn:taylor2} for $H$, for any vector $\by\in\RR^d$, we have
\begin{equation*}
\|\by\|_{H(\bar\bx)}^2 = \|\by\|_{H(\bx^*)}^2 + \partial^3_{abc}F(\bx^*)y^ay^b(\bar{x}^c-(x^*)^c) + O(\|\by\|^2\|\bar{\bx}-\bx^*\|^2).
\end{equation*}
Therefore, for $\bar{c}$ we have
\begin{align}
-\bar{c} &= \frac{1}{2n^2} \sum\limits_{i=1}^n \|\bx_i-\bar\bx\|^2_{H(\bx^*)} + \frac{1}{2n^2}\sum\limits_{i=1}^n \partial^3_{abc} F(\bx^*)(x_i^a-\bar{x}^a)(x_i^b-\bar{x}^b)(\bar{x}^c-(x^*)^c) \nonumber\\
 &\ \ \ + O\left(\frac{1}{2n^2}\sum\limits_{i=1}^n \|\bx_i-\bar\bx\|^2_2 \|\bar\bx-\bx^*\|^2\right) + \frac{1}{6n^3} \sum\limits_{i=1}^n \partial^3_{abc} F(\bar\bx)(x_i^a-\bar{x}^a)(x_i^b-\bar{x}^b)(x_i^c-\bar{x}^c) \nonumber\\
 &\ \ \ + \EE_{\tilde\bx}O(\|\tilde\bx-\bar\bx\|^4). \nonumber
\end{align}
Let 
\begin{align*}
c_1 &= -\frac{1}{2n^2}\sum\limits_{i=1}^n \partial^3_{abc} F(\bx^*)(x_i^a-\bar{x}^a)(x_i^b-\bar{x}^b)(\bar{x}^c-(x^*)^c),\quad  c_2=\frac{1}{2n^2}\sum\limits_{i=1}^n \|\bx_i-\bar\bx\|_2^2 \|\bar\bx-\bx^*\|^2, \\
c_3 &=-\frac{1}{6n^3} \sum\limits_{i=1}^n \partial^3_{abc} F(\bar\bx)(x_i^a-\bar{x}^a)(x_i^b-\bar{x}^b)(x_i^c-\bar{x}^c),\quad\quad\ \  c_4=\EE_{\tilde\bx}\|\tilde\bx-\bar\bx\|^4.
\end{align*}
Then, we have 
\begin{equation}\label{eqn:c_decomp}
\bar{c} = -\frac{1}{2n^2} \sum\limits_{i=1}^n \|\bx_i-\bar\bx\|^2_{H(\bx^*)} + c_1 + O(c_2) + c_3 + O(c_4),
\end{equation}
Taking square and expectation for~\eqref{eqn:c_decomp}, there exists a constant $C>0$ such that
\begin{align}
\EE\bar{c}^2 & \leq \EE\frac{1}{4n^4}\left(\sum\limits_{i=1}^n \|\bx_i-\bar\bx\|^2_{H(\bx^*)}\right)^2 + C^2\EE(|c_1|+|c_2|+|c_3|+|c_4|)^2 \nonumber\\
&\ \ \ \ + \EE \frac{C}{2n^2}\left|\sum\limits_{i=1}^n \|\bx_i-\bar\bx\|^2_{H(\bx^*)}\right|(|c_1|+|c_2|+|c_3|+|c_4|) \nonumber\\
& \leq \EE\frac{1}{4n^4}\left(\sum\limits_{i=1}^n \|\bx_i-\bar\bx\|^2_{H(\bx^*)}\right)^2 + C^2\EE(|c_1|+|c_2|+|c_3|+|c_4|)^2 \nonumber\\
&\ \ \ \ + \frac{1}{2\sqrt{n}}\EE\left(\frac{1}{2n^2}\sum\limits_{i=1}^n \|\bx_i-\bar\bx\|^2_{H(\bx^*)}\right)^2 + \frac{\sqrt{n}}{2}C^2\EE(|c_1|+|c_2|+|c_3|+|c_4|)^2. \label{eqn:c2_bound_1}
\end{align}

Next, we give estimates to the terms in~\eqref{eqn:c2_bound_1}. We will use results in Lemma~\ref{lm:moments2}. First, Letting $\lambda_{\max}$ be {the eigenvalue of $H(\bx^*)$ with largest absolute value}, we have 
\begin{align}
\EE \left(\frac{1}{2n^2}\sum\limits_{i=1}^n \|\bx_i-\bar\bx\|^2_{ H(\bx^*)}\right)^2 & \leq \frac{\lambda_{\max}^2}{4n^4} \EE \left(\sum\limits_{i=1}^n \|\bx_i-\bar\bx\|^2\right)^2 = O(\frac{1}{n^2}).\nonumber
\end{align}
Therefore, 
\begin{equation}\label{eqn:c2_bound_est_1}
\frac{1}{2\sqrt{n}}\EE \left(\frac{1}{2n^2}\sum\limits_{i=1}^n \|\bx_i-\bar\bx\|^2_{ H(\bx^*)}\right)^2 = O(\frac{1}{n^{2.5}}).
\end{equation}
Next, for $\EE c_1^2$, we have 
\begin{align}
\EE c_1^2 & \leq \frac{1}{4n^4} \EE \left(\sum\limits_{i=1}^n \partial^3_{abc}F(\bx^*)(x_i^a-\bar{x}^a)(x_i^b-\bar{x}^b)(\bar{x}^c-(x^*)^c)\right)^2 \nonumber\\
 & \leq \frac{d^3\max\limits_{a,b,c} \partial_{abc}^3F(\bx^*)^2}{4n^4}\EE \left(\sum\limits_{i=1}^n \|\bx_i-\bar\bx\|^2\|\bar\bx-\bx^*\|\right)^2 \nonumber\\
& \leq \frac{d^3\max\limits_{a,b,c} \partial_{abc}^3F(\bx^*)^2}{4n^4}\sum\limits_{i,j=1}^n \EE\|\bx_i-\bar\bx\|^2\|\bx_j-\bar\bx\|^2\|\bar\bx-\bx^*\|^2 \nonumber\\
& \leq \frac{d^3\max\limits_{a,b,c} \partial^3_{abc}F(\bx^*)^2}{4n^4} \sum\limits_{i,j=1}^n \left(\EE\|\bx_i-\bar\bx\|^6\right)^{\frac{1}{3}}\left(\EE\|\bx_j-\bar\bx\|^6\right)^{\frac{1}{3}}\left(\EE\|\bar\bx-\bx^*\|^6\right)^{\frac{1}{3}} \nonumber\\
& = \frac{d^3\max\limits_{i,j,k} \partial_{ijk}F(\bx^*)^2}{4n^4}\cdot n^2 O(\frac{1}{n})\nonumber\\
& = O(\frac{1}{n^3}),\label{eqn:c2_bound_est_2}
\end{align}
where the fourth to the fifth line is given by Lemma~\ref{lm:moments2}.
Similarly, we can obtain $\EE c_2^2=O(\frac{1}{n^4})$. And by Lemma~\ref{lm:moments3} and similar derivations, we can obtain $\EE c_4^2=O(\frac{1}{n^4})$. For $c_3$, the Taylor expansion~\ref{eqn:taylor3} gives
\begin{align*}
|c_3| &= \left|\frac{1}{6n^3} \sum\limits_{i=1}^n \partial^3_{abc} F(\bx^*)(x_i^a-\bar{x}^a)(x_i^b-\bar{x}^b)(x_i^c-\bar{x}^c) + \frac{1}{6n^3}O(\sum\limits_{i=1}^n\|\bar\bx-\bx^*\|\|\bx_i-\bar\bx\|^3)\right| \\
& \leq \frac{d^3\max\limits_{a,b,c}\partial^3_{abc}|F(\bx^*)|}{6n^3}\sum\limits_{i=1}^n \|\bx_i-\bar\bx\|^3 + \frac{1}{6n^3}O(\sum\limits_{i=1}^n\|\bar\bx-\bx^*\|\|\bx_i-\bar\bx\|^3).
\end{align*}
Taking square and by Lemma~\ref{lm:moments2} we have $\EE c_3^2=O(\frac{1}{n^4})$. 
Therefore, considering $(|c_1|+|c_2|+|c_3|+|c_4|)^2\leq 4(c_1^2+c_2^2+c_3^2+c_4^2)$, back to~\eqref{eqn:c2_bound_1} we have 
\begin{equation}\label{eqn:c2_bound_2}
\EE \bar{c}^2 \leq \frac{1}{4n^4}\EE\left(\sum\limits_{i=1}^n \|\bx_i-\bar\bx\|^2_{H(\bx^*)}\right)^2 + O(\frac{1}{n^{2.5}}).
\end{equation}

Next, by $\bx_i-\bar\bx=(\bx_i-\bx^*)+(\bx^*-\bar\bx)$ we obtain 
\begin{align*}
\|\bx_i-\bar\bx\|^2_{H(\bx^*)} = \|\bx_i-\bx^*\|^2_{H(\bx^*)} + 2 (\bx_i-\bx^*)^TH(\bx^*)(\bar\bx-\bx^*) +\|\bar\bx-\bx^*\|^2_{H(\bx^*)}.
\end{align*}
Note that the terms with powers of $\|\bar\bx-\bx^*\|$ become small quantities (no bigger than $O(\frac{1}{\sqrt{n}})$) after taking expectation, we have 
\begin{equation*}
\frac{1}{4n^4}\EE\left(\sum\limits_{i=1}^n \|\bx_i-\bar\bx\|^2_{H(\bx^*)}\right)^2 = \frac{1}{4n^4}\EE\left(\sum\limits_{i=1}^n \|\bx_i-\bx^*\|^2_{H(\bx^*)}\right)^2 + O(\frac{1}{n^{2.5}}).
\end{equation*}
Thus, 
\begin{equation*}
\EE \bar{c}^2 \leq \frac{1}{4n^4}\EE\left(\sum\limits_{i=1}^n \|\bx_i-\bx^*\|^2_{H(\bx^*)}\right)^2 + O(\frac{1}{n^{2.5}}).
\end{equation*}
Further notice that
\begin{align}
\frac{1}{4n^4}\EE\left(\sum\limits_{i=1}^n \|\bx_i-\bx^*\|^2_{H(\bx^*)}\right)^2 & = \frac{1}{4n^4}\sum\limits_{i,j=1}^n \EE \|\bx_i-\bx^*\|^2_{H(\bx^*)}\|\bx_j-\bx^*\|^2_{H(\bx^*)} \nonumber\\
& = \frac{1}{4n^4}\sum\limits_{i,j=1}^n \EE \|\bx_i-\bx^*\|^2_{H(\bx^*)}\EE\|\bx_j-\bx^*\|^2_{H(\bx^*)} \nonumber\\
& \ \ \ + \frac{1}{4n^4}\sum\limits_{i=1}^n\left(\EE \|\bx_i-\bx^*\|^4_{H(\bx^*)} - \big(\EE \|\bx_i-\bx^*\|^2_{H(\bx^*)}\big)^2\right)\nonumber\\
& = \frac{1}{4n^2}\big(\EE\|\bx-\bx^*\|^2_{H(\bx^*)}\big)^2 + O(\frac{1}{n^{3}}). \nonumber
\end{align}
Finally, since $\EE\|\bx-\bx^*\|^2_{H(\bx^*)}=\tr(M_2H(\bx^*))$, we have the following estimate for $\EE \bar{c}^2$:
\begin{equation}\label{eqn:I_est}
    \EE \bar{c}^2 \leq \frac{\tr(M_2H(\bx^*))^2}{4n^2} + O(\frac{1}{n^{2.5}}) = \frac{\sigma_2^2}{4n^2} + O(\frac{1}{n^{2.5}}).
\end{equation}

\subsubsection{Estimate of $\EE\delta^2$.}\label{sssec:delta_square}
For $\EE\delta^2$, first notice that
\begin{align*}
\EE \delta^2 & = \EE \left(\EE_{\tilde{\bx}} F(\tilde{\bx}) - \frac{1}{K}\sum\limits_{k=1}^K F(\tilde{\bx}_k)\right)^2 = \EE \textrm{Var}_{\tilde\bx}\left(\frac{1}{K}\sum\limits_{k=1}^K F(\tilde{\bx}_k)\right) = \frac{1}{K}\EE \textrm{Var}_{\tilde\bx}(F(\tilde\bx)) \\
& = \frac{1}{K}\EE\EE_{\tilde\bx} \big(F(\tilde{\bx})-\EE_{\tilde{\bx}}F(\tilde{\bx})\big)^2 \leq \frac{1}{K}\EE\EE_{\tilde\bx} \big(F(\tilde{\bx})-F(\bar\bx)\big)^2 = \frac{1}{K}\EE (F(\tilde{\bx})-F(\bar\bx))^2.
\end{align*}
Hence, taking the Taylor expansion for $F(\tilde\bx)$ at $\bar\bx$, the last term above gives
\begin{align}
\EE \delta^2 & \leq \frac{1}{C_Kn}\EE\left(\nabla F(\bar\bx)^T(\tilde{\bx}-\bar\bx) + \frac{1}{2}\|\tilde{\bx}-\bar{\bx}\|^2_{H(\bar\bx)}\right. \nonumber\\
  &\ \ \ + \left.\frac{1}{6}\partial^3_{abc}F(\bar\bx)(\tilde{x}^a-\bar{x}^a)(\tilde{x}^b-\bar{x}^b)(\tilde{x}^c-\bar{x}^c) + O(\|\tilde\bx-\bar{\bx}\|^4) \right)^2 \nonumber\\
 & = \frac{1}{C_K n} \EE\left( \nabla F(\bx^*)^T(\tilde{\bx}-\bar\bx) + O(\|\bar\bx-\bx^*\|\|\tilde\bx-\bar\bx\|) + O(\|\bar\bx-\bx^*\|^2\|\tilde\bx-\bar\bx\|) \right. \nonumber\\ 
 &\ \ \ + O(\|\bar\bx-\bx^*\|^3\|\tilde\bx-\bar\bx\|) 
  + \frac{1}{2}\|\tilde{\bx}-\bar{\bx}\|^2_{H(\bx^*)} + O(\|\bar\bx-\bx^*\|\|\tilde\bx-\bar\bx\|^2) \nonumber\\ 
 &\ \ \  \left. +O(\|\bar\bx-\bx^*\|^2\|\tilde\bx-\bar\bx\|^2) +\frac{1}{6}\nabla^3F(\bx^*)[\tilde\bx-\bar\bx] + O(\|\bar\bx-\bx^*\|\|\tilde\bx-\bar\bx\|^3) + O(\|\tilde\bx-\bar\bx\|^4) \right)^2 \nonumber\\
 & = \frac{1}{C_Kn}\left(\EE \big(\nabla F(\bx^*)^T(\tilde\bx-\bar\bx)\big)^2 + \hot(\|\tilde\bx-\bar\bx\|,\|\bar\bx-\bx^*\|;3)\right).
\label{eqn:delta_est_1}
\end{align}
where the second equality follows from the Taylor expansions of $\nabla F(\bar\bx)$, $H(\bar\bx)$, and $\nabla^3F(\bar\bx)$ at $\bx^*$. 
The highest-order terms in $\hot(\|\tilde\bx-\bar\bx\|,\|\bar\bx-\bx^*\|;3)$ have order $8$. Therefore, by Lemma~\ref{lm:moments2} and~\ref{lm:moments3}, we have $\EE \hot(\|\tilde\bx-\bar\bx\|,\|\bar\bx-\bx^*\|;3) = O(\frac{1}{n^{1.5}})$, and thus
\begin{align*}
\EE\delta^2 \leq \frac{1}{C_Kn}\EE \big(\nabla F(\bx^*)^T(\tilde\bx-\bar\bx)\big)^2 + O(\frac{1}{n^{2.5}}) = \frac{1}{C_Kn}\EE \|\tilde\bx-\bar\bx\|^2_{\nabla F(\bx^*)\nabla F(\bx^*)^T} + O(\frac{1}{n^{2.5}}).
\end{align*}

Next, by Lemma~\ref{lm:quadratic_form}, we have 
\begin{align}
\EE \|\tilde\bx-\bar\bx\|_{\nabla F(\bx^*)\nabla F(\bx^*)^T}^2 & =\frac{1}{n^2}\EE \sum\limits_{i=1}^n\|\bx_i-\bar\bx\|_{\nabla F(\bx^*)\nabla F(\bx^*)^T}^2 \nonumber\\
& = \frac{1}{n} \EE \|\bx_1-\bar\bx\|_{\nabla F(\bx^*)\nabla F(\bx^*)^T}^2 \nonumber\\
&= \frac{1}{n}\EE \|\bx_1-\bx^*\|_{\nabla F(\bx^*)\nabla F(\bx^*)^T}^2 + O(\frac{1}{n^{1.5}}). \label{eqn:used_in_next_thm}
\end{align}
Hence, $\EE\delta^2$ can be bounded by:
\begin{equation}\label{eqn:delta_est}
\EE\delta^2 \leq \frac{1}{C_Kn^2}\EE \big(\nabla F(\bx^*)^T(\bx-\bx^*)\big)^2 + O(\frac{1}{n^{2.5}}) = \frac{\sigma_1}{C_Kn^2} + O(\frac{1}{n^{2.5}}).
\end{equation}

\subsubsection{Estimate of $\EE (F(\bar\bx) - F(\bx^*))\bar{c}$}\label{sssec:cbar_Fdiff_shift}
Taking Taylor expansion for $F$, we have 
\begin{align}
\EE(F(\bar\bx)-F(\bx^*))\bar{c} &= \EE \big(\nabla F(\bx^*)^T(\bar\bx-\bx^*) + \frac{1}{2}\|\bar\bx-\bx^*\|^2_{H(\bx^*)} \nonumber\\
&\ \ \ + \frac{1}{6}\partial^3_{abc} F(\bx^*)(\bar{x}^a-(x^*)^a)(\bar{x}^b-(x^*)^b)(\bar{x}^c-(x^*)^c) + O(\|\bar\bx-\bx^*\|^4)\big)\bar{c} \nonumber\\
& = \EE\nabla F(\bx^*)^T(\bar\bx-\bx^*)\bar{c} + \EE \frac{1}{2}\|\bar\bx-\bx^*\|^2_{H(\bx^*)}\bar{c} + \EE O(\|\bar\bx-\bx^*\|^3)\bar{c} \label{eqn:J_est_1}
\end{align}
We will estimate the three terms above using the estimate~\eqref{eqn:c_decomp} for $\bar{c}$. 

For the first term on the right hand side of~\eqref{eqn:J_est_1}, there exists a constant $C$, such that 
{\small
\begin{align}
\EE\nabla F(\bx^*)^T(\bar\bx-\bx^*)\bar{c} & = \EE\nabla F(\bx^*)^T(\bar\bx-\bx^*)\left(-\frac{1}{2n^2}\sum_{i=1}^n \|\bx_i-\bar\bx\|^2_{H(\bx^*)}+c_1+O(c_2)+c_3+O(c_4)\right)\nonumber\\
& \leq -\frac{1}{2n^2}\EE\nabla F(\bx^*)^T(\bar\bx-\bx^*)\sum_{i=1}^n \|\bx_i-\bar\bx\|^2_{H(\bx^*)} + \EE\nabla F(\bx^*)^T(\bar\bx-\bx^*)c_1 \nonumber\\
&\ \ + C\|\nabla F(\bx^*)\| \EE \|\bar\bx-\bx^*\|(|c_2|+|c_3|+|c_4|). \label{eqn:J_est_J1_1}
\end{align}}

We first consider the first term on the right hand side of~\eqref{eqn:J_est_J1_1}. For convenience, let $\bz_i=\bx_i-\bx^*$, and $\bar\bz=\bar\bx-\bx^*$. Then, we have $\EE\bz_i=0$, and
\begin{equation*}
\EE\nabla F(\bx^*)^T(\bar\bx-\bx^*)\sum\limits_{i=1}^n (\bx_i-\bar\bx)^T H(\bx^*)(\bx_i-\bar\bx) = \EE\nabla F(\bx^*)^T\bar\bz\sum\limits_{i=1}^n (\bz_i-\bar\bz)^T H(\bx^*)(\bz_i-\bar\bz).
\end{equation*}
For the right hand side, we have
\begin{align}
& \EE\nabla F(\bx^*)^T\bar\bz\sum\limits_{i=1}^n (\bz_i-\bar\bz)^T H(\bx^*)(\bz_i-\bar\bz) \nonumber\\
= & \EE\nabla F(\bx^*)^T\bar\bz\left(\sum\limits_{i=1}^n\bz_i^TH(\bx^*)\bz_i-n\bar{\bz}^TH(\bx^*)\bar{\bz}\right) \nonumber\\
= &\EE\nabla F(\bx^*)^T\frac{1}{n}\sum\limits_{i=1}^n\bz_i\sum\limits_{i=1}^n\bz_i^TH(\bx^*)\bz_i - n\EE\nabla F(\bx^*)^T\bar\bz\cdot\bar{\bz}^TH(\bx^*)\bar{\bz} \nonumber\\
= & \frac{1}{n}\sum\limits_{i,j=1}^n \EE \nabla F(\bx^*)^T\bz_i\bz_j^TH(\bx^*)\bz_j - n\EE\nabla F(\bx^*)^T\bar\bz\cdot\bar{\bz}^TH(\bx^*)\bar{\bz}. \label{eqn:J1_leading_pf1}
\end{align}
For the first term above, because of the independence of $\bz_1, ..., \bz_n$ and $\EE\bz_i=0$, we have
\begin{align}
\frac{1}{n}\sum\limits_{i,j=1}^n \EE \nabla F(\bx^*)^T\bz_i\bz_j^TH(\bx^*)\bz_j & = \frac{1}{n}\sum\limits_{i=1}^n \EE \nabla F(\bx^*)^T\bz_i\bz_i^TH(\bx^*)\bz_i \nonumber\\
& = \EE \nabla F(\bx^*)^T\bz_1\bz_1^TH(\bx^*)\bz_1 \nonumber \\
& = \sigma_3\label{eqn:J_est_J1_2}
\end{align}
For the second term in~\eqref{eqn:J1_leading_pf1}, by Lemma~\ref{lm:moments2} we have
\begin{align}
\big|n\EE\nabla F(\bx^*)^T\bar\bz\cdot\bar{\bz}^TH(\bx^*)\bar{\bz}\big| & \leq n\|F(\bx^*)\|\lambda_{\max} \EE \|\bar\bx-\bx^*\|^3 = O(\frac{1}{\sqrt{n}}).\label{eqn:J_est_J1_3}
\end{align}
Combining~\eqref{eqn:J_est_J1_2} and~\eqref{eqn:J_est_J1_3} gives 
\begin{align}
-\frac{1}{2n^2}\EE\nabla F(\bx^*)^T(\bar\bx-\bx^*)\sum_{i=1}^n \|\bx_i-\bar\bx\|^2_{H(\bx^*)} 
&\leq -\frac{\sigma_3}{2n^2} + O(\frac{1}{n^{2.5}}).\label{eqn:J_est_J1_4}
\end{align}

Next, we consider the second term on the right hand side of~\eqref{eqn:J_est_J1_1}. Using the same $\bz$ notations as above, we have 
\begin{align}
\EE\nabla F(\bx^*)^T(\bar\bx-\bx^*)c_1 & = -\frac{1}{2n^2}\EE\nabla F(\bx^*)^T\bar\bz\sum\limits_{i=1}^n \partial_{abc}^3 F(\bx^*)(z_i^a-\bar{z}^a)(z_i^b-\bar{z}^b)\bar{z}^c. \nonumber
\end{align}
For the right hand side, we have 
\begin{align}
&-\EE\nabla F(\bx^*)^T\bar\bz\sum\limits_{i=1}^n \partial_{abc}^3 F(\bx^*)(z_i^a-\bar{z}^a)(z_i^b-\bar{z}^b)\bar{z}^c \nonumber\\
= & -\EE \nabla F(\bx^*)^T\bar\bz\sum\limits_{i=1}^n \partial_{abc}^3 F(\bx^*)z_i^a z_i^b\bar{z}^c + n\EE\nabla F(\bx^*)^T\bar\bz\partial_{abc}^3 F(\bx^*)\bar{z}^a\bar{z}^b\bar{z}^c \nonumber\\
= & -\frac{1}{n^2}\sum\limits_{i,j,k=1}^n \EE\nabla F(\bx^*)^T\bz_i \big(\partial^3_{abc} F(\bx^*)z_j^a z_j^b z_k^c\big) + n\EE\hot(\bar\bz;4) \nonumber\\
= & -\frac{1}{n^2}\sum\limits_{i,j,k=1}^n \EE\nabla F(\bx^*)^T\bz_i \big(\partial^3_{abc} F(\bx^*)z_j^a z_j^b z_k^c\big) + O(\frac{1}{n}).
\end{align}
By the independence of $\bz_i$ and $\EE\bz_i=0$, we obtain
\begin{align}
&\frac{1}{n^2}\sum\limits_{i,j,k=1}^n \EE \nabla F(\bx^*)^T\bz_i \big(\partial^3_{abc} F(\bx^*)z_j^a z_j^b z_k^c\big) \nonumber\\
= & \frac{1}{n^2}\sum\limits_{i,j=1}^n \EE \nabla F(\bx^*)^T\bz_i \big(\partial^3_{abc} F(\bx^*)z_j^a z_j^b z_i^c\big) \nonumber\\
= &\frac{1}{n^2}\sum\limits_{i,j=1}^n(\EE \bz_i^T\nabla F(\bx^*)\bz_i)^T (\EE \partial^3_{abc}F(\bx^*) z_j^b z_j^c) \nonumber\\
&+ \frac{1}{n^2}\sum\limits_{i=1}^n \left(\EE \nabla F(\bx^*)^T\bz_i \partial_{abc}^3 F(\bx^*)z_i^a z_i^b z_i^c - (\EE \bz_i^T\nabla F(\bx^*)\bz_i)^T (\EE \partial_{abc}^3 F(\bx^*)z_i^bz_i^c\right) \nonumber\\
= & (\EE \bz^T\nabla F(\bx^*)\bz)^T (\EE\partial_{abc}^3F(\bx^*)z^bz^c) + O(\frac{1}{n}) \nonumber\\
= & (\nabla F(\bx^*)^TM_2)(\partial_{abc}^3F(\bx^*)(M_2)^{bc}) + O(\frac{1}{n}) \nonumber\\
= & \sigma_4 + O(\frac{1}{n}).
\end{align}
Therefore, we have 
\begin{equation}\label{eqn:J_est_J1_c1}
\EE \nabla F(\bx^*)^T(\bar\bx-\bx^*)c_1 = -\frac{\sigma_4}{2n^2} + O(\frac{1}{n^3}).
\end{equation}

For the third term in~\eqref{eqn:J_est_J1_1}, recall that $\EE c_2^2=O(\frac{1}{n^4})$, $\EE c_3^2=O(\frac{1}{n^4})$, and $\EE c_4^2=O(\frac{1}{n^4})$, by H\"{o}lder's inequality we have 
\begin{equation}\label{eqn:J_est_J1_c23}
\EE \|\bar\bx-\bx^*\|(|c_2|+|c_3|+|c_4|) = O(\frac{1}{n^{2.5}}).
\end{equation}

Combining~\eqref{eqn:J_est_J1_4}, \eqref{eqn:J_est_J1_c1}, and~\eqref{eqn:J_est_J1_c23}, we have 
\begin{align}
\EE\nabla F(\bx^*)^T(\bar\bx-\bx^*)\bar{c} &\leq -\frac{\sigma_3+\sigma_4}{2n^2} + O(\frac{1}{n^{2.5}}). \label{eqn:J_est_J1}
\end{align}

For the second term on the right hand side of~\eqref{eqn:J_est_1}, there exists a constant $C$, such that
\begin{align}
\EE \frac{1}{2}\|\bar\bx-\bx^*\|^2_{H(\bx^*)}\bar{c} & \leq -\frac{1}{4n^2}\EE \|\bar\bx-\bx^*\|^2_{H(\bx^*)}\sum_{i=1}^n \|\bx_i-\bar\bx\|^2_{H(\bx^*)} + \EE \left|\frac{1}{2}\|\bar\bx-\bx^*\|^2_{H(\bx^*)}\right|C\sum\limits_{i=1}^4 |c_i| \nonumber\\
 & = -\frac{1}{4n^2}\EE \|\bar\bx-\bx^*\|^2_{H(\bx^*)}\sum_{i=1}^n \|\bx_i-\bar\bx\|^2_{H(\bx^*)} + O(\frac{1}{n^{2.5}}),\label{eqn:J_est_J2_0}
\end{align}
where the estimate of $\EE\|\bar\bx-\bx^*\|_{H(\bx^*)}(|c_1|+|c_2|+|c_3|+|c_4|)$ follows similarly the estimate of\\ $\frac{1}{2n^2}\EE \sum\limits_{i=1}^n\|\bx_i-\bar\bx\|^2_{ H(\bx^*)}(|c_1|+|c_2|+|c_3|+|c_4|)$ when we deal with $\EE \bar{c}^2$. 
For the first term on the right hand side of~\eqref{eqn:J_est_J2_0}, we have
\begin{align*}
&\EE(\bar\bx-\bx^*)^TH(\bx^*)(\bar\bx-\bx^*)\sum_{i=1}^n (\bx_i-\bar\bx)^T H(\bx^*)(\bx_i-\bar\bx) \nonumber\\
= & \EE(\bar\bx-\bx^*)^TH(\bx^*)(\bar\bx-\bx^*)\sum_{i=1}^n (\bx_i-\bx^*)^T H(\bx^*)(\bx_i-\bx^*) + O(\frac{1}{\sqrt{n}}) \nonumber\\
= & \frac{1}{n^2}\sum\limits_{i,j,k=1}^n \EE (\bx_i-\bx^*)^TH(\bx^*)(\bx_j-\bx^*)(\bx_k-\bx^*)^T H(\bx^*)(\bx_k-\bx^*)+ O(\frac{1}{\sqrt{n}}) \nonumber\\
= & \frac{1}{n^2}\sum\limits_{i,k=1}^n \EE (\bx_i-\bx^*)^TH(\bx^*)(\bx_i-\bx^*)(\bx_k-\bx^*)^T H(\bx^*)(\bx_k-\bx^*)+ O(\frac{1}{\sqrt{n}}) \nonumber\\
= & \frac{1}{n^2}\sum\limits_{i,k=1}^n \EE (\bx_i-\bx^*)^TH(\bx^*)(\bx_i-\bx^*)\EE(\bx_k-\bx^*)^T H(\bx^*)(\bx_k-\bx^*)+ O(\frac{1}{\sqrt{n}}) \nonumber\\
= & \sigma_2^2 + O(\frac{1}{\sqrt{n}}).
\end{align*}
Note that from the fourth to the fifth line we are moving terms with $i=k$ into the $O(\frac{1}{\sqrt{n}})$ term, and adding $\frac{1}{n^2}\big(\EE(\bx_i-\bx^*)^TH(\bx^*)(\bx_i-\bx^*)\big)^2$ terms, which are also no bigger than $O(\frac{1}{\sqrt{n}})$. Hence, back to~\eqref{eqn:J_est_J2_0} we have 
\begin{equation}\label{eqn:J_est_J2}
\EE \frac{1}{2}\|\bar\bx-\bx^*\|^2_{H(\bx^*)}\bar{c} \leq -\frac{\sigma_2^2}{4n^2} + O(\frac{1}{n^{2.5}}).
\end{equation}

For the third term on the right hand side of~\eqref{eqn:J_est_1}, recall that $\EE \bar{c}^2 = O(\frac{1}{n^2})$, we have 
\begin{equation}\label{eqn:J_est_J3}
\EE \|\bar\bx-\bx^*\|^3\bar{c} \leq \sqrt{\EE \|\bar\bx-\bx^*\|^6 \EE \bar{c}^2} = O(\frac{1}{n^{2.5}}).
\end{equation}

Finally, combining~\eqref{eqn:J_est_J1}, \eqref{eqn:J_est_J2} and~\eqref{eqn:J_est_J3}, we have 
\begin{align}
\EE(F(\bar\bx)-F(\bx^*))\bar{c} &\leq -\frac{\sigma_2^2}{4n^2} - \frac{\sigma_3+\sigma_4}{2n^2} + O(\frac{1}{n^{2.5}}). \label{eqn:J_est}
\end{align}

\subsubsection{Putting together.}
Finally, putting the estimates~\eqref{eqn:I_est}, \eqref{eqn:delta_est}, \eqref{eqn:J_est} together, we have 
\begin{align}
\EE(\hc^2+2(F(\bar\bx)-F(\bx^*))\hc) &\leq \frac{\sigma_2^2}{4n^2} + \frac{\sigma_1}{C_Kn^2}-\frac{\sigma_2^2}{2n^2} - \frac{\sigma_3}{n^2} - \frac{\sigma_4}{n^2} + O(\frac{1}{n^{2.5}}) \nonumber\\
& = \frac{1}{n^2}\left(-\frac{\sigma_2^2}{4} + \frac{\sigma_1}{C_K}-\sigma_3-\sigma_4\right) + O(\frac{1}{n^{2.5}}). \label{eqn:together}
\end{align}
By the condition~\eqref{eqn:thm_cond}, the coefficient for $\frac{1}{n^2}$ is negative. Hence, when $n$ is sufficiently large, we have 
\begin{equation}
    \EE(\hc^2+2(F(\bar\bx)-F(\bx^*))\hc) < 0.
\end{equation}
This completes the proof.

\subsection{Lemmas}
In this section we provide and prove several lemmas used in the proof above. These lemmas may also be used in the proof of Theorem~\ref{thm:main_scaling}.

The first lemma gives Taylor expansions for $F$ and its derivatives based on Assumption~\ref{assump:F}. The results are standard in calculus, hence we ignore the proof. 
\begin{lemma}\label{lm:taylor}
Let $F:\RR^d\rightarrow\RR$ be a function satisfying Assumption~\ref{assump:F}. Then, for any $\bx,\by\in\RR^d$, we have the following Taylor expansions for $F$ and its derivatives:
\begin{align}
F(\by) &= F(\bx) + \nabla F(\bx)^T(\by-\bx) + \frac{1}{2}(\by-\bx)^T \nabla^2F(\bx)(\by-\bx) \nonumber\\
 &\ \ \ + \frac{1}{6}\partial^3_{abc}F(\bx)(y^a-x^a)(y^b-x^b)(y^c-x^c) + O(\|\by-\bx\|^4),\label{eqn:taylor} \\
\nabla F(\by) &= \nabla F(\bx) + \nabla^2 F(\bx)(\by-\bx) + \frac{1}{2}\partial^3_{abc} F(\bx)(y^b-x^b)(y^c-x^c) + O(\|\by-\bx\|^3),\label{eqn:taylor1} \\
\nabla^2 F(\by) &= \nabla^2 F(\bx) + \partial^3_{abc} F(\bx)(y^c-x^c) + O(\|\by-\bx\|^2), \label{eqn:taylor2} \\
\nabla^3 F(\by) &= \nabla^3 F(\bx) + O(\|\by-\bx\|). \label{eqn:taylor3}
\end{align}
\end{lemma}

The next two lemmas deal with the expectation of quadratic forms and third order monomials over the choice of $\tilde{\bx}$.
\begin{lemma}\label{lm:quadratic_form}
Let $\nu$ be the uniform distribution on $n$ points $\bx_1, \bx_2, ..., \bx_n\in\RR^d$, and $\bar\bx=\frac{1}{n}\sum_{i=1}^n \bx_i = \EE_\nu \bx$. Let $\tilde{\bx}$ be the empirical mean of $m$ samples i.i.d. sampled from $\nu$, i.e. $\tilde{\bx} = \frac{1}{m}\sum_{i=1}^m \tilde{\bx}_i$, $\tilde{\bx}_i\stackrel{iid}{\sim}\nu$. Let $A\in\RR^{d\times d}$ be an arbitrary matrix. Then, 
\begin{equation*}
    \EE_\nu (\tilde{\bx}-\bar\bx)^TA(\tilde{\bx}-\bar\bx) = \frac{1}{nm} \sum\limits_{i=1}^n (\bx_i-\bar\bx)^TA(\bx_i-\bar\bx). 
\end{equation*}
\end{lemma}

\begin{proof}
We ignore the subscript $\nu$ in the expectation. Obviously, we have $\EE\tilde{\bx}=\bar{\bx}$ and $\EE\tilde{\bx}_i = \bar\bx$. Hence,
\begin{align*}
\EE (\tilde{\bx}-\bar\bx)^TA(\tilde{\bx}-\bar\bx) & = \frac{1}{m^2}\sum\limits_{i,j=1}^m \EE(\tilde{\bx}_i-\bar\bx)^TA(\tilde{\bx}_j-\bar\bx) \\
& = \frac{1}{m^2} \left(\sum\limits_{i=1}^m \EE(\tilde{\bx}_i-\bar\bx)^TA(\tilde{\bx}_i-\bar\bx) + \sum\limits_{i\neq j} \EE(\tilde{\bx}_i-\bar\bx)^TA(\tilde{\bx}_j-\bar\bx)\right) \\
& =\frac{1}{m^2} \sum\limits_{i=1}^m \EE(\tilde{\bx}_i-\bar\bx)^TA(\tilde{\bx}_i-\bar\bx) \\
& = \frac{1}{m} \frac{1}{n}\sum\limits_{i=1}^n (\bx_i-\bar\bx)^T A(\bx_i-\bar\bx).
\end{align*}
\end{proof}

\begin{lemma}\label{lm:3rd_form}
Let $\nu$ be the uniform distribution on $n$ points $\bx_1, \bx_2, ..., \bx_n\in\RR^d$, and $\bar\bx=\frac{1}{n}\sum_{i=1}^n \bx_i = \EE_\nu \bx$. Let $\tilde{\bx}$ be the empirical mean of $m$ samples i.i.d. sampled by $\nu$, i.e. $\tilde{\bx} = \frac{1}{m}\sum_{i=1}^m \tilde{\bx}_i$, $\tilde{\bx}_i\stackrel{iid}{\sim}\nu$. Let $T\in\RR^{d\times d\times d}$ be an arbitrary $d\times d\times d$ tensor. Then, 
\begin{equation*}
    \EE_\nu T_{abc}(\tilde{x}^a-\bar{x}^a)(\tilde{x}^b-\bar{x}^b) (\tilde{x}^c-\bar{x}^c)= \frac{1}{m^2n} \sum\limits_{i=1}^n T_{abc}(x_i^a-\bar{x}^a)(x_i^b-\bar{x}^b)(x_i^c-\bar{x}^c). 
\end{equation*}
\end{lemma}

\begin{proof}
The proof is similar to that of Lemma~\ref{lm:quadratic_form}. Still ignore the subscript $\nu$ in the expectation. We have 
{\small
\begin{align*}
\EE T_{abc}(\tilde{x}^a-\bar{x}^a)(\tilde{x}^b-\bar{x}^b) (\tilde{x}^c-\bar{x}^c) & = \EE T_{abc}\left(\frac{1}{m}\sum\limits_{i=1}^m \tilde{x}^a_i-\bar{x}^a\right)\left(\frac{1}{m}\sum\limits_{i=1}^m \tilde{x}^b_i-\bar{x}^b\right)\left(\frac{1}{m}\sum\limits_{i=1}^m \tilde{x}^c_i-\bar{x}^c\right) \\
& = \frac{1}{m^3} \sum\limits_{i,j,k=1}^m \EE T_{abc}(\tilde{x}^a_i-\bar{x}^a)(\tilde{x}^b_j-\bar{x}^b)(\tilde{x}^c_k-\bar{x}^c) \\
& = \frac{1}{m^3} \sum\limits_{i=1}^m \EE T_{abc}(\tilde{x}^a_i-\bar{x}^a)(\tilde{x}^b_i-\bar{x}^b)(\tilde{x}^c_i-\bar{x}^c) \\
& = \frac{1}{m^2}\frac{1}{n}\sum\limits_{i=1}^n T_{abc}(x_i^a-\bar{x}^a)(x_i^b-\bar{x}^b)(x_i^c-\bar{x}^c).
\end{align*}}
\end{proof}

The next lemma deals with the residual term $\EE_{\tilde\bx}\|\tilde\bx-\bar\bx\|^3$ of $\bar{c}$.
\begin{lemma}\label{lm:3rd_order_term}
Let $\nu$ be the uniform distribution on $n$ points $\bx_1, \bx_2, ..., \bx_n\in\RR^d$, and $\bar\bx=\frac{1}{n}\sum_{i=1}^n \bx_i = \EE_\nu \bx$. Let $\tilde{\bx}$ be the empirical mean of $n$ samples i.i.d. sampled by $\nu$, i.e. $\tilde{\bx} = \frac{1}{n}\sum_{i=1}^n \tilde{\bx}_i$, $\tilde{\bx}_i\stackrel{iid}{\sim}\nu$. Then, 
\begin{equation*}
    \EE_\nu \|\tilde\bx-\bar\bx\|^3 \leq \frac{1}{n^3}\left(\sum\limits_{i=1}^n\|\bx_i-\bar\bx\|^2\sum\limits_{i=1}^n\|\bx_i-\bar\bx\|^4 \right)^{\frac{1}{2}} + \frac{\sqrt{3}}{n^3}\left(\sum\limits_{i=1}^n\|\bx_i-\bar\bx\|^2\right)^{\frac{3}{2}}.
\end{equation*}
\end{lemma}

\begin{proof}
We ignore the subscript $\nu$ in the expectation. By Cauchy-Schwartz inequality, we have 
\begin{equation}\label{eqn:lm:3rd:pf:1}
\EE \|\tilde\bx-\bar\bx\|^3 \leq \sqrt{\EE \|\tilde\bx-\bar\bx\|^2\EE \|\tilde\bx-\bar\bx\|^4}.
\end{equation}
For $\EE_\nu \|\tilde\bx-\bar\bx\|^2$, by Lemma~\ref{lm:quadratic_form}, we have
\begin{equation}\label{eqn:lm:3rd:pf:2}
    \EE \|\tilde\bx-\bar\bx\|^2 = \frac{1}{n^2}\sum\limits_{i=1}^n \|\bx_i-\bar\bx\|^2.
\end{equation}
For $\EE_\nu \|\tilde\bx-\bar\bx\|^4$, noting the independency of $\{\tilde{\bx}_i\}$, we expand the expression and drop the zero terms, and obtain
{\small
\begin{align}
\EE_\nu \|\tilde\bx-\bar\bx\|^4 & = \frac{1}{n^4}\sum\limits_{i,j,k,l=1}^n \EE (\tilde{\bx}_i-\bar\bx)^T(\tilde{\bx}_j-\bar\bx)(\tilde{\bx}_k-\bar\bx)^T(\tilde{\bx}_l-\bar\bx) \nonumber\\
& = \frac{1}{n^4}\left(\sum\limits_{i=1}^n \EE\|\tilde{\bx}_i-\bar\bx\|^4 + \sum\limits_{i\neq j}\EE \|\tilde{\bx}_i-\bar\bx\|^2\|\tilde{\bx}_j-\bar\bx\|^2 + 2\sum\limits_{i\neq j}\EE \big((\tilde{\bx}_i-\bar\bx)^T(\tilde{\bx}_j-\bar\bx)\big)^2 \right) \nonumber\\
& \leq \frac{1}{n^3} \EE\|\tilde{\bx}_1-\bar\bx\|^4 + \frac{3}{n^2}\big(\EE \|\tilde{\bx}_1-\bar\bx\|^2\big)^2 \nonumber\\
& = \frac{1}{n^4}\sum\limits_{i=1}^n \|\bx_i-\bar\bx\|^4 + \frac{3}{n^4}\left(\sum\limits_{i=1}^n \|\bx_i-\bar\bx\|^2\right)^2. \label{eqn:lm:3rd:pf:3}
\end{align}}
Substituting~\eqref{eqn:lm:3rd:pf:2} and~\eqref{eqn:lm:3rd:pf:3} into~\eqref{eqn:lm:3rd:pf:1}, and use the result $\sqrt{a+b}\leq\sqrt{a}+\sqrt{b}$ completes the proof. 
\end{proof}

The next two lemmas estimates the moments of $\bx-\bx^*$, $\bx-\bar\bx$ and $\bar\bx-\bx^*$. The first one is a general result for moments of the average of i.i.d random variables.
\begin{lemma}\label{lm:moments1}
Let $X$ be a random variable taking values on $\RR$, and $X_1, X_2, ..., X_n$ are i.i.d copies of $X$. Let $\bar{X}=\frac{1}{n}(X_1+X_2+\cdots+X_n)$. Let $k$ be a positive integer. Assume $\EE X=0$, and $X$ has up to $2k$-th order finite moments. Then, for any $l\leq 2k$, we have $\EE |\bar{X}|^l = O(\frac{1}{n^{l/2}})$.
\end{lemma}

\begin{proof}
We first show the results for even $l$. Consider the set
\begin{equation*}
    \mathcal{S}_l = \left\{(r_1,r_2,...,r_p):\ r_1\geq r_2\geq...\geq r_p>0,\ \sum\limits_{i=1}^p r_i=l\right\}.
\end{equation*}
Let $|\mathcal{S}_l|$ be the number of elements in $\mathcal{S}_l$. Then, $|\mathcal{S}_l|$ is a number that depends only on $l$. When $l$ is an even number, we have
\begin{equation}\label{eqn:lm:moments1:pf:1}
\EE |\bar{X}|^l = \frac{1}{n^l}\EE\big(X_1+...+X_n\big)^l = \frac{1}{n^l}\EE\left(\sum\limits_{(r_1,...,r_p)\in\cS}C_{(r_1,...,r_p)}\sum\limits_{\substack{i_1,...,i_p=1 \\ i_1\neq...\neq i_p}}^n X_{i_1}^{r_1}X_{i_2}^{r_2}\cdots X_{i_p}^{r_p}\right).
\end{equation}
The right hand side of~\eqref{eqn:lm:moments1:pf:1} is the multinomial expansion of $(X_1+...+X_n)^l$, organized according to the powers of different $X_i$'s in each term, and $C_{(r_1,...,r_p)}$ are positive integers related with multinomial coefficients. $C_{(r_1,...,r_p)}$ depends on $l$ and $r_1,...,r_p$, but is independent with $n$. Since $\EE X=0$, the expectation of $X_{i_1}^{r_1}X_{i_2}^{r_2}\cdots X_{i_p}^{r_p}$ is nonzero only when there is no $r_i$ that equals to $1$. Hence, \eqref{eqn:lm:moments1:pf:1} equals to
\begin{equation}\label{eqn:lm:moments1:pf:2}
\frac{1}{n^l}\left(\sum\limits_{\substack{(r_1,...,r_p)\in\cS \\ r_p\geq2}}C_{(r_1,...,r_p)}\sum\limits_{\substack{i_1,...,i_p=1 \\ i_1\neq...\neq i_p}}^n \EE \big(X_{i_1}^{r_1}X_{i_2}^{r_2}\cdots X_{i_p}^{r_p}\big)\right).
\end{equation}
Because $X_{i_1}, ..., X_{i_p}$ are independent, we have $\EE \big(X_{i_1}^{r_1}X_{i_2}^{r_2}\cdots X_{i_p}^{r_p}\big) = \prod_{j=1}^p \EE X_{i_j}^{r_j}$. By the finite moment assumption of $X$, we can find a constant $C$ such that $|\EE X^j|\leq C$ holds for any $1\leq j\leq 2k$. Then, we have 
{\small
\begin{align*}
\frac{1}{n^l}\left(\sum\limits_{\substack{(r_1,...,r_p)\in\cS \\ r_p\geq2}}C_{(r_1,...,r_p)}\sum\limits_{\substack{i_1,...,i_p=1 \\ i_1\neq...\neq i_p}}^n \EE \big(X_{i_1}^{r_1}X_{i_2}^{r_2}\cdots X_{i_p}^{r_p}\big)\right) &\leq \frac{1}{n^l}\left(\sum\limits_{\substack{(r_1,...,r_p)\in\cS \\ r_p\geq2}}C_{(r_1,...,r_p)}\sum\limits_{\substack{i_1,...,i_p=1 \\ i_1\neq...\neq i_p}}^n C^p\right).
\end{align*}}
Moreover, since $r_1\geq r_2\geq...\geq r_p\geq2$, we have $p\leq l/2$, and hence
\begin{equation*}
\frac{1}{n^l}\sum\limits_{\substack{(r_1,...,r_p)\in\cS \\ r_p\geq2}}C_{(r_1,...,r_p)}\sum\limits_{\substack{i_1,...,i_p=1 \\ i_1\neq...\neq i_p}}^n C^l \leq \frac{1}{n^l}\sum\limits_{\substack{(r_1,...,r_p)\in\cS \\ r_p\geq2}}C_{(r_1,...,r_p)}n^{\frac{l}{2}}C^l
\leq \frac{|S_l|\max C_{(r_1,...,r_p)}C^l}{n^{l/2}}.
\end{equation*}
This completes the proof for even $l$. 

For odd $l$, note that $2k\geq l+1$. The result follows by a Cauchy-Schwartz inequality:
\begin{equation*}
    \EE|\bar{X}|^l \leq \sqrt{\EE|\bar{X}|^{l+1}\EE|\bar{X}|^{l-1}}.
\end{equation*}
\end{proof}

\begin{lemma}\label{lm:moments2}
Let $\mu$ be a probability distribution on $\RR^d$ that satisfies Assumption~\ref{assump:mu}, and $\bx, \bx_1, \bx_2, ..., \bx_n$ be i.i.d samples from $\mu$. Let $\bx^*=\EE_\mu \bx$, and $\bar\bx=\frac{1}{n}\sum_{i=1}^n \bx_i$. Then, for any positive integer $k\leq8$, we have
\begin{equation*}
\EE\|\bx-\bx^*\|^k=O(1),\quad \EE \|\bx_i-\bar\bx\|^k=O(1),\quad \EE\|\bar\bx-\bx^*\|^k=O(\frac{1}{n^{k/2}}).
\end{equation*}
\end{lemma}

\begin{proof}
The first result is given by Assumption~\ref{assump:mu}. The third result is a direct corollary of Lemma~\ref{lm:moments1}. The second results follows the triangle inequality:
\begin{equation*}
\|\bx_i-\bar\bx\|^k \leq \big(\|\bx-\bx^*\|+\|\bar\bx-\bx^*\|\big)^k \leq 2^k\big(\|\bx-\bx^*\|^k+\|\bar\bx-\bx^*\|^K\big).
\end{equation*}
\end{proof}

Finally, we study $\EE\|\tilde\bx-\bar{\bx}\|$ using similar approaches.
\begin{lemma}\label{lm:moments3}
Let $\mu$ be a probability distribution on $\RR^d$ that satisfies Assumption~\ref{assump:mu}, and $\nu$ is the uniform ditribution on $[n]$. Consider $\bx, \bx_1, \bx_2, ..., \bx_n$ i.i.d sampled from $\mu$, and $k_1,...,k_n$ i.i.d sampled from $\nu$. Let $\bar\bx=\frac{1}{n}\sum_{i=1}^n \bx_i$, and $\tilde\bx=\frac{1}{n}\sum_{i=1}^n \bx_{k_i}$. Then, for any positive integer $k\leq8$, we have
\begin{equation*}
    \EE \|\tilde\bx-\bar{\bx}\|^k = O(\frac{1}{n^{k/2}}),
\end{equation*}
where the expectation is taken on both $\mu$ and $\nu$.
\end{lemma}

\begin{proof}
Following the proof of Lemma~\ref{lm:moments1}, fixing $\bx_1,...,\bx_n$, for even number $l\leq8$ we have
{\small
\begin{align}
\EE_\nu \|\tilde\bx-\bar\bx\|^l & \leq \frac{1}{n^l}\left(\sum\limits_{\substack{(r_1,...,r_p)\in\cS \\ r_p\geq2}}C_{(r_1,...,r_p)}\sum\limits_{\substack{i_1,...,i_p=1 \\ i_1\neq...\neq i_p}}^n \EE_\nu \|\bx_{k_{i_1}}-\bar\bx\|^{r_1}\EE_\nu \|\bx_{k_{i_2}}-\bar\bx\|^{r_2}\cdots \EE_\nu\|\bx_{k_{i_p}}-\bar\bx\|^{r_p}\right) \nonumber\\
& = \frac{1}{n^l}\left(\sum\limits_{\substack{(r_1,...,r_p)\in\cS \\ r_p\geq2}}C_{(r_1,...,r_p)}\sum\limits_{\substack{i_1,...,i_p=1 \\ i_1\neq...\neq i_p}}^n \prod\limits_{j=1}^p \big(\frac{1}{n}\sum\limits_{i=1}^n \|\bx_i-\bar\bx\|^{r_j}\big) \right). \label{eqn:lm:moments3:pf:1}
\end{align}}
When we consider the sampling of $\bx_1,...,\bx_n$, we need to take an expectation of~\eqref{eqn:lm:moments3:pf:1} over $\{\bx_i\}$. For each term in the sums in~\eqref{eqn:lm:moments3:pf:1}, we have 
\begin{align*}
\EE_\mu \prod\limits_{j=1}^p \big(\frac{1}{n}\sum\limits_{i=1}^n \|\bx_i-\bar\bx\|^{r_j}\big) & = \frac{1}{n^p}\sum_{k_1,...,k_p=1}^n \EE_\mu \prod\limits_{j=1}^p \|\bx_{k_j}-\bar\bx\|^{r_j} \\
  & \leq \frac{1}{n^p}\sum_{k_1,...,k_p=1}^n \prod\limits_{j=1}^p \big(\EE_\mu \|\bx_{k_j}-\bar\bx\|^l\big)^{\frac{r_j}{l}} \\
  & = \frac{1}{n^p}\sum_{k_1,...,k_p=1}^n \prod\limits_{j=1}^p O(1) \\
  & = O(1),
\end{align*}
where the first to the second line is given by the H\"{o}lder's inequality, and the second to the third line is given by Lemma~\ref{lm:moments2}.
Substituting the estimate above to~\eqref{eqn:lm:3rd:pf:1}, we have that  there exists a constant $C$ independent with $n$, such that 
\begin{equation*}
\EE \|\tilde\bx-\bar\bx\|^l \leq \frac{1}{n^l}\left(\sum\limits_{\substack{(r_1,...,r_p)\in\cS \\ r_p\geq2}}C_{(r_1,...,r_p)}\sum\limits_{\substack{i_1,...,i_p=1 \\ i_1\neq...\neq i_p}}^n C \right) = O(\frac{1}{n^{\frac{l}{2}}}).
\end{equation*}
For odd $l$ we can show the result similar to the proof of Lemma~\ref{lm:moments1} using the Cauchy-Schwartz inequality.
\end{proof}

\section{Proof of Theorem~\ref{thm:main_scaling}}\label{sec:pf2}
\subsection{Additional notations}
In this section, we will expand our use of the $\hot$ notations. We use $\hot_{\bar\bx}$ to denote higher-order-terms with coefficients depending on $F(\bar\bx)$ or derivatives of $F$ evaluated at $\bar\bx$. For example,
\begin{equation*}
\big(\nabla F(\bar\bx)^T(\tilde\bx-\bar\bx)\big)^2 + F(\bar\bx)\|\tilde\bx-\bar\bx\|^2_{H(\bar\bx)}
\end{equation*}
can be represented by $\hot_{\bar\bx}(\tilde\bx-\bar\bx; 2)$. By Lemma~\ref{lm:taylor}, $F(\bar\bx)$, $\nabla F(\bar\bx)$, $\nabla^2F(\bar\bx)$ and $\nabla^3F(\bar\bx)$ can be represented by $\hot(\bar\bx-\bx^*;0)$. (Note that in the Taylor expansion the coefficients depends on $\bx^*$, which is treated as a constant.) Hence, any $\hot_{\bar\bx}(f_1,\cdots,f_k;r)$ can be represented by $\hot(f_1,\cdots,f_k,\bar\bx-\bx^*;r)$. Moreover, we use $\hot_{\EE_{\tilde\bx}}$ to denote higher-order-terms that expectations with respect to $\tilde\bx$ are taken for some parts. For example, 
\begin{equation*}
    \nabla F(\bx^*)^T(\bar\bx-\bx^*)\EE_{\tilde\bx} \|\tilde\bx-\bar\bx\|^2_{H(\bx^*)}
\end{equation*}
can be represented by $\hot_{\EE_{\tilde\bx}}(\bar\bx-\bx^*,\tilde\bx-\bar\bx;3)$. Using the H\"{o}lder's inequality and Jensen's inequality, it is easy to show that the $\EE_{\tilde\bx}$ does not change the order of the expectation of the higher-order-terms. For any $r>0$, we have
\begin{equation*}
    \EE \hot_{\EE_{\tilde\bx}}(\bar\bx-\bx^*,\tilde\bx-\bar\bx;r) = O(\frac{1}{n^{r/2}}).
\end{equation*}

\subsection{Main proof}
We take a similar path as the proof of Theorem~\ref{thm:main_shifting}. Specifically, in this section we consider another $\hc$ defined as $\hc=(\hs-1)F(\bar\bx)$, then we have 
\begin{equation*}
\EE (\hs F(\bar\bx)-F(\bx^*))^2 = \EE(F(\bar\bx)+\hc-F(\bx^*))^2,
\end{equation*}
which gives the same form as the shifting debiasing, but with a different debiasing quantity. To prove the theorem, we still show $\EE \hc^2+2\hc(F(\bar\bx)-F(\bx^*))<0$. 

We first decompose $\hc$ to separate the effect of the sampling of $\{\bx_i\}$ and the sampling of $\{\tilde{\bx}_k\}$. Recall that 
\begin{equation*}
    \hs = \frac{F(\bar\bx)\sum\limits_{k=1}^K F(\tilde{\bx}_k)}{\sum\limits_{k=1}^K F(\tilde{\bx}_k)^2} = \frac{F(\bar\bx)\frac{1}{K}\sum\limits_{k=1}^K F(\tilde{\bx}_k)}{\frac{1}{K}\sum\limits_{k=1}^K F(\tilde{\bx}_k)^2}.
\end{equation*}
Define $\bar{s}=\frac{F(\bar\bx)\EE_{\tilde\bx}F(\tilde\bx)}{\EE_{\tilde\bx}F(\tilde\bx)^2}$, $\bar{c}=(\bar{s}-1)F(\bar\bx)$, and
\begin{equation}\label{eqn:s_def}
s_1 = \frac{1}{K}\sum\limits_{k=1}^K F(\tilde{\bx}_k) - \EE_{\tilde\bx}F(\tilde\bx),\quad s_2 = \frac{1}{K}\sum\limits_{k=1}^K F(\tilde{\bx}_k)^2 - \EE_{\tilde\bx}F(\tilde\bx)^2.
\end{equation}
Then, we have 
\begin{align}
\hc & = \bar{c} + (\hs-\bar{s})F(\bar\bx) = \bar{c} + \frac{F(\bar\bx)^2\big(s_1\EE_{\tilde\bx} F(\tilde\bx)^2-s_2\EE_{\tilde\bx} F(\tilde\bx)\big)}{(\EE_{\tilde\bx} F(\tilde\bx)^2+s_2)\EE_{\tilde\bx} F(\tilde\bx)^2} \nonumber\\
& = \bar{c} + \frac{F(\bar\bx)^2\big(s_1\EE_{\tilde\bx} F(\tilde\bx)^2-s_2\EE_{\tilde\bx} F(\tilde\bx)\big)}{(\EE_{\tilde\bx} F(\tilde\bx)^2)^2} - \frac{F(\bar\bx)^2\big(s_1\EE_{\tilde\bx} F(\tilde\bx)^2-s_2\EE_{\tilde\bx} F(\tilde\bx)\big)s_2}{(\EE_{\tilde\bx} F(\tilde\bx)^2+s_2)(\EE_{\tilde\bx} F(\tilde\bx)^2)^2}. \nonumber
\end{align}
Denote 
\begin{equation*}
\delta_1=\frac{F(\bar\bx)^2\big(s_1\EE_{\tilde\bx} F(\tilde\bx)^2-s_2\EE_{\tilde\bx} F(\tilde\bx)\big)}{(\EE_{\tilde\bx} F(\tilde\bx)^2)^2},\quad \delta_2=- \frac{F(\bar\bx)^2\big(s_1\EE_{\tilde\bx} F(\tilde\bx)^2-s_2\EE_{\tilde\bx} F(\tilde\bx)\big)s_2}{(\EE_{\tilde\bx} F(\tilde\bx)^2+s_2)(\EE_{\tilde\bx} F(\tilde\bx)^2)^2}.
\end{equation*}
Then, we can write $\hc=\bar{c}+\delta_1+\delta_2$. Since $\EE_{\tilde\bx}s_1=\EE_{\tilde\bx}s_2=0$, we know $\EE_{\tilde\bx}\delta_1=0$, and hence
{\small
\begin{align}
\EE \hc^2+2\hc(F(\bar\bx)-F(\bx^*)) &= \EE\left( \bar{c}^2+\delta_1^2+\delta_2^2+2\bar{c}\delta_1+2\bar{c}\delta_2+2\delta_1\delta_2\right. \nonumber\\
&+\left. 2\bar{c}(F(\bar\bx)-F(\bx^*) + 2\delta_1(F(\bar\bx)-F(\bx^*)+2\delta_2(F(\bar\bx)-F(\bx^*) \right)\nonumber\\
& = \EE\left(\bar{c}^2+\delta_1^2+\delta_2^2+2\bar{c}\delta_2+2\delta_1\delta_2+2\bar{c}(F(\bar\bx)-F(\bx^*) + 2\delta_2(F(\bar\bx)-F(\bx^*)\right). \label{eqn:scaling_to_est}
\end{align}}

Next, we estimate the terms in~\eqref{eqn:scaling_to_est}.

\subsubsection{Estimate of $\EE\delta_2^2$}\label{sssec:delta2_2}
By Assumption~\ref{assump:pos}, $F$ has a positive lower bound $B$. Hence, we have $\EE_{\tilde\bx}F(\tilde\bx)^2\geq B^2$ and
\begin{equation*}
\EE_{\tilde\bx} F(\tilde\bx)^2+s_2 = \frac{1}{K}\sum\limits_{k=1}^K F(\tilde{\bx}_k)^2 \geq B^2.
\end{equation*}
Substituting the lower bounds into $\delta_2$, we have 
\begin{equation*}
|\delta_2| \leq \frac{F(\bar{\bx})^2}{B^6}\left|\big(s_1\EE_{\tilde\bx} F(\tilde\bx)^2-s_2\EE_{\tilde\bx} F(\tilde\bx)\big)s_2\right|\leq \frac{F(\bar{\bx})^2\EE_{\tilde\bx} F(\tilde\bx)^2}{B^6}|s_1s_2| + \frac{F(\bar{\bx})^2\EE_{\tilde\bx} F(\tilde\bx)}{B^6}|s_2^2|.
\end{equation*}
Therefore, by the H\"{o}lder's inequality and the Jensen's inequality,
\begin{align*}
\EE\delta_2^2 &\leq \frac{2}{B^{12}} \left(\EE [F(\bar{\bx})^4(\EE_{\tilde\bx} F(\tilde\bx)^2)^2s_1^2s_2^2] + \EE [F(\bar{\bx})^4(\EE_{\tilde\bx} F(\tilde\bx))^2s_2^4]\right) \\
& \leq \frac{2}{B^{12}}\left(\big(\EE F(\bar\bx)^{12}\big)^{\frac{1}{3}}\big(\EE F(\tilde\bx)^{12}\big)^{\frac{1}{3}}\big(\EE s_1^{12}\EE s_2^{12}\big)^{\frac{1}{6}} + \big(\EE F(\bar\bx)^{10}\big)^{\frac{2}{5}}\big(\EE F(\tilde\bx)^{10}\big)^{\frac{1}{5}}\big(\EE s_2^{10}\big)^{\frac{2}{5}} \right).
\end{align*}
For the $s_1$, $s_2$ terms, by Lemma~\ref{lm:s}, we have 
\begin{equation*}
\big(\EE s_1^{12}\EE s_2^{12}\big)^{\frac{1}{6}} = O(\frac{1}{K^2n^2})=O(\frac{1}{n^4}),\ \textrm{and}\ \big(\EE s_2^{10}\big)^{\frac{2}{5}}=O(\frac{1}{K^2n^2})=O(\frac{1}{n^4}).
\end{equation*}
For the $\EE F(\bar\bx)^{12}$, $\EE F(\bar\bx)^{10}$, $\EE F(\tilde\bx)^{12}$, and $\EE F(\tilde\bx)^{10}$ terms, by Lemma~\ref{lm:F_moments}, they are all $O(1)$. Therefore, we have the following estimate for $\EE\delta_2^2$:
\begin{equation}\label{eqn:est_delta2_2}
    \EE\delta_2^2=O(\frac{1}{n^4}).
\end{equation}

\subsubsection{Estimate of $\EE\delta_1\delta_2$}
Similar to the estimate of $\delta_2$, we have 
\begin{equation*}
|\delta_1| \leq \frac{F(\bar\bx)^2\EE_{\tilde\bx}F(\tilde\bx)^2}{B^4}|s_1| + \frac{F(\bar\bx)^2\EE_{\tilde\bx}F(\tilde\bx)}{B^4}|s_2|.
\end{equation*}
Therefore, together with the estimate of $\delta_2$, we have 
\begin{align}
|\EE\delta_1\delta_2| & \leq \frac{1}{B^{10}} \left(\EE [F(\bar{\bx})^4(\EE_{\tilde\bx} F(\tilde\bx)^2)^2s_1^2|s_2|] + 2\EE [F(\bar{\bx})^4(\EE_{\tilde\bx} F(\tilde\bx)^2)(\EE_{\tilde\bx} F(\tilde\bx))|s_1|s_2^2] \right. \nonumber\\
&\ \ \ \ \ \ \ \ \  \left.+ \EE [F(\bar{\bx})^4(\EE_{\tilde\bx} F(\tilde\bx))^2|s_2|^3]\right) \nonumber\\
& \leq \frac{1}{B^{10}}\left((\EE F(\bar\bx)^{11})^{\frac{4}{11}}\EE F(\tilde\bx)^{11})^{\frac{4}{11}}(\EE|s_1|^{11})^{\frac{2}{11}}(\EE|s_2|^{11})^{\frac{1}{11}} \right. \nonumber\\
&\ \ \ \ \ \ \ \ \  + 2(\EE F(\bar\bx)^{10})^{\frac{4}{10}}\EE F(\tilde\bx)^{10})^{\frac{3}{10}}(\EE|s_1|^{10})^{\frac{1}{10}}(\EE|s_2|^{10})^{\frac{2}{10}} \nonumber\\
&\ \ \ \ \ \ \ \ \  +\left.(\EE F(\bar\bx)^{9})^{\frac{4}{9}}\EE F(\tilde\bx)^{9})^{\frac{2}{9}}(\EE|s_2|^{9})^{\frac{3}{9}}\right) \nonumber\\
& = O(\frac{1}{n^3}) + O(\frac{1}{n^3}) + O(\frac{1}{n^3}) \nonumber\\
& = O(\frac{1}{n^3}).\label{eqn:est_delta1_delta2}
\end{align}
In the derivations above, we use Lemma~\ref{lm:s} and the fact that $\EE F(\bar\bx)^l=O(1)$, $\EE F(\tilde\bx)=O(1)$ for $l>0$ (Lemma~\ref{lm:F_moments}).

\subsubsection{Estimate of $\EE\delta_2(F(\bar\bx)-F(\bx^*))$}
By the H\"{o}lder's inequality and the result in Section~\ref{sssec:delta2_2}, we have 
\begin{align*}
|\EE\delta_2(F(\bar\bx)-F(\bx^*))| & \leq \sqrt{\EE\delta_2^2 \EE(F(\bar\bx)-F(\bx^*))^2} \\
& \leq O(\frac{1}{n^2})\sqrt{\EE(F(\bar\bx)-F(\bx^*))^2}.
\end{align*}
For $\EE(F(\bar\bx)-F(\bx^*))^2$, by the Taylor expansion and Lemma~\ref{lm:moments2_2} we have 
{\small
\begin{align*}
\EE(F(\bar\bx)-F(\bx^*))^2 & = \EE \left(\nabla F(\bx^*)^T(\bar\bx-\bx^*) + \frac{1}{2}\|\bar\bx-\bx^*\|_{H(\bx^*)} + \frac{1}{6}\nabla^3F(\bx^*)[\bar\bx-\bx^*] + O(\|\bar\bx-\bx^*\|^4)\right)^2 \\
& \leq \EE\big(\hot(\|\bar\bx-\bx^*\|;2)\big) \\
& = O(\frac{1}{n}).
\end{align*}}
Therefore, 
\begin{equation}\label{eqn:delta2_F}
\EE\delta_2(F(\bar\bx)-F(\bx^*)) = O(\frac{1}{n^{2.5}}).
\end{equation}

\subsubsection{Estimate of $\EE\delta_1^2$}
Note that 
\begin{equation*}
\delta_1^2 = \frac{F(\bar\bx)^4}{(\EE_{\tilde\bx} F(\tilde\bx)^2)^4}\big(s_1\EE_{\tilde\bx} F(\tilde\bx)^2-s_2\EE_{\tilde\bx} F(\tilde\bx)\big)^2.
\end{equation*}
We first simplify the problem by noticing that $F(\bar\bx)$ and $F(\tilde\bx)$ are both close to $F(\bx^*)$. Specifically, let 
\begin{equation*}
\alpha := \frac{1}{F(\bx^*)^4}\big(s_1\EE_{\tilde\bx} F(\tilde\bx)^2-s_2\EE_{\tilde\bx} F(\tilde\bx)\big)^2, \qquad \beta=\delta_1^2-\alpha.
\end{equation*}
Then, 
\begin{equation}
\beta = \frac{F(\bar\bx)^4F(\bx^*)^4-(\EE_{\tilde\bx} F(\tilde\bx)^2)^4}{(\EE_{\tilde\bx} F(\tilde\bx)^2)^4 F(\bx^*)^4}\big(s_1\EE_{\tilde\bx} F(\tilde\bx)^2-s_2\EE_{\tilde\bx} F(\tilde\bx)\big)^2.
\end{equation}
By the lower bound of $F$ and the H\"{o}lder's inequality, we have 
\begin{equation}\label{eqn:delta1_2_beta}
\EE |\beta| \leq \frac{1}{B^{12}} \sqrt{\EE\left(F(\bar\bx)^4F(\bx^*)^4-(\EE_{\tilde\bx} F(\tilde\bx)^2)^4\right)^2 \EE\left(s_1\EE_{\tilde\bx} F(\tilde\bx)^2-s_2\EE_{\tilde\bx} F(\tilde\bx)\right)^4}.
\end{equation}
For the term in~\eqref{eqn:delta1_2_beta} with $s_1$ and $s_2$, note that this term contains fourth-order monomials of $s_1$ and $s_2$. Similar to the analysis in Section~\ref{sssec:delta2_2}, we know this term has order $O(\frac{1}{n^4})$. For the other term, by Jensen's inequality we have 
\begin{align*}
\EE\left(F(\bar\bx)^4F(\bx^*)^4-(\EE_{\tilde\bx} F(\tilde\bx)^2)^4\right)^2 & \leq \EE\left(F(\bar\bx)^4F(\bx^*)^4-\EE_{\tilde\bx} F(\tilde\bx)^8\right)^2 \\
& = \EE \left(\EE_{\tilde\bx} \big(F(\bar\bx)^4F(\bx^*)^4-F(\tilde\bx)^8\big)\right)^2 \\
& \leq \EE \left(F(\bar\bx)^4F(\bx^*)^4-F(\tilde\bx)^8\right)^2
\end{align*}
For the term on the last line, a Taylor expansion at $\bx^*$ gives
\begin{align*}
& \EE\left(F(\bx^*)^4\big(F(\bx^*)+\hot(\bar\bx-\bx^*;1)\big)^4 - \big(F(\bx^*) + \hot(\tilde\bx-\bx^*;1)\big)^8\right)^2 \\
= & \EE \big(\hot(\bar\bx-\bx^*;1) + \hot(\tilde\bx-\bx^*;1)\big)^2 \\
\leq & \EE\big(\hot(\|\bar\bx-\bx^*\|, \|\tilde\bx-\bx^*\|; 2)\big)\\
= & O(\frac{1}{n}).
\end{align*}
Therefore, we have 
\begin{equation*}
    \EE\left(F(\bar\bx)^4F(\bx^*)^4-(\EE_{\tilde\bx} F(\tilde\bx)^2)^4\right)^2 = O(\frac{1}{n}),
\end{equation*}
and back to~\eqref{eqn:delta1_2_beta} we have $\EE |\beta| = O(\frac{1}{n^{2.5}})$, and hence 
\begin{equation*}
    \EE\delta_1^2 = \EE\alpha + O(\frac{1}{n^{2.5}}).
\end{equation*}

Next, we study $\EE\alpha$. By the definition of $s_1$ and $s_2$, using the theorem of total expectation, we have 
\begin{align}
\EE\big(s_1\EE_{\tilde\bx} F(\tilde\bx)^2-s_2\EE_{\tilde\bx} F(\tilde\bx)\big)^2 & = \EE \textrm{Var}_{\tilde\bx} \left(\frac{1}{K}\sum\limits_{k=1}^K \big(F(\tilde{\bx}_k)\EE_{\tilde\bx} F(\tilde\bx)^2-F(\tilde{\bx}_k)^2\EE_{\tilde\bx} F(\tilde\bx)\big)\right) \nonumber\\
& = \frac{1}{K}\EE \EE_{\tilde\bx} \left(F(\tilde\bx)\EE_{\tilde\bx} F(\tilde\bx)^2-F(\tilde\bx)^2\EE_{\tilde\bx} F(\tilde\bx)\right)^2\nonumber\\
& = \frac{1}{C_Kn}\EE \left(F(\tilde\bx)\EE_{\tilde\bx} F(\tilde\bx)^2-F(\tilde\bx)^2\EE_{\tilde\bx} F(\tilde\bx)\right)^2.\label{eqn:alpha_1}
\end{align}
Taylor expansion of $F(\tilde\bx)$ at $\bar\bx$ gives
\begin{align*}
F(\tilde\bx) & = F(\bar\bx) + \nabla F(\bar\bx)^T(\tilde\bx-\bar\bx) + \hot_{\bar\bx}(\tilde\bx-\bar\bx; 2), \\
F(\tilde\bx)^2 & = F(\bar\bx)^2 + 2F(\bar\bx)\nabla F(\bar\bx)^T(\tilde\bx-\bar\bx) + \hot_{\bar\bx}(\tilde\bx-\bar\bx; 2).
\end{align*}
Here, the higher-order-terms notation $\hot_{\bar\bx}$ means these terms depends on the derivatives of $F$ at $\bar\bx$. We note that a Taylor expansion of these derivatives at $\bx^*$ will replace the dependency on $\nabla F(\bar\bx)$, $H(\bar\bx)$ and $\nabla F(\bar\bx)$ by powers of $\bar\bx-\bx^*$. It will not lower the order of the terms, Hence, the higher-order-terms can be written as $\hot(\tilde\bx-\bar\bx, \bar\bx-\bx^*;2)$. Plugging the Taylor expansions into~\eqref{eqn:alpha_1}, since $\EE_{\tilde\bx}(\tilde\bx-\bar\bx)=0$, we obtain
\begin{align}
F(\tilde\bx)\EE_{\tilde{\bx}}F(\tilde\bx)^2 & = F(\bar\bx)^3 + F(\bar\bx)^2\nabla F(\bar\bx)^T(\tilde\bx-\bar\bx) + F(\bar\bx)^2\hot(\tilde\bx-\bar\bx, \bar\bx-\bx^*;2) \nonumber\\
&\ \ + \left(F(\bar\bx) + \nabla F(\bar\bx)^T(\tilde\bx-\bar\bx) + \hot_{\bar\bx}(\tilde\bx-\bar\bx, \bar\bx-\bx^*;2)\right)\EE_{\tilde\bx}\hot_{\bar\bx}(\tilde\bx-\bar\bx, \bar\bx-\bx^*;2) \nonumber\\
& = F(\bar\bx)^3 + F(\bar\bx)^2\nabla F(\bar\bx)^T(\tilde\bx-\bar\bx) + \hot_{\EE_{\tilde\bx}}(\tilde\bx-\bar\bx, \bar\bx-\bx^*;2),\label{eqn:alpha_1.1}
\end{align}
where $\hot_{\EE_{\tilde\bx}}$ denotes higher-order-terms in which expectations over $\bar\bx$ are taken for some terms. Similarly, we have 
\begin{equation}\label{eqn:alpha_1.2}
F(\tilde\bx)^2\EE_{\tilde\bx}F(\tilde\bx) = F(\bar\bx)^3 + 2F(\bar\bx)^2\nabla F(\bar\bx)^T(\tilde\bx-\bar\bx) + \hot_{\EE_{\tilde\bx}}(\tilde\bx-\bar\bx, \bar\bx-\bx^*;2).
\end{equation}
Combining~\eqref{eqn:alpha_1.1} and~\eqref{eqn:alpha_1.2}, \eqref{eqn:alpha_1} becomes
\begin{align}
&\frac{1}{C_Kn}\EE\left(F(\bar\bx)^2\nabla F(\bar\bx)^T(\tilde\bx-\bar\bx) + \hot_{\EE_{\tilde\bx}}(\tilde\bx-\bar\bx, \bar\bx-\bx^*;2)\right)^2 \nonumber\\
=& \frac{1}{C_Kn} \EE\left(F(\bar\bx)^2\nabla F(\bar\bx)^T(\tilde\bx-\bar\bx)\right)^2 + \frac{1}{C_Kn}\EE \hot_{\EE_{\tilde\bx}}(\tilde\bx-\bar\bx, \bar\bx-\bx^*;3) \nonumber\\
=& \frac{1}{C_Kn}\EE F(\bar\bx)^4(\nabla F(\bar\bx)^T(\tilde\bx-\bar\bx))^2 + O(\frac{1}{n^{2.5}}).
\end{align}
Expanding $F(\bar\bx)$ and $\nabla F(\bar\bx)$ at $\bx^*$, similar to arguments above Equation~\eqref{eqn:alpha_1.1}, we have 
\begin{align*}
\EE F(\bar\bx)^4\left(\nabla F(\bar\bx)(\tilde\bx-\bar\bx)\right)^2 &= \EE F(\bx^*)^4\left(\nabla F(\bx^*)(\tilde\bx-\bar\bx)\right)^2 + \EE \big(\hot(\tilde\bx-\bar\bx, \bar\bx-\bx^*;3)\big) \\
&= \EE F(\bx^*)^4\left(\nabla F(\bx^*)(\tilde\bx-\bar\bx)\right)^2 + O(\frac{1}{n^{1.5}}).
\end{align*}
Moreover, by the analysis~\eqref{eqn:used_in_next_thm} in Section~\ref{sssec:delta_square}, we have 
\begin{equation*}
\EE\left(\nabla F(\bx^*)(\tilde\bx-\bar\bx)\right)^2 = \frac{1}{n}\nabla F(\bx^*)^TM_2\nabla F(\bx^*) + O(\frac{1}{n^{1.5}}) = \frac{\sigma_1}{n} + O(\frac{1}{n^{1.5}}).
\end{equation*}
Back to~\eqref{eqn:alpha_1}, we have 
\begin{equation*}
\frac{1}{C_Kn}\EE \left(F(\tilde\bx)\EE_{\tilde\bx} F(\tilde\bx)^2-F(\tilde\bx)^2\EE_{\tilde\bx} F(\tilde\bx)\right)^2 = \frac{F(\bx^*)^4\sigma_1}{C_Kn^2} + O(\frac{1}{n^{2.5}}).
\end{equation*}
Hence, for $\alpha$, we have 
\begin{equation*}
    \EE \alpha = \frac{\sigma_1}{C_Kn^2} + O(\frac{1}{n^{2.5}}),
\end{equation*}
and for $\delta_1^2$ we have the same estimate 
\begin{equation}\label{eqn:est_delta1_2}
    \EE \delta_1^2 = \frac{\sigma_1}{C_Kn^2} + O(\frac{1}{n^{2.5}}).
\end{equation}

\subsubsection{Estimate of $\EE \bar{c}^2$}
To estimate $\EE \bar{c}^2$, we first study $\bar{c}$. Recall that \begin{equation*}
\bar{c}=(\bar{s}-1)F(\bar\bx) = \frac{F(\bar\bx)^2\EE F(\tilde\bx)}{\EE F(\tilde\bx)^2} - F(\bar\bx).
\end{equation*}
To simplify notations, we use $F$, $E_1$, $E_2$ to represent $F(\bar\bx)$, $\EE_{\tilde\bx}F(\tilde\bx)$, and $\EE_{\tilde\bx}F(\tilde\bx)^2$ when no confusion is caused. Our intuition is that $E_1$ is close to $F$ and $E_2$ is close to $F^2$. Hence, we will try to expand $F(\tilde\bx)$ and $F(\tilde\bx)^2$ at $\bar\bx$. First, applying the identity $\frac{1}{1+a}=1-a+\frac{a^2}{1+a}$ for $a=\frac{E_2}{F^2}-1$, we have
\begin{align}
\bar{c} &= \frac{E_1}{\frac{E_2}{F^2}} - F \nonumber\\
 & = E_1\left(1 - \big(\frac{E_2}{F^2}-1\big) + \frac{\big(\frac{E_2}{F^2}-1\big)^2}{\frac{E_2}{F^2}}\right) - F \nonumber \\
 & = (E_1-F) + \frac{E_1}{F^2}(F^2-E_2) + \frac{E_1(E_2-F^2)^2}{F^2E_2}. \label{eqn:cbar_1}
\end{align}
We will estimate the three terms in~\eqref{eqn:cbar_1}. We start from the first and the second terms.

For the first term $E_1-F$, expanding $F(\tilde\bx)$ at $\bar\bx$ gives
\begin{equation}
E_1-F = \EE_{\tilde\bx}\frac{1}{2}\|\tilde\bx-\bar\bx\|^2_{H(\bar\bx)} + \EE_{\tilde\bx}\frac{1}{6}\partial_{abc}^3 F(\bar\bx)(\tilde{x}^a-\bar{x}^a)(\tilde{x}^b-\bar{x}^b)(\tilde{x}^c-\bar{x}^c) + \EE_{\tilde\bx}O(\|\tilde\bx-\bar\bx\|^4). \label{eqn:cbar_part1}
\end{equation}

For the second term in~\eqref{eqn:cbar_1}, we first expand $E_1$ and $F^2-E_2$. For $E_1$, similar to~\eqref{eqn:cbar_part1} we have 
\begin{equation}\label{eqn:cbar_part2_E1}
E_1 = F(\bar\bx) + \EE_{\tilde\bx}\frac{1}{2}\|\tilde\bx-\bar\bx\|^2_{H(\bar\bx)} + \EE_{\tilde\bx}\frac{1}{6}\partial_{abc}^3 F(\bar\bx)(\tilde{x}^a-\bar{x}^a)(\tilde{x}^b-\bar{x}^b)(\tilde{x}^c-\bar{x}^c) + \EE_{\tilde\bx}O(\|\tilde\bx-\bar\bx\|^4).
\end{equation}
For $E_2$, we have
{\small
\begin{align}
E_2 &= \EE_{\tilde\bx} \left(F(\bar\bx)+\nabla F(\bar\bx)^T(\tilde\bx-\bar\bx) + \frac{1}{2}\|\tilde\bx-\bar\bx\|^2_{H(\bar\bx)} + \frac{1}{6}\partial_{abc}^3 F(\bar\bx)(\tilde{x}^a-\bar{x}^a)(\tilde{x}^b-\bar{x}^b)(\tilde{x}^c-\bar{x}^c)\right. \nonumber\\
&\ \ \quad\quad + O(\|\tilde\bx-\bar\bx\|^4)\biggr)^2 \nonumber\\
& =F(\bar\bx)^2 + \EE_{\tilde\bx}\|\tilde\bx-\bar\bx\|^2_{(F(\bar\bx)H(\bar\bx)+\nabla F(\bar\bx)\nabla F(\bar\bx)^T)} + \EE_{\tilde\bx}\frac{1}{3}F(\bar\bx)\partial_{abc}^3F(\bar\bx)(\tilde{x}^a-\bar{x}^a)(\tilde{x}^b-\bar{x}^b)(\tilde{x}^c-\bar{x}^c) \nonumber\\
&\ \ \ + \EE_{\tilde\bx}\nabla F(\bar\bx)^T(\tilde\bx-\bar\bx)\|\tilde\bx-\bar\bx\|^2_{H(\bar\bx)} + \EE_{\tilde\bx}\hot_{\bar\bx}(\tilde\bx-\bar\bx;4). \label{eqn:cbar_part2_E2}
\end{align}}
We still replace $\hot_{\bar\bx}(\tilde\bx-\bar\bx;4)$ by $\hot(\tilde\bx-\bar\bx,\bar\bx-\bx^*;4)$. Combining~\eqref{eqn:cbar_part2_E1} and~\eqref{eqn:cbar_part2_E2}, we have
\begin{align}
& E_1(E_2-F^2) \nonumber\\
=& \left(F(\bar\bx) + \EE_{\tilde\bx}\frac{1}{2}\|\tilde\bx-\bar\bx\|^2_{H(\bar\bx)} + \EE_{\tilde\bx}\frac{1}{6}\partial_{abc}^3 F(\bar\bx)(\tilde{x}^a-\bar{x}^a)(\tilde{x}^b-\bar{x}^b)(\tilde{x}^c-\bar{x}^c) + \EE_{\tilde\bx}O(\|\tilde\bx-\bar\bx\|^4)\right) \nonumber\\
&\times\left(\EE_{\tilde\bx}\|\tilde\bx-\bar\bx\|^2_{(F(\bar\bx)H(\bar\bx)+\nabla F(\bar\bx)\nabla F(\bar\bx)^T)} + \EE_{\tilde\bx}\frac{1}{3}F(\bar\bx)\partial_{abc}^3F(\bar\bx)(\tilde{x}^a-\bar{x}^a)(\tilde{x}^b-\bar{x}^b)(\tilde{x}^c-\bar{x}^c)\right. \nonumber\\
&\qquad + \EE_{\tilde\bx}\nabla F(\bar\bx)^T(\tilde\bx-\bar\bx)\|\tilde\bx-\bar\bx\|^2_{H(\bar\bx)} + \EE_{\tilde\bx}\hot(\tilde\bx-\bar\bx,\bar\bx-\bx^*;4)\biggr) \nonumber\\
= & \EE_{\tilde\bx}F(\bar\bx)\|\tilde\bx-\bar\bx\|^2_{(F(\bar\bx)H(\bar\bx)+\nabla F(\bar\bx)\nabla F(\bar\bx)^T)} + \EE_{\tilde\bx}\frac{1}{3}F(\bar\bx)^2\partial_{abc}^3F(\bar\bx)(\tilde{x}^a-\bar{x}^a)(\tilde{x}^b-\bar{x}^b)(\tilde{x}^c-\bar{x}^c) \nonumber\\
&+ \EE_{\tilde\bx}F(\bar\bx)\nabla F(\bar\bx)^T(\tilde\bx-\bar\bx)\|\tilde\bx-\bar\bx\|^2_{H(\bar\bx)} + \hot_{\EE_{\tilde\bx}}(\tilde\bx-\bar\bx,\bar\bx-\bx^*;4). \nonumber
\end{align}
Therefore,
\begin{align}
\frac{E_1}{F^2}(F^2-E_2) & = -\EE_{\tilde\bx}\|\tilde\bx-\bar\bx\|^2_{\left(H(\bar\bx)+\frac{\nabla F(\bar\bx)\nabla F(\bar\bx)^T}{F(\bar\bx)}\right)} - \EE_{\tilde\bx}\frac{1}{3}\partial_{abc}^3F(\bar\bx)(\tilde{x}^a-\bar{x}^a)(\tilde{x}^b-\bar{x}^b)(\tilde{x}^c-\bar{x}^c) \nonumber\\
&\ \ \ - \EE_{\tilde\bx}\frac{\nabla F(\bar\bx)^T}{F(\bar\bx)}(\tilde\bx-\bar\bx)\|\tilde\bx-\bar\bx\|^2_{H(\bar\bx)} + \hot_{\EE_{\tilde\bx}}(\tilde\bx-\bar\bx,\bar\bx-\bx^*;4).
\label{eqn:cbar_part2}
\end{align}

Combining~\eqref{eqn:cbar_part1} and~\eqref{eqn:cbar_part2}, we have \begin{align*}
(E_1-F) + \frac{E_1}{F^2}(F^2-E_2) &= -\frac{1}{2}\EE_{\tilde\bx}\|\tilde\bx-\bar\bx\|^2_{\left(H(\bar\bx)+\frac{2\nabla F(\bar\bx)\nabla F(\bar\bx)^T}{F(\bar\bx)}\right)} \nonumber\\
&\ \ \ - \EE_{\tilde\bx}\frac{1}{6}\partial_{abc}^3F(\bar\bx)(\tilde{x}^a-\bar{x}^a)(\tilde{x}^b-\bar{x}^b)(\tilde{x}^c-\bar{x}^c) \nonumber\\
&\ \ \ - \EE_{\tilde\bx}\frac{\nabla F(\bar\bx)^T}{F(\bar\bx)}(\tilde\bx-\bar\bx)\|\tilde\bx-\bar\bx\|^2_{H(\bar\bx)} + \hot_{\EE_{\tilde\bx}}(\tilde\bx-\bar\bx,\bar\bx-\bx^*;4).
\end{align*}
Let $A(\bx)=H(\bx)+\frac{2\nabla F(\bx)\nabla F(\bx)^T}{F(\bx)}$. By Lemma~\ref{lm:quadratic_form} and~\ref{lm:3rd_form}, we have 
{\small
\begin{align}
(E_1-F) + \frac{E_1}{F^2}(F^2-E_2) &=-\frac{1}{2n^2}\sum\limits_{i=1}^n \|\bx_i-\bar\bx\|^2_{A(\bar\bx)} - \frac{1}{6n^3}\sum\limits_{i=1}^n \partial_{abc}^3F(\bar\bx)(x_i^a-\bar{x}^a)(x_i^b-\bar{x}^b)(x_i^c-\bar{x}^c) \nonumber\\
&\ \ \ - \frac{1}{n^3}\sum\limits_{i=1}^n \frac{\nabla F(\bar\bx)^T}{F(\bar\bx)}(\bx_i-\bar\bx)\|\bx_i-\bar\bx\|^2_{H(\bar\bx)} + \hot_{\EE_{\tilde\bx}}(\tilde\bx-\bar\bx,\bar\bx-\bx^*;4). \label{eqn:cbar_part12}
\end{align}}
Let
\begin{align*}
c_1 = - \frac{1}{6n^3}\sum\limits_{i=1}^n \partial_{abc}^3F(\bar\bx)(x_i^a-\bar{x}^a)(x_i^b-\bar{x}^b)(x_i^c-\bar{x}^c) - \frac{1}{n^3}\sum\limits_{i=1}^n \frac{\nabla F(\bar\bx)^T}{F(\bar\bx)}(\bx_i-\bar\bx)\|\bx_i-\bar\bx\|^2_{H(\bar\bx)},
\end{align*}
$c_2$ be the higher-order-terms in~\eqref{eqn:cbar_part12}, and 
$c_3 = \frac{E_1(E_2-F^2)^2}{F^2E_2}$. 
Then, we have 
\begin{equation*}
\bar{c} = -\frac{1}{2n^2}\sum\limits_{i=1}^n \|\bx_i-\bar\bx\|^2_{A(\bar\bx)} + c_1 + c_2 + c_3.
\end{equation*}
Also let $c_4 = \frac{1}{2n^2}\sum\limits_{i=1}^n \|\bx_i-\bar\bx\|^2_{A(\bx^*)-A(\bar\bx)}$,
then
\begin{equation}
\bar{c} = -\frac{1}{2n^2}\sum\limits_{i=1}^n \|\bx_i-\bar\bx\|^2_{A(\bx^*)} + c_1 + c_2 + c_3 + c_4. \label{eqn:cbar_decomp}
\end{equation}

Next, we come to estimate $\EE \bar{c}^2$ using the decomposition~\eqref{eqn:cbar_decomp}. Similar to~\eqref{eqn:c2_bound_1}, we have 
\begin{align}
\EE \bar{c}^2 &\leq \frac{1}{4n^4}\EE\left(\sum\limits_{i=1}^n \|\bx_i-\bar\bx\|^2_{A(\bx^*)}\right)^2 + \EE(|c_1|+|c_2|+|c_3|+|c_4|)^2 \nonumber\\
&\ \ \ \ + \frac{1}{2\sqrt{n}}\EE\left(\frac{1}{2n^2}\sum\limits_{i=1}^n \|\bx_i-\bar\bx\|^2_{A(\bx^*)}\right)^2 + \frac{\sqrt{n}}{2}\EE(|c_1|+|c_2|+|c_3|+|c_4|)^2.\label{eqn:cbar_decomp2}
\end{align}
Since $A(\bx^*)$ is a fixed matrix, it is easy to show that 
\begin{equation*}
\frac{1}{2\sqrt{n}}\EE\left(\frac{1}{2n^2}\sum\limits_{i=1}^n \|\bx_i-\bar\bx\|^2_{A(\bx^*)}\right)^2 = \frac{1}{2\sqrt{n}}O(\frac{1}{n^2}) = O(\frac{1}{n^{2.5}}).
\end{equation*}
For $c_1$, except the $\frac{1}{n^3}$ factors, it consists the sum of $2n$ terms. All terms and their products are $O(1)$ after taking expectation. (This can be shown rigorously by taking Taylor expansions of $F(\bar\bx)$, $\nabla F(\bar\bx)$, $H(\bar\bx)$ and $\nabla^3 F(\bar\bx)$ at $\bx^*$ and applying Lemma~\ref{lm:moments2_2} and~\ref{lm:moments3_2}.) Therefore, $\EE c_1^2=\frac{1}{n^6}O(n^2)=O(\frac{1}{n^4})$.
For the higher-order-terms $c_2$, by Lemma~\ref{lm:moments2_2} and~\ref{lm:moments3_2}, we easily have $\EE c_2^2=O(\frac{1}{n^4})$.

For $c_3$, by the lower bound of $F$, we have
\begin{equation*}
\EE c_3^2 \leq \frac{1}{B^8}\EE E_1^2(E_2-F^2)^4\leq \frac{1}{B^8}\sqrt{\EE F(\tilde\bx)^4 \EE (E_2-F^2)^8}.
\end{equation*}
By Lemma~\ref{lm:F_moments}, $\EE F(\tilde\bx)^4=O(1)$. By~\eqref{eqn:cbar_part2_E2}, the leading term of $E_2-F^2$ is a second-order term. Hence, the leading term of $(E_2-F^2)^8$ has order $16$, which gives $\EE (E_2-F^2)^8 = O(\frac{1}{n^8})$. Totally, we have $\EE c_3^2=O(\frac{1}{n^4})$. For $c_4$, by a Taylor expansion of $A(\bar\bx)$ at $\bx^*$, we have 
{\small
\begin{align}
A(\bar\bx) - A(\bx^*) &= \partial_{abc}^3F(\bx^*)(\bar{x}^c-(x^*)^c) + \frac{2\big((\bar\bx-\bx^*)^TH(\bx^*)\nabla F(\bx^*)^T + \nabla F(\bx^*)H(\bx^*)(\bar\bx-\bx^*)\big)}{F(\bx^*)} \nonumber\\
 &\ \ \ -\frac{2\nabla F(\bx^*)\nabla F(\bx^*)^T}{F(\bx^*)^2}\nabla F(\bx^*)(\bar\bx-\bx^*) + O(\|\bar\bx-\bx^*\|^2).\label{eqn:taylor_A}
\end{align}}
Note that the leading term of $A(\bar\bx)-A(\bx^*)$ is the $\bar\bx-\bx^*$ term, which gives
\begin{equation*}
    \EE \|A(\bar\bx)-A(\bx^*)\|^2 = O(\frac{1}{n}).
\end{equation*}
Hence, 
\begin{equation*}
\EE c_4^2 \leq \frac{1}{4n^4}\EE \left(\sum\limits_{i=1}^n \|\bx_i-\bx^*\|^2\|A(\bar\bx)-A(\bx^*)\|\right)^2 = \frac{1}{4n^4}\cdot n^2O(\frac{1}{n}) = O(\frac{1}{n^3}).
\end{equation*}

Finally, back to~\eqref{eqn:cbar_decomp2}, we have 
\begin{align}
\EE\bar{c}^2 &\leq \frac{1}{4n^4}\EE\left(\sum\limits_{i=1}^n \|\bx_i-\bar\bx\|^2_{A(\bx^*)}\right)^2 + O(\frac{1}{n^{2.5}})\quad = \frac{1}{4n^2} \left(\EE\|\bx-\bx^*\|^2_{A(\bx^*)}\right)^2 + O(\frac{1}{n^{2.5}}) \nonumber\\
& = \frac{\tr(A(\bx^*)M_2)^2}{4n^2} + O(\frac{1}{n^{2.5}})\quad = \frac{1}{4n^2}\left(\sigma_2+\frac{2\sigma_1}{F(\bx^*)}\right)^2 + O(\frac{1}{n^{2.5}}). \label{eqn:est_cbar2}
\end{align}

\subsubsection{Estimate of $\EE\bar{c}\delta_2$}
The H\"{o}lder's inequality and the estimates~ \eqref{eqn:est_delta2_2} and \eqref{eqn:est_cbar2} give 
\begin{equation}\label{eqn:est_cbar_delta2}
\EE \bar{c}\delta_2 \leq \sqrt{\EE \bar{c}^2 \EE\delta_2^2} = O(\frac{1}{n^3}). 
\end{equation}

\subsubsection{Estimate of $\EE \bar{c}(F(\bar\bx)-F(\bx^*))$}
We take the same technique as in the proof of Theorem~\ref{thm:main_shifting}. Specifically, a Taylor expansion of $F(\bar\bx)$ at $\bx^*$ gives
\begin{align}
\EE \bar{c}(F(\bar\bx)-F(\bx^*)) & = \EE\nabla F(\bx^*)^T(\bar\bx-\bx^*)\bar{c} + \EE \frac{1}{2}\|\bar\bx-\bx^*\|^2_{H(\bx^*)}\bar{c} + \EE O(\|\bar\bx-\bx^*\|^3)\bar{c} \nonumber\\
& = \EE\nabla F(\bx^*)^T(\bar\bx-\bx^*)\bar{c} + \EE \frac{1}{2}\|\bar\bx-\bx^*\|^2_{H(\bx^*)}\bar{c} + O(\frac{1}{n^{2.5}}). \label{eqn:cbar_Fdiff_1}
\end{align}
Recall that when estimating $\EE \bar{c}^2$ we have the decomposition~\eqref{eqn:cbar_decomp} $\bar{c}$:
\begin{equation*}
\bar{c} = -\frac{1}{2n^2}\sum\limits_{i=1}^n \|\bx_i-\bar\bx\|^2_{A(\bx^*)} + c_1 + c_2 + c_3 + c_4,
\end{equation*}
and $\EE c_1^2=O(\frac{1}{n^4})$, $\EE c_2^2=O(\frac{1}{n^4})$, $\EE c_3^2=O(\frac{1}{n^4})$, $\EE c_4^2=O(\frac{1}{n^3})$, we have
{\small
\begin{align}
\EE \nabla F(\bx^*)^T(\bar\bx-\bx^*)\bar{c} & = -\frac{1}{2n^2}\EE \nabla F(\bx^*)^T(\bar\bx-\bx^*)\sum\limits_{i=1}^n \|\bx_i-\bar\bx\|^2_{A(\bx^*)} + \EE \nabla F(\bx^*)^T(\bar\bx-\bx^*)c_4 \nonumber\\
&\ \ \ +\EE \nabla F(\bx^*)^T(\bar\bx-\bx^*)(c_1+c_2+c_3) \nonumber\\
& \leq -\frac{1}{2n^2}\EE \nabla F(\bx^*)^T(\bar\bx-\bx^*)\sum\limits_{i=1}^n \|\bx_i-\bar\bx\|^2_{A(\bx^*)} + \EE \nabla F(\bx^*)^T(\bar\bx-\bx^*)c_4 \nonumber\\
&\ \ \ + \|\nabla F(\bx^*)\|\sqrt{\EE\|\bar\bx-\bx^*\|^2\EE(c_1+c_2+c_3)^2} \nonumber\\
& = -\frac{1}{2n^2}\EE \nabla F(\bx^*)^T(\bar\bx-\bx^*)\sum\limits_{i=1}^n \|\bx_i-\bar\bx\|^2_{A(\bx^*)} + \EE \nabla F(\bx^*)^T(\bar\bx-\bx^*)c_4 + O(\frac{1}{n^{2.5}}), \label{eqn:cbar_Fdiff_J11}
\end{align}}
and
\begin{align}
\EE \frac{1}{2}\|\bar\bx-\bx^*\|^2_{H(\bx^*)}\bar{c} &= -\frac{1}{4n^2}\EE\|\bar\bx-\bx^*\|^2_{H(\bx^*)}\sum\limits_{i=1}^n \|\bx_i-\bar\bx\|^2_{A(\bx^*)} + \frac{1}{2}\EE\|\bar\bx-\bx^*\|^2_{H(\bx^*)}\sum\limits_{i=1}^4 c_i \nonumber\\
& \leq -\frac{1}{4n^2}\EE\|\bar\bx-\bx^*\|^2_{H(\bx^*)}\sum\limits_{i=1}^n \|\bx_i-\bar\bx\|^2_{A(\bx^*)} \nonumber\\
&\ \ \ + \frac{1}{2}\|H(\bx^*)\|\sqrt{\EE\|\bar\bx-\bx^*\|^4\EE(c_1+c_2+c_3+c_4)^2} \nonumber\\
& = -\frac{1}{4n^2}\EE\|\bar\bx-\bx^*\|^2_{H(\bx^*)}\sum\limits_{i=1}^n \|\bx_i-\bar\bx\|^2_{A(\bx^*)} + O(\frac{1}{n^{2.5}}).\label{eqn:cbar_Fdiff_J12}
\end{align}
Substituting~\eqref{eqn:cbar_Fdiff_J11} and~\eqref{eqn:cbar_Fdiff_J12} into~\eqref{eqn:cbar_Fdiff_1}, we have 
\begin{align}
\EE \bar{c}(F(\bar\bx)-F(\bx^*)) & = -\frac{1}{2n^2}\EE \nabla F(\bx^*)^T(\bar\bx-\bx^*)\sum\limits_{i=1}^n \|\bx_i-\bar\bx\|^2_{A(\bx^*)} + \EE \nabla F(\bx^*)^T(\bar\bx-\bx^*)c_4 \nonumber\\
&\ \ \ -\frac{1}{4n^2}\EE\|\bar\bx-\bx^*\|^2_{H(\bx^*)}\sum\limits_{i=1}^n \|\bx_i-\bar\bx\|^2_{A(\bx^*)} + O(\frac{1}{n^{2.5}}). \label{eqn:cbar_Fdiff_2}
\end{align}

For the first and the third terms in~\eqref{eqn:cbar_Fdiff_2}, note that $H(\bx^*)$ and $A(\bx^*)$ are fixed matrices. Following the proof in Section~\ref{sssec:cbar_Fdiff_shift}, we have
\begin{align}
-\frac{1}{2n^2}\EE \nabla F(\bx^*)^T(\bar\bx-\bx^*)\sum\limits_{i=1}^n \|\bx_i-\bar\bx\|^2_{A(\bx^*)} & = -\frac{1}{2n^2}\partial_a F(\bx^*)A(\bx^*)_{bc}(M_3)^{abc} + O(\frac{1}{n^{2.5}}) \nonumber\\
& = -\frac{1}{2n^2}\left(\sigma_3 + \frac{2\sigma'_3}{F(\bx^*)}\right) + O(\frac{1}{n^{2.5}}), \label{eqn:cbar_Fdiff_term1}
\end{align}
and
\begin{align}
-\frac{1}{4n^2}\EE\|\bar\bx-\bx^*\|^2_{H(\bx^*)}\sum\limits_{i=1}^n \|\bx_i-\bar\bx\|^2_{A(\bx^*)} &= -\frac{\tr(M_2H(\bx^*))\tr(M_2A(\bx^*))}{4n^2} + O(\frac{1}{n^{2.5}}) \nonumber\\
= & -\frac{1}{4n^2}\left(\sigma_2^2 + \frac{2\sigma_1\sigma_2}{F(\bx^*)}\right) + O(\frac{1}{n^{2.5}}).\label{eqn:cbar_Fdiff_term3}
\end{align}
For the second term in~\eqref{eqn:cbar_Fdiff_2}, we define a tensor $B(\bx^*)\in\RR^{d\times d\times d}$ as 
\begin{equation*}
B_{abc}(\bx^*) := \partial_{abc}^3 F(\bx^*) + \frac{4\partial_a F(\bx^*)\partial^2_{bc} F(\bx^*)}{F(\bx^*)} - \frac{2\partial_a F(\bx^*)\partial_b F(\bx^*)\partial_c F(\bx^*)}{F(\bx^*)^2}. 
\end{equation*}
Then, by the Taylor expansion for $A$ \eqref{eqn:taylor_A}, we have 
\begin{equation*}
A(\bar\bx) - A(\bx^*) = B_{abc}(\bx^*)(\bar{x}^c-(x^*)^c) + O(\|\bar\bx-\bx^*\|^2). 
\end{equation*}
Therefore, $c_4$ can be written as 
\begin{equation*}
c_4 = -\frac{1}{2n^2}\sum\limits_{i=1}^n B_{abc}(\bx^*)(x_i^a-\bar{x}^a) (x_i^b-\bar{x}^b)(\bar{x}^c-(x^*)^c) + \frac{1}{2n^2}\sum\limits_{i=1}^n\|\bx_i-\bar\bx\|^2O(\|\bar\bx-\bx^*\|^2),
\end{equation*}
and we have 
\begin{equation}
\small
\EE\nabla F(\bx^*)^T(\bar\bx-\bx^*)c_4 = -\frac{1}{2n^2}\nabla F(\bx^*)(\bar\bx-\bx^*)\sum\limits_{i=1}^n B_{abc}(\bx^*)(x_i^a-\bar{x}^a) (x_i^b-\bar{x}^b)(\bar{x}^c-(x^*)^c) + O(\frac{1}{n^{2.5}}).
\end{equation}
Note that $B(\bx^*)$ is a fixed tensor. Repeating the analysis that gives~\eqref{eqn:J_est_J1_c1} with $B(\bx^*)$ replacing $\nabla^3F(\bx^*)$, we have 
\begin{align}
\EE\nabla F(\bx^*)^T(\bar\bx-\bx^*)c_4 &\leq -\frac{1}{2n^2}\partial_a F(\bx^*)(M_2)^{ab}B(\bx^*)_{bcd}(M_2)^{cd} + O(\frac{1}{n^{2.5}}) \nonumber\\
& = -\frac{1}{2n^2}\left(\sigma_4 + \frac{4\sigma_1\sigma_2}{F(\bx^*)}-\frac{2\sigma_1^2}{F(\bx^*)^2}\right) + O(\frac{1}{n^{2.5}}).
\label{eqn:cbar_Fdiff_term2}
\end{align}

Finally, substituting~\eqref{eqn:cbar_Fdiff_term1}, \eqref{eqn:cbar_Fdiff_term3}, and~\eqref{eqn:cbar_Fdiff_term2} into~\eqref{eqn:cbar_Fdiff_2}, we have
\begin{align}
\EE \bar{c}(F(\bar\bx)-F(\bx^*)) &\leq \frac{1}{4n^2}\left(-\sigma_2^2 - 2\sigma_3 - 2\sigma_4 - \frac{4\sigma'_3}{F(\bx^*)} - \frac{10\sigma_1\sigma_2}{F(\bx^*)} + \frac{4\sigma_1^2}{F(\bx^*)^2}\right) + O(\frac{1}{n^{2.5}}).
\label{eqn:est_cbar_Fdiff}
\end{align}

\subsubsection{Putting together}
Finally, collecting the estimates~\eqref{eqn:est_delta2_2}, \eqref{eqn:est_delta1_delta2}, \eqref{eqn:delta2_F}, \eqref{eqn:est_delta1_2}, \eqref{eqn:est_cbar2}, \eqref{eqn:est_cbar_delta2}, \eqref{eqn:est_cbar_Fdiff}, we have 
\begin{equation}
\EE \hc^2+2\hc(F(\bar\bx)-F(\bx^*)) = \frac{1}{n^2}\left(-\frac{\sigma_2^2}{4} + \frac{\sigma_1}{C_K} - \sigma_3 - \sigma_4 - \frac{2\sigma'_3}{F} - \frac{4\sigma_1\sigma_2}{F(\bx^*)} + \frac{3\sigma_1^2}{F(\bx^*)^2}\right) + O(\frac{1}{n^{2.5}}). 
\end{equation}
By the condition~\eqref{eqn:thm_cond_scaling}, the coefficients for $\frac{1}{n^2}$ is negative. Therefore, when $n$ is sufficiently large, we have 
\begin{equation}
    \EE(\hc^2+2(F(\bar\bx)-F(\bx^*))\hc) < 0.
\end{equation}
This completes the proof.

\subsection{Lemmas}
In this subsection we provide lemmas used in the proof of Theorem~\ref{thm:main_scaling}. Note that we will also make use of the lemmas in Section~\ref{sec:pf}. 

The first two lemmas are counterparts of Lemma~\ref{lm:moments2} and~\ref{lm:moments3} under the new assumption~\ref{assump:mu_2}. The proof is similar to the proof of those Lemmas.
\begin{lemma}\label{lm:moments2_2}
Let $\mu$ be a probability distribution on $\RR^d$ that satisfies Assumption~\ref{assump:mu_2}, and $\bx, \bx_1, \bx_2, ..., \bx_n$ are i.i.d sampled from $\mu$. Let $\bx^*=\EE_\mu \bx$, and $\bar\bx=\frac{1}{n}\sum_{i=1}^n \bx_i$. Then, for any positive integer $k$, we have
\begin{equation*}
\EE\|\bx-\bx^*\|^k=O(1),\quad \EE \|\bx_i-\bar\bx\|^k=O(1),\quad \EE\|\bar\bx-\bx^*\|^k=O(\frac{1}{n^{k/2}}).
\end{equation*}
\end{lemma}

\begin{lemma}\label{lm:moments3_2}
Let $\mu$ be a probability distribution on $\RR^d$ that satisfies Assumption~\ref{assump:mu_2}, and $\nu$ is the uniform ditribution on $[n]$. Consider $\bx, \bx_1, \bx_2, ..., \bx_n$ i.i.d sampled from $\mu$, and $k_1,...,k_n$ i.i.d sampled from $\nu$. Let $\bar\bx=\frac{1}{n}\sum_{i=1}^n \bx_i$, and $\tilde\bx=\frac{1}{n}\sum_{i=1}^n \bx_{k_i}$. Then, for any positive integer $k$, we have
\begin{equation*}
    \EE \|\tilde\bx-\bar{\bx}\|^k = O(\frac{1}{n^{k/2}}),
\end{equation*}
where the expectation is taken on both $\mu$ and $\nu$.
\end{lemma}

The next two lemmas characterize the moments of $s_1$, $s_2$, and $F(\bar\bx)$, $F(\tilde\bx)$. The $s_1$ and $s_2$ are defined in~\eqref{eqn:s_def}.
\begin{lemma}\label{lm:s}
Under the same assumptions in Theorem~\ref{thm:main_scaling}, for any positive integer $l$, we have 
\begin{equation*}
    \EE |s_1|^l \leq O(\frac{1}{(Kn)^{\frac{l}{2}}}),\ \textrm{and}\ \ \EE |s_2|^l \leq O(\frac{1}{(Kn)^{\frac{l}{2}}})
\end{equation*}
\end{lemma}

\begin{proof}
First, we consider $s_1$. Without loss of generality, assume $l$ is an even number. By the law of total expectation, we have 
\begin{equation}\label{eqn:lm:s:pf0}
\EE |s_1|^l = \EE\EE_{\tilde\bx}\left|\frac{1}{K}\sum\limits_{k=1}^K F(\tilde{\bx}_k)-\EE_{\tilde{\bx}}F(\tilde\bx)\right|^l.
\end{equation}
For $k=1,2,...,K$, let $Y_k=F(\tilde{\bx}_k)-\EE_{\tilde\bx}F(\tilde\bx)$, and $Y$ be an i.i.d. copy of $Y_k$. Then, $\EE_{\tilde\bx}Y = 0$. Since $\bx_1,\cdots,\bx_n$ are finite, we have $\EE_{\tilde\bx}|Y|^l < \infty$. Let $\hat{Y}=\frac{Y}{\big(\EE_{\tilde\bx}|Y|^l\big)^{1/l}}$ and $\hat{Y}_k=\frac{Y_k}{\big(\EE_{\tilde\bx}|Y|^l\big)^{1/l}}$. Then, by the H\"{o}lder's inequality, for any $r\leq l$, we have $\EE_{\tilde{\bx}} |\hat{Y}|^r\leq 1$. Therefore, applying Lemma~\ref{lm:moments1}, we have 
\begin{equation*}
    \EE_{\tilde\bx}\left|\frac{1}{K}\sum\limits_{k=1}^K \hat{Y}_k\right|^l \leq O(\frac{1}{K^{l/2}})\EE_{\tilde\bx}|\hat{Y}|^l \leq O(\frac{1}{K^{l/2}}),
\end{equation*}
which implies 
\begin{equation*}
\EE_{\tilde\bx}\left|\frac{1}{K}\sum\limits_{k=1}^K Y_k\right|^l \leq O(\frac{1}{K^{l/2}}) \EE_{\tilde\bx}|Y|^l. 
\end{equation*}
Back to~\eqref{eqn:lm:s:pf0}, we have 
\begin{equation}\label{eqn:lm:s:pf1}
\EE|s_1|^l \leq O(\frac{1}{K^{l/2}})\EE\EE_{\tilde\bx}\left|F(\tilde\bx)-\EE_{\tilde{\bx}}F(\tilde\bx)\right|^l.
\end{equation}

By H\"{o}lder's inequality and Jensen's inequality, we have 
\begin{align*}
\left|F(\tilde\bx)-\EE_{\tilde{\bx}}F(\tilde\bx)\right|^l & \leq 2^{l-1} \left(|F(\tilde\bx)-F(\bar\bx)|^l+|\EE_{\tilde\bx}F(\tilde\bx)-F(\bar\bx)|^l\right) \\
& \leq 2^{l-1} \left(|F(\tilde\bx)-F(\bar\bx)|^l+\EE_{\tilde\bx}|F(\tilde\bx)-F(\bar\bx)|^l\right).
\end{align*}
Hence, 
\begin{equation*}
    \EE\EE_{\tilde\bx}\left|F(\tilde\bx)-\EE_{\tilde{\bx}}F(\tilde\bx)\right|^l \leq 2^{l} \EE|F(\tilde\bx)-F(\bar\bx)|^l
\end{equation*}
For $F(\tilde\bx)-F(\bar\bx)$, by Taylor expansion, we have
\begin{align*}
&F(\tilde\bx)-F(\bar\bx) \nonumber\\
= & \nabla F(\bar\bx)^T(\tilde\bx-\bar\bx) + \frac{1}{2}\|\tilde\bx-\bar\bx\|_{H(\bar\bx)} + \frac{1}{6}\partial_{abc}^3F(\bar\bx)(\tilde{x}^a-\bar{x}^a)(\tilde{x}^b-\bar{x}^b)(\tilde{x}^c-\bar{x}^c) + O(\|\tilde\bx-\bar\bx\|^4) \\
= & \nabla F(\bx^*)^T(\tilde\bx-\bar\bx) + (\bar\bx-\bx^*)^TH(\bx^*)(\tilde\bx-\bar\bx) + \frac{1}{2}\partial^3_{abc}F(\bx^*)(\bar{x}^a-(x^*)^a)(\bar{x}^b-(x^*)^b)(\tilde{x}^c-\bar{x}^c) \\
&+ O(\|\bar\bx-\bx^*\|^3\|\tilde\bx-\bar\bx\|) + \frac{1}{2}\|\tilde\bx-\bar\bx\|_{H(\bx^*)} + \frac{1}{2}\partial_{abc}^3F(\bx^*)(\bar{x}^a-(x^*)^a)(\tilde{x}^b-\bar{x}^b)(\tilde{x}^c-\bar{x}^c) \\
&+ O(\|\bar\bx-\bx^*\|^2\|\tilde\bx-\bar\bx\|^2) + \frac{1}{6}\partial_{abc}^3F(\bx^*)(\tilde{x}^a-\bar{x}^a)(\tilde{x}^b-\bar{x}^b)(\tilde{x}^c-\bar{x}^c) + O(\|\bar\bx-\bx^*\|\|\tilde\bx-\bar\bx\|^3) \\
&+ O(\|\tilde\bx-\bar\bx\|^4) \\
= & \hot(\|\bar\bx-\bx^*\|, \|\tilde\bx-\bar\bx\|; 1),
\end{align*}
in which the second equality is obtained by Taylor expansions of $\nabla F(\bar{bx})$, $H(\bar\bx)$, and $\nabla^3 F(\bar\bx)$ at $\bx^*$. 
Therefore, by Lemma~\ref{lm:moments2_2},
\begin{align}
\EE\EE_{\tilde\bx}\left|F(\tilde\bx)-\EE_{\tilde{\bx}}F(\tilde\bx)\right|^l & \leq 2^l\EE\big(\hot(\|\bar\bx-\bx^*\|, \|\tilde\bx-\bar\bx\|; 1)\big)^l \leq O(\frac{1}{n^{l/2}}).\label{eqn:lm:s:pf2}
\end{align}
Substituting~\eqref{eqn:lm:s:pf2} back to~\eqref{eqn:lm:s:pf1} completes the proof for $s_1$. The proof for $s_2$ is similar. 
\end{proof}

\begin{lemma}\label{lm:F_moments}
Under the assumptions of Theorem~\ref{thm:main_scaling}, for any positive integer $l$, we have 
\begin{equation*}
    \EE F(\bar{\bx})^l = O(1),\ \textrm{and}\ \ \EE F(\tilde{\bx})^l = O(1).
\end{equation*}
\end{lemma}

\begin{proof}
For $F(\bar\bx)$, by the Taylor expansion at $\bx^*$, we have 
\begin{equation*}
F(\bar\bx) = F(\bx^*) + \hot(\|\bar\bx-\bx^*\|;1).
\end{equation*}
Therefore, 
\begin{equation*}
F(\bar\bx)^l = (F(\bx^*) + \hot(\|\bar\bx-\bx^*\|;1))^l = F(\bx^*)^l + \hot(\|\bar\bx-\bx^*\|;1).
\end{equation*}
Taking expectation, we have 
\begin{equation*}
\EE F(\bar\bx)^l = F(\bx^*)^l + \EE\big(\hot(\|\bar\bx-\bx^*\|;1)\big) = F(\bx^*)^l + O(\frac{1}{\sqrt{n}}) = O(1).
\end{equation*}

For $F(\tilde\bx)$, still expanding $F$ at $\bx^*$, we have
\begin{equation*}
F(\tilde\bx)^l = F(\bx^*)^l + \hot(\|\tilde\bx-\bx^*\|;1).
\end{equation*}
By Lemma~\ref{lm:moments2_2} and~\ref{lm:moments3}, we have 
\begin{equation*}
\EE \|\tilde\bx-\bx^*\|^l \leq 2^{l-1}\big(\EE\|\tilde\bx-\bar\bx\|^l+\EE\|\bar\bx-\bx^*\|^l\big) = O(\frac{1}{n^{l/2}}). 
\end{equation*}
Hence, we still have 
\begin{equation*}
    \EE\big(\hot(\|\tilde\bx-\bx^*\|;1)\big) = O(\frac{1}{\sqrt{n}}),
\end{equation*}
which completes the proof.
\end{proof}

\bibliography{ref}
\bibliographystyle{plain}

\end{document}